\documentclass[10pt,english]{article}
\usepackage[T1]{fontenc}
\usepackage[latin9]{inputenc}
\usepackage[a4paper]{geometry}
\usepackage{color}
\usepackage{babel}
\usepackage{verbatim}
\usepackage{float}
\usepackage{mathrsfs}
\usepackage{amsmath}
\usepackage{amsthm}
\usepackage{amssymb}
\usepackage{graphicx}
\usepackage{esint}
\usepackage[numbers]{natbib}
\usepackage[unicode=true,pdfusetitle,
 bookmarks=true,bookmarksnumbered=false,bookmarksopen=false,
 breaklinks=false,pdfborder={0 0 0},pdfborderstyle={},backref=false,colorlinks=true]
 {hyperref}
\hypersetup{
 pdfborderstyle={},citecolor=blue,urlcolor=blue}

\makeatletter

\floatstyle{ruled}
\newfloat{algorithm}{tbp}{loa}
\providecommand{\algorithmname}{Algorithm}
\floatname{algorithm}{\protect\algorithmname}

\newcommand{\lyxaddress}[1]{
\par {\raggedright #1
\vspace{1.4em}
\noindent\par}
}
\theoremstyle{plain}
\newtheorem{thm}{\protect\theoremname}
  \theoremstyle{plain}
  \newtheorem{prop}[thm]{\protect\propositionname}
  \theoremstyle{remark}
  \newtheorem{rem}[thm]{\protect\remarkname}
  \theoremstyle{plain}
  \newtheorem{assumption}{\protect\assumptionname}
  \theoremstyle{plain}
  \newtheorem{lem}[thm]{\protect\lemmaname}

\usepackage{babel}
\usepackage{tikz}

\providecommand{\assumptionname}{Assumption}
\providecommand{\lemmaname}{Lemma}
\providecommand{\propositionname}{Proposition}
\providecommand{\remarkname}{Remark}

\providecommand{\theoremname}{Theorem}

\newcounter{hypA}

\newenvironment{condition}{\refstepcounter{hypA}\begin{itemize}
\item[({\bf A\arabic{hypA}})]}{\end{itemize}}

\makeatother

  \providecommand{\assumptionname}{Assumption}
  \providecommand{\lemmaname}{Lemma}
  \providecommand{\propositionname}{Proposition}
  \providecommand{\remarkname}{Remark}
\providecommand{\theoremname}{Theorem}

\begin{document}

\title{Piecewise-Deterministic Markov Chain Monte Carlo}

\author{Paul Vanetti $^{1}$, Alexandre Bouchard-Côté $^{2}$, George Deligiannidis
$^{1}$, Arnaud Doucet $^{1}$}
\maketitle

\lyxaddress{$^{1}$Department of Statistics, University of Oxford, UK.}

\lyxaddress{$^{2}$Department of Statistics, University of British Columbia,
Canada.}
\begin{abstract}
A novel class of non-reversible Markov chain Monte Carlo schemes relying
on continuous-time piecewise-deterministic Markov Processes has recently
emerged. In these algorithms, the state of the Markov process evolves
according to a deterministic dynamics which is modified using a Markov
transition kernel at random event times. These methods enjoy remarkable
features including the ability to update only a subset of the state
components while other components implicitly keep evolving and the
ability to use an unbiased estimate of the gradient of the log-target
while preserving the target as invariant distribution. However, they
also suffer from important limitations. The deterministic dynamics
used so far do not exploit the structure of the target. Moreover,
exact simulation of the event times is feasible for an important yet
restricted class of problems and, even when it is, it is application
specific. This limits the applicability of these techniques and prevents
the development of a generic software implementation of them. We introduce
novel MCMC methods addressing these shortcomings. In particular, we
introduce novel continuous-time algorithms relying on exact Hamiltonian
flows and novel non-reversible discrete-time algorithms which can
exploit complex dynamics such as approximate Hamiltonian dynamics
arising from symplectic integrators while preserving the attractive
features of continuous-time algorithms. We demonstrate the performance
of these schemes on a variety of applications. %
\end{abstract}
{\small{}{}Keywords:} generalized Metropolis\textendash Hastings;
Hamiltonian dynamics; intractable likelihood; non-reversible Markov
chain Monte Carlo; piecewise-deterministic Markov process; weak convergence.

\global\long\def\rd{\text{d}}

\global\long\def\flow{\mathbf{{\Phi}}}

\global\long\def\P{\mathbb{P}}

\section{Introduction\label{section:introduction}}

Markov chain Monte Carlo (MCMC) methods are the tools of choice to
sample non-standard probability distributions. In high-dimensional
scenarios, the celebrated Metropolis\textendash Hastings algorithm
performs usually poorly and alternative algorithms are required. Two
of the most popular alternatives are slice sampling \citep{neal2003}
and Hamiltonian Monte Carlo (HMC) methods \citep{duane1987,neal2011,leimkuhler2015,B_B_L_G_17}
which have had much empirical success over recent years. More recently,
continuous-time non-reversible MCMC algorithms based on Piecewise-Deterministic
Markov Processes (PDMP)\ schemes have also appeared in the literature
in applied probability \citep{M_16,DM_P_16,B_R_16}, automatic control
\citep{M_H_12}, physics \citep{P_dW_12,M_K_K_14,I_K_15,N_H_16},
statistics and machine learning \citep{BC_D_V_15,B_F_R_16,F_B_P_R_16,B_BC_D_D_F_R_J_16,P_G_C_P_17,W_R_17}.
In physics, these schemes have become quickly popular as they provide
state-of-the-art performance when applied to the simulation of large
scale physical models. They also show promise for statistics applications,
in particular for high dimensional sparse graphical models \citep{BC_D_V_15}
and big data \citep{BC_D_V_15,B_F_R_16,G_16,P_G_C_P_17}.

However, the PDMP-based schemes currently available suffer from shortcomings
which limit both their applicability and performance. To ensure invariance
with respect to the target distribution, one needs to be able to simulate
these continuous-time processes exactly. In practice, this restricts
severely the deterministic dynamics one can use: all the existing
algorithms use a simple linear dynamics that does not exploit the
geometry of the target. Moreover, exact simulation of the event times
is problem specific and may be impossible in certain scenarios. This
prevents the development of a generic software implementation of these
techniques.

In this paper, we address these limitations by developing novel continuous-time
and discrete-time Piecewise-Deterministic Markov Chain Monte Carlo
(PD-MCMC) techniques which bring together HMC, PDMP and generalized
Metropolis\textendash Hastings.

First, we show that it is possible to develop continuous-time PD-MCMC
algorithms relying on Hamiltonian dynamics. In this context, exact
simulation of the resulting PDMP remains possible for an important
class of target distributions. The resulting algorithms provide an
alternative to elliptical slice sampling-type algorithms \citep{M_A_2010,B_C_2016}%
. We also exploit a generalized version of Metropolis\textendash Hastings
algorithm (see, e.g., \citep{lelievreroussetstoltz2010}) satisfying
a skewed detailed balance condition to derive novel schemes.

Second, we introduce novel discrete-time PD-MCMC algorithms. These
non-reversible algorithms can be thought of as a discretized version
of continuous-time PD-MCMC but preserve the target distribution as
invariant distribution for all discretization steps. These schemes
are not only able to exploit complex dynamics, such as approximate
Hamiltonian dynamics arising from symplectic integrators, but it is
also always possible to simulate the event times%
. Moreover some versions of these discrete-time algorithms do not
even require being able to compute the gradient of the log-target.
These methods enjoy the same attractive features as their continuous-time
counterparts: they can leverage any representation of the target as
a product of non-negative factors. Additionally they can use unbiased
estimators of the log-target distribution and its gradient and still
provide algorithms with the correct invariant distribution. %

The rest of the paper is organised as follows. In Section \ref{sec:Continuous-Time-Piecewise-Determ}
we review continuous-time PDMPs, provide sufficient conditions to
ensure invariance of a PDMP with respect to a given target distribution,
discuss existing PD-MCMC algorithms and finally introduce novel algorithms
relying on Hamiltonian dynamics. In Section \ref{sec:Discrete-Time-PDMP-PDMCMC},
we introduce the class of discrete-time PDMP and provide sufficient
conditions to ensure invariance of a PDMP with respect to a given
target distribution which parallel the ones obtained in the continuous-time
scenarios. We review existing and describe novel discrete-time PD-MCMC
algorithms. Section \ref{sec:Discrete-time-LocalPDMCMC} is dedicated
to the efficient implementation of discrete-time algorithms using
subsampling and prefetching ideas while Section \ref{sec:Discrete-time-doublystochasticPDMCMC}
proposes discrete-time algorithms to handle scenarios where the target
is intractable but its logarithm and the logarithm of its gradient
can be estimated unbiasedly. Empirical performance of some of these
schemes are reviewed in Section \ref{sec:Numerical-results}. Appendix
\ref{sec:Proofofinvariance} contains all the proofs of validity of
the proposed algorithms while weak convergence of a specific discrete-time
scheme to a PDMP is proven in Appendix \ref{sec:Weak-convergence-of}.

\section{Continuous-Time PDMP and PD-MCMC\label{sec:Continuous-Time-Piecewise-Determ}}

\subsection{PDMP}

PDMPs were introduced in \citep{D_84}. We will only provide here
an informal review of this class of processes in the spirit of \citep{M_H_12,DM_P_16,F_B_P_R_16,B_BC_D_D_F_R_J_16}
and refer the reader to \citep{D_93} for a detailed theoretical treatment.
For the sake of simplicity, assume that $\mathcal{Z}=\mathbb{R}^{n}$.
A $\mathcal{Z}$-valued continuous-time PDMP process $\left\{ z_{t};t\geq0\right\} $
is a càdlàg process involving a deterministic dynamics altered by
random jumps at random event times. It is defined through 
\begin{enumerate}
\item an Ordinary Differential Equation (ODE) with differentiable drift
$\phi:\mathcal{Z}\rightarrow\mathcal{Z}$, i.e.,
\begin{equation}
\frac{\rd z_{t}}{\rd t}=\phi\left(z_{t}\right),\label{eq:ODEdrift}
\end{equation}
which induces a deterministic flow
\begin{eqnarray}
(t,z) & \in & \mathbb{R}^{+}\times\mathcal{Z}\mapsto\Phi_{t}\left(z\right)\in{\cal Z}\label{eq:ODEflow}
\end{eqnarray}
satisfying the semi-group property $\Phi_{s}\circ\Phi_{t}=\Phi_{s+t}$
and such that $t\mapsto\Phi_{t}\left(z\right)$ is càdlàg,
\item an event rate $\lambda:\mathcal{Z}\rightarrow\mathbb{R}^{+}$, with
$\lambda\left(z_{t}\right)\epsilon+o\left(\epsilon\right)$ being
the probability of having an event in the time interval $\left[t,t+\epsilon\right]$,
and
\item a Markov transition kernel $Q$ from $\mathcal{Z}$ to $\mathcal{Z}$
where the state at event time $t$ is given by $z_{t}\sim Q\left(z_{t^{-}},\cdot\right)$,
$z_{t^{-}}$ being the state of the process just before the event.
\end{enumerate}
Algorithm \ref{alg:continuoustimePDMP} describes how to simulate
the path of a PDMP.

\begin{algorithm}[H]
\protect\caption{Simulation of continuous-time PDMP ~\label{alg:continuoustimePDMP}}

\begin{enumerate}
\item Initialize $z_{0}$ arbitrarily on $\mathcal{Z}$ and set $t_{0}\gets0$. 
\item for $k=1,2,\ldots$ do
\begin{enumerate}
\item Sample inter-event time $\tau_{k}$, where $\tau_{k}$ is a non-negative
random variable such that 
\begin{equation}
\mathbb{P}\left(\tau_{k}\geq t\right)=\exp\left[-\int_{r=0}^{t}\lambda\left\{ \Phi_{r}(z_{t_{k-1}})\right\} \rd r\right].\label{eq:SimEventTime}
\end{equation}
\item For $r\in(0,\tau_{k}),~$set 
\begin{equation}
z_{t_{k-1}+r}\gets\Phi_{r}(z_{t_{k-1}}).\label{eq:SolveFlow}
\end{equation}
\item Set $t_{k}\gets t_{k-1}+\tau_{k}$ and sample
\begin{equation}
z_{t_{k}}\sim Q(z_{t_{k}^{-}},\cdot).\label{eq:SimTransition}
\end{equation}
\end{enumerate}
\end{enumerate}
\end{algorithm}
To be able to exactly simulate a PDMP, we thus need to be able to
simulate from the distribution (\ref{eq:SimEventTime}) and compute
the flow (\ref{eq:SolveFlow}). Finally we also need to be able to
simulate from the transition kernel $Q$. In important scenarios,
exact simulation of the event times can be performed using inversion
of the integrated rate function as in \citep{P_dW_12} or using adaptive
thinning procedures as in \citep{BC_D_V_15}. 

We now introduce the generator associated with the PDMP. For functions
in the domain of the generator, it is defined by 
\[
\mathit{\mathcal{L}}f\left(z\right)=\lim_{\epsilon\rightarrow0}\frac{\mathbb{E}\left[\left.f\left(z_{t+\epsilon}\right)\right|z_{t}=z\right]-f\left(z\right)}{\epsilon}.
\]
Under suitable regularity conditions \citep[Theorem 26.14]{D_93},
it can be shown that this generator is given by 
\begin{equation}
\mathit{\mathcal{L}}f\left(z\right)=\left\langle \phi\left(z\right),\nabla f\left(z\right)\right\rangle +\lambda\left(z\right)\int\left[f\left(z'\right)-f\left(z\right)\right]Q\left(z,\rd z'\right),\label{eq:GeneratorPDMP}
\end{equation}
where $\left\langle a,b\right\rangle $ denotes the scalar product
between vectors $a,b$ and $|a|^{2}=\left\langle a,a\right\rangle $.
The first term on the right hand side of (\ref{eq:GeneratorPDMP})
arises from the deterministic dynamics while the second term corresponds
to the jump component of the process.

\subsection{From PDMP to PD-MCMC\label{subsec:Sufficient-conditions-for-invariance}}

Assume we are interested in sampling from a given target probability
distribution on the Borel space $\left(\mathcal{Z},\mathcal{B}\left(\mathcal{Z}\right)\right)$.
If we want to use a PDMP mechanism to sample this target distribution,
this PDMP needs at least to admit this distribution as invariant distribution.
We provide here sufficient conditions to ensure this is satisfied.
If additionally the PDMP is ergodic, this will allow us to estimate
consistently expectations with respect to the invariant distribution.

From now onward, the target distribution will be assumed to have a
strictly positive density $\rho\left(z\right)$ with respect to the
Lebesgue measure $\rd z$ where
\begin{equation}
\rho\left(z\right)=\exp\left(-H\left(z\right)\right).\label{eq:Pho(z)=00003Dexp(-H(z))}
\end{equation}
Invariance with respect to $\rho$ will be satisfied if 
\[
\mathit{\int\rho\left(\rd z\right)\mathcal{L}}f\left(z\right)=0
\]
 for all functions $f$ in the domain of the generator \citep[Proposition 34.7]{D_93}.
From (\ref{eq:GeneratorPDMP}), this means that we need
\[
\int\rho\left(\rd z\right)\left\langle \phi\left(z\right),\nabla f\left(z\right)\right\rangle +\int\rho\left(\rd z\right)\lambda\left(z\right)\int Q\left(z,\rd z'\right)\left[f\left(z'\right)-f\left(z\right)\right]=0.
\]
However, using integration by parts, we obtain %
\begin{align*}
\int\rho\left(\rd z\right)\left\langle \phi\left(z\right),\nabla f\left(z\right)\right\rangle  & =-\int\rho\left(\rd z\right)\left\{ \nabla\cdot\phi\left(z\right)-\left\langle \nabla H\left(z\right),\phi\left(z\right)\right\rangle \right\} f\left(z\right)
\end{align*}
where $\nabla\cdot\phi\left(z\right):=\sum_{i=1}^{n}\partial_{i}\phi_{i}\left(z\right)$
is the divergence of the vector field $\phi$. Hence, a sufficient
condition to ensure invariance of a PDMP with respect to $\rho$ is
to have 
\begin{equation}
\int\rho\left(\rd z\right)\left[\lambda\left(z\right)\int Q\left(z,\rd z'\right)\left[f\left(z'\right)-f\left(z\right)\right]-\left\{ \nabla\cdot\phi\left(z\right)-\left\langle \nabla H\left(z\right),\phi\left(z\right)\right\rangle \right\} f\left(z\right)\right]=0.\label{eq:PDMPinvariant}
\end{equation}

The following notation will prove useful to formulate sufficient conditions
to ensure invariance of a PDMP with respect to $\rho$. Suppose that
we are given a a measure $\nu$ on $\mathcal{Z},\mathcal{B}\left(\mathcal{Z}\right)$
and a measurable mapping $\mathcal{\varGamma}:\mathcal{Z}\mapsto\mathcal{Z}$.
Then the push-forward of the measure $\nu$ under the mapping $\mathcal{\varGamma}$,
often denoted by $\mathcal{\varGamma}_{\ast}\nu\left(\mathrm{d}z\right)$,
is the measure $A\mapsto\nu\left(\mathcal{\varGamma}^{-1}\left(A\right)\right)$
for any $A\in\mathcal{B}\left(\mathcal{Z}\right)$. We will use here
the notation $\nu\left(\mathcal{\varGamma}^{-1}\left(\mathrm{d}z\right)\right)$.
For any measurable $f:\mathcal{Z}\mapsto\mathbb{R}$, the following
identity holds
\[
\int_{\mathcal{Z}}f(z)\mathcal{\varGamma}_{\ast}\nu\left(\mathrm{d}z\right)=\int_{\mathcal{Z}}f\circ\mathcal{\varGamma}(z)\nu\left(\mathrm{d}z\right).
\]

\subsubsection{Sufficient conditions for global methods\label{subsec:Sufficient-conditions-for-Global}}

We provide here useful sufficient conditions on $\phi,$ $\lambda$,
and $Q$ to ensure $\rho$-invariance of the associated PDMP, without
making any structural assumptions on these objects.

\begin{condition}\label{hypA:Gassumptions}Conditions on $\phi,$
$\lambda$, and $Q$
\begin{enumerate}
\item There exists a $\rho$-preserving mapping $\mathcal{S}:\mathcal{Z}\to\mathcal{Z}$;
that is $\mathcal{S}$ is measurable and satisfies $\rho\left(\mathcal{S}^{-1}(\mathrm{d}z)\right)=\rho(\mathrm{d}z)$.
\item The event rate $\lambda$ satisfies 
\begin{equation}
\lambda\left(\mathcal{S}\left(z\right)\right)-\lambda\left(z\right)=\nabla\cdot\phi\left(z\right)-\left\langle \nabla H\left(z\right),\phi\left(z\right)\right\rangle .\label{eq:sufficientconditioninvariance}
\end{equation}
\item The kernel $Q$ satisfies
\begin{equation}
\int\rho\left(\rd z\right)\lambda\left(z\right)Q\left(z,\rd z'\right)=\rho\left(\mathcal{S}^{-1}\left(\rd z'\right)\right)\lambda\left(\mathcal{S}\left(z'\right)\right).\label{eq:Qpseudoinvariance}
\end{equation}
\end{enumerate}
\end{condition}

Based on these assumptions, straightforward calculations show that
the following result holds.
\begin{prop}
\label{prop:Ginvariance}Assume (A\ref{hypA:Gassumptions}). Then
the PDMP admits $\rho$ as invariant distribution.
\end{prop}

\subsubsection{Sufficient conditions for local methods\label{subsec:Sufficient-conditions-for-Local}}

Assume that $H\left(z\right)$ can be decomposed as follows 
\begin{equation}
H\left(z\right)=\sum_{i=1}^{n}H_{i}\left(z\right),\label{eq:Hlocalcsum}
\end{equation}
where potentially each $H_{i}\left(z\right)$ only depends on a subset
of the components of $z$. In this context, like in standard MCMC,
we might be interested in using a transition kernel which is a mixture
of $n$ kernels performing local updates. This can be achieved in
the PDMP framework by introducing an event rate of the form 
\begin{equation}
\lambda\left(z\right)={\displaystyle {\displaystyle \sum_{i=1}^{n}\lambda_{i}\left(z\right)}}\label{eq:Ratelocal}
\end{equation}
and a transition kernel of the form 
\begin{equation}
Q\left(z,\rd z'\right)=\sum_{i=1}^{n}\frac{\lambda_{i}\left(z\right)}{\lambda\left(z\right)}Q_{i}\left(z,\rd z'\right)\label{eq:Kernellocal}
\end{equation}
where $Q_{i}$ are Markov transition kernels. Let us write $\left[n\right]:=\left\{ 1,2,...,n\right\} $.
To simulate the event times of the resulting PDMP, one can associate
a clock to each index $i\in[n]$ and use a priority queue \citep{P_dW_12,M_K_K_14,BC_D_V_15}.
When it is possible to bound $\{\lambda_{i};i\in\left[n\right]\}$
locally in time, more elaborate thinning strategies have been developed
in \citep[Section 3.3.2]{BC_D_V_15} and \citep{K_K_16}.

Based on these structural assumptions on $\lambda$ and $Q$, we can
provide useful sufficient ``local'' conditions on $\phi,$ $\{\lambda_{i}:i\in\left[n\right]\}$
and $\{Q_{i}:i\in[n]\}$ to ensure that invariance of the associated
PDMP with respect to $\rho$ is satisfied. 

\begin{condition}\label{hypA:Lassumptions}Conditions on $\phi,$
$\{\lambda_{i}:i\in\left[n\right]\}$ and $\{Q_{i}:i\in[n]\}$
\begin{enumerate}
\item There exists a $\rho$-preserving mapping $\mathcal{S}:\mathcal{Z}\to\mathcal{Z}$.
\item The event rates $\{\lambda_{i}:i\in\left[n\right]\}$ satisfy
\begin{equation}
\sum_{i=1}^{n}\left\{ \lambda_{i}\left(\mathcal{S}\left(z\right)\right)-\lambda_{i}\left(z\right)\right\} =\nabla\cdot\phi\left(z\right)-\left\langle \nabla H\left(z\right),\phi\left(z\right)\right\rangle .\label{eq:LocalRateBalance}
\end{equation}
\item For all $i\in\left[n\right]$, the transition kernel $Q_{i}$ satisfies
\begin{equation}
\int\rho\left(\rd z\right)\lambda_{i}\left(z\right)Q_{i}\left(z,\rd z'\right)=\rho(\mathcal{S}^{-1}\left(\rd z'\right))\lambda_{i}(\mathcal{S}\left(z'\right)).\label{eq:Qlocalpseudoinvariance}
\end{equation}
\end{enumerate}
\end{condition}

If the functions $\{H_{i}:i\in\left[n\right]\}$ are differentiable
then Assumption A\ref{hypA:Lassumptions}.2 is satisfied for a divergence-free
vector field, i.e.~$\nabla\cdot\phi=0$\emph{, }if for all $i\in\left[n\right]$
\begin{equation}
\lambda_{i}\left(\mathcal{S}\left(z\right)\right)-\lambda_{i}\left(z\right)=-\left\langle \nabla H_{i}\left(z\right),\phi\left(z\right)\right\rangle .\label{eq:LocalLocalRateBalance}
\end{equation}

\begin{prop}
\label{prop:Linvariance}Assume (A\ref{hypA:Lassumptions}). Then
the PDMP admits $\rho$ as invariant distribution.
\end{prop}

\subsubsection{Sufficient conditions for doubly stochastic methods\label{subsec:Sufficient-conditions-for-Stochastic}}

Consider now a slight generalization of the previous scenario where
the target distribution cannot even be evaluated pointwise up to a
normalizing constant but there exists a measure $\mu$ on some measurable
space $\left(\Omega,\mathcal{G}\right)$ and a function $H_{\omega}\left(z\right):\Omega\times\mathcal{Z}\rightarrow\mathbb{R}$
which can be evaluated pointwise up to an additive constant such that 

\begin{equation}
H\left(z\right)=\int H_{\omega}\left(z\right)\thinspace\mu\left(\rd\omega\right).\label{eq:Unbiasedestimate}
\end{equation}
In this context, we consider an event rate of the form 
\begin{equation}
\lambda\left(z\right)=\int\lambda_{\omega}\left(z\right)\mu\left(\rd\omega\right)\label{eq:ratemixture}
\end{equation}
where $\lambda_{\omega}:\Omega\rightarrow\mathbb{R}^{+}$ and a transition
kernel of the form 
\begin{equation}
Q\left(z,\rd z'\right)=\frac{\int\lambda_{\omega}\left(z\right)\mu\left(\rd\omega\right)Q_{\omega}\left(z,\rd z'\right)}{\int\lambda_{\omega}\left(z\right)\mu\left(\rd\omega\right)},\label{eq:kernelmixture}
\end{equation}
where $Q_{\omega}$ is a Markov transition kernel from $\mathcal{Z}$
to $\mathcal{Z}$. In Section \ref{subsec:Sufficient-conditions-for-Local},
(\ref{eq:Hlocalcsum}), (\ref{eq:Ratelocal}) and (\ref{eq:Kernellocal})
simply correspond to (\ref{eq:Unbiasedestimate}), (\ref{eq:ratemixture})
and (\ref{eq:kernelmixture}) if we select $\mu$ as the measure such
that $\mu\left(\left\{ i\right\} \right)=1$ for all $i\in\Omega=[n]$.
The sufficient conditions of the previous section can be directly
generalized.

\begin{condition}\label{hypA:Sassumptions}Conditions on $\phi,$
$\{\lambda_{\omega}:\omega\in\Omega\}$ and $\{Q_{\omega}:\omega\in\Omega\}$ 
\begin{enumerate}
\item There exists a $\rho$-preserving mapping $\mathcal{S}:\mathcal{Z}\to\mathcal{Z}$.
\item The event rates $\{\lambda_{\omega}:\omega\in\Omega\}$ satisfy
\begin{equation}
\int\left\{ \lambda_{\omega}\left(\mathcal{S}\left(z\right)\right)-\lambda_{\omega}\left(z\right)\right\} \mu\left(\rd\omega\right)=\nabla\cdot\phi\left(z\right)-\left\langle \nabla H\left(z\right),\phi\left(z\right)\right\rangle .\label{eq:Ssufficientconditioninvariance}
\end{equation}
\item For all $\omega\in\Omega$, the transition kernel $Q_{\omega}$ satisfies
\begin{equation}
\int\rho\left(\rd z\right)\lambda_{\omega}\left(z\right)Q_{\omega}\left(z,\rd z'\right)=\rho\left(\mathcal{S}^{-1}\left(\rd z'\right)\right)\lambda_{\omega}\left(\mathcal{S}\left(z'\right)\right).\label{eq:SQpseudoinvariance}
\end{equation}
\end{enumerate}
\end{condition}

If $\mu$ is a probability measure and the derivative $\nabla H_{\omega}\left(z\right)$
is well-defined for almost all $\omega\in\Omega$ then under weak
regularity conditions, it follows from (\ref{eq:Unbiasedestimate})
that $\nabla H_{\omega}\left(z\right)$ is an unbiased estimate of
$\nabla H\left(z\right)$ when $\omega\sim\mu$ and\emph{ }Assumption
A\ref{hypA:Sassumptions}.2 will be satisfied for a divergence-free
field if\emph{ }
\begin{equation}
\lambda_{\omega}\left(\mathcal{S}\left(z\right)\right)-\lambda_{\omega}\left(z\right)=-\left\langle \nabla H_{\omega}\left(z\right),\phi\left(z\right)\right\rangle .\label{eq:SLocalLocalRateBalance}
\end{equation}
We will refer to this class of PD-MCMC as ``doubly stochastic''
in reference to doubly-stochastic Poisson processes.
\begin{prop}
\label{prop:Sinvariance}Assume (A\ref{hypA:Sassumptions}). Then
the PDMP admits $\rho$ as invariant distribution.
\end{prop}

\subsection{Existing PD-MCMC algorithms\label{subsec:Existing-PDMCMC-algorithms}}

All the existing algorithms we are aware of are based on the following
framework. The target distribution admits a density with respect to
Lebesgue measure on $\mathcal{X}=\mathbb{R}^{d}$ equal to $\pi\left(x\right)=\exp\left(-U\left(x\right)\right)$.
Letting $z=(x,v)$, an extended target distribution $\rho\left(\rd z\right)$
on $\mathcal{Z}=\mathcal{X}\times\mathcal{V}$ is then defined as
\begin{equation}
\rho\left(\rd z\right)=\pi\left(\rd x\right)\psi\left(\rd v\right),\label{eq:extendedtarget}
\end{equation}
where $\psi$ is an auxiliary distribution on $\mathcal{V}$, where
$\mathcal{V}$ can be for example either $\mathbb{R}^{d}$ or the
unit hypersphere $\mathbb{S}^{d-1}$ so that $n=2d$. The following
linear dynamics is then considered 
\[
\phi\left(z\right)=\left(v,0_{d}\right),
\]
so the resulting flow is analytically tractable and given by
\begin{equation}
\Phi_{t}\left(z\right)=\left(x+vt,v\right).\label{eq:LINEARDYNAMICS}
\end{equation}
In this case, we have $\nabla\cdot\phi=0$. Additionally, all these
algorithms rely on ~$\mathcal{S}(x,v)=(x,-v)$ which can be viewed
as a time reversal, so (\ref{eq:sufficientconditioninvariance}) becomes
\begin{equation}
\lambda\left(\mathcal{S}\left(z\right)\right)-\lambda\left(z\right)=\lambda\left(x,-v\right)-\lambda\left(x,v\right)=-\left\langle \nabla U\left(x\right),v\right\rangle .\label{eq:RateinvarianceLineardynamics}
\end{equation}
These algorithms differ in the way the event rate and the transition
kernels are specified. We just give a few examples here and refer
the reader to the list of references for other examples.

\subsubsection{Bouncy particle sampler }

This algorithm proposed in \citep{P_dW_12} exploits any additive
decomposition of the potential $U$, i.e. 
\begin{equation}
U\left(x\right)={\displaystyle \sum_{i=1}^{m}U_{i}\left(x\right).}\label{eq:localBPSintensity}
\end{equation}

For $\lambda_{\textrm{ref}}>0$, it uses the event rate
\[
\lambda\left(z\right)=\lambda_{\textrm{ref}}+{\displaystyle \sum_{i=1}^{m}}\left\langle \nabla U_{i}\left(x\right),v\right\rangle _{+}
\]
where $x_{+}:=\mathrm{max}\left(0,x\right)$. It also relies on the
transition kernel
\begin{equation}
Q\left(z,\rd z'\right)=\frac{\lambda_{\textrm{ref}}}{\lambda\left(z\right)}\delta_{x}(\rd x')\psi(\rd v')+{\displaystyle \sum_{i=1}^{m}}\frac{\left\langle \nabla U_{i}\left(x\right),v\right\rangle _{+}}{\lambda\left(z\right)}\delta_{x}(\rd x')\delta_{R_{\nabla U_{i}}(x)v}(\rd v'),\label{eq:localBPSkernel}
\end{equation}
where, for any vector field $\nabla W:\mathbb{R}^{d}\rightarrow\mathbb{R}^{d}$,
we define $R_{\nabla W}(x)$ as 
\begin{equation}
R_{\nabla W}(x)v:=v-2\frac{\langle\nabla W(x),v\rangle}{|\nabla W(x)|{}^{2}}\nabla W(x).\label{eq:localbounces}
\end{equation}
We note that (\ref{eq:localbounces}) corresponds to a bounce as it
can be interpreted as a Newtonian collision with the plane perpendicular
to $\nabla W$ at $x$. In \citep{P_dW_12}, a normal distribution
is used for $\psi$ but the uniform distribution on $\mathbb{S}^{d-1}$
can also been used \citep{M_16,D_BC_D_17}. We are in the scenario
where $\lambda$ and $Q$ are of the form (\ref{eq:Ratelocal}) and
(\ref{eq:Kernellocal}) with $n=m+1$, $\lambda_{i}\left(z\right)=\frac{1}{m}\left\langle \nabla U_{i}\left(x\right),v\right\rangle _{+}$
and $Q_{i}\left(z,\rd z'\right)=\delta_{x}(\rd x')\delta_{R_{\nabla U_{i}}(x)v}(\rd v')$
for $i\in\left[m\right]$ and $\lambda_{n}\left(z\right)=\lambda_{\textrm{ref}}$,
$Q_{n}\left(z,\rd z'\right)=\delta_{x}(\rd x')\psi(\rd v')$. It can
be checked that Assumption A\ref{hypA:Lassumptions} holds in this
scenario. In particular, Assumption A\ref{hypA:Lassumptions}.2 can
be verified by checking the stronger condition (\ref{eq:LocalLocalRateBalance}).
Indeed, if we write $\nabla H_{i}:=(\nabla_{x}H_{i},\nabla_{v}H_{i})$
then (\ref{eq:LocalLocalRateBalance}) becomes $\lambda_{i}\left(x,-v\right)-\lambda_{i}\left(x,v\right)=-\left\langle \nabla_{x}H_{i},v\right\rangle $
which is satisfied for $H_{i}\left(z\right):=U_{i}\left(x\right)$
for $i\in\left[m\right]$ and $H_{n}\left(z\right):=0$. For $m=1$,
we refer to this algorithm as the global BPS and for $m>1$ as the
local BPS. The local BPS is computationally advantageous compared
to BPS when either $U_{i}\left(x\right)$ only depends of a subset
of the components of $x$, as for sparse graphical models, and/or
when $m$ is very large, as for big data applications.

The BPS algorithm has been further extended to the scenario where
one has access to an unbiased estimate of $\nabla U$; see \citep{P_G_C_P_17}
and \citep[Section 4.4.2]{F_B_P_R_16}. The validity of this algorithm
can be established as an application of the results of Section \ref{subsec:Sufficient-conditions-for-Stochastic}.
We are not aware of any implementation of this algorithm in scenarios
where $\mu$ is not an atomic measure with finite support, in which
case the algorithm is the local BPS.

\subsubsection{Zig-Zag sampler}

This algorithm proposed in \citep{B_F_R_16,B_R_16} uses for $\psi$
the uniform distribution on $\left\{ -1,1\right\} ^{d}$\footnote{In this scenario, $\rho\left(\rd z\right)$ does not admit a density
with respect to Lebesgue measure but the results discussed previously
can be directly extended to this scenario. }. It relies on the following event rates 
\[
\lambda_{i}\left(z\right)=\lambda_{\textrm{ref},i}+\left\langle \nabla_{i}U\left(x\right),v_{i}\right\rangle _{+},
\]
while the transition kernel is selected as
\[
Q_{i}\left(z,\rd z'\right)=\delta_{x}(\rd x')\delta_{-v_{i}}(\rd v'_{i})\prod_{j\neq i}\delta_{v_{j}}(\rd v'_{j}).
\]

It is also possible to further exploit any additive decomposition
of $U\left(x\right)$ within this framework and this has been used
to develop an efficient sampling algorithm for big data \citep{B_F_R_16}.
Again, it is easy to show that Assumption A\ref{hypA:Lassumptions}
is satisfied. 

\subsubsection{BPS sampler with randomized bounces\label{subsec:Randomizedbounces}}

Alternatives to bounces of the form (\ref{eq:localbounces}) have
been proposed where one uses
\begin{equation}
Q\left(z,\rd z'\right)=\delta_{x}(\rd x')Q_{x}\left(v,\rd v'\right)\label{eq:degeneratekernel}
\end{equation}
and $\psi\left(v\right)=g\left(|v|\right)$. In this case, Assumption
A\ref{hypA:Gassumptions}.3 is verified if 
\begin{equation}
\int\psi\left(\rd v\right)\lambda\left(x,v\right)Q_{x}\left(v,\rd v'\right)=\psi\left(\rd v'\right)\lambda\left(x,-v'\right).\label{eq:Qpseudodetailedbalancesimplified}
\end{equation}

Here $\psi$ will be the standard multivariate normal distribution.
We consider the scenario where $\lambda\left(x,v\right)=\left\langle \nabla U\left(x\right),v\right\rangle _{+}$
as in the global BPS.%
{} To present the various methods proposed in the literature, a decomposition
of the velocity similar to that adopted in \citep{M_S_17} is useful:
\begin{equation}
v=a_{\perp}\thinspace n_{\perp}+a_{\parallel}\thinspace n_{\parallel},\label{eq:Vdecomposition}
\end{equation}
where $n_{\perp}$ and $n_{\parallel}$ are unit norm vectors such
that
\begin{equation}
n_{\parallel}\propto-\nabla U\left(x\right),~n_{\perp}\propto v-\left\langle n_{\parallel},v\right\rangle n_{\parallel}.\label{eq:ParallelOrthogonal}
\end{equation}
All the randomized bounce procedures return a vector $v'$
\begin{equation}
v'=a_{\perp}'\thinspace n'_{\perp}+a_{\parallel}'\thinspace n{}_{\parallel},\label{eq:Vdecompositionafter}
\end{equation}
where $\langle n'_{\perp},n_{\parallel}\rangle=0$. With this notation,
we obtain $\lambda\left(x,-v'\right)=a_{\parallel}'{}_{+}|\nabla U\left(x\right)|.$

Let $\chi\left(k\right)$ and $\chi^{2}\left(k\right)$ be the $\chi$
and $\chi^{2}$ distributions respectively, with $k$ degrees of freedom.
Under $\psi$, the random variables $a_{\perp}$ and $a_{\parallel}$
are independent and satisfy
\begin{align}
a_{\perp}\sim\chi\left(d-1\right),~a_{\parallel}\sim\mathcal{N}\left(0,1\right).\label{eq:priorab}
\end{align}
Indeed, we have $a_{\perp}^{2}\sim\chi^{2}\left(d-1\right)$ and $a_{\perp}\geq0$.
We give below some examples of kernels $Q_{x}(v,\rd v')$ satisfying
Equation (\ref{eq:Qpseudodetailedbalancesimplified}).%

\begin{enumerate}
\item Independent sampling \citep{F_B_P_R_16}: \citep{F_B_P_R_16} proposes
using $Q_{x}\left(v,\rd v'\right)\propto\psi\left(\rd v'\right)\lambda\left(x,-v'\right)\propto a'_{\parallel+}\psi\left(\rd v'\right)$
which satisfies (\ref{eq:Qpseudodetailedbalancesimplified}) but a
scheme to sample this distribution was not given. Using the parameterization
(\ref{eq:Vdecomposition})-(\ref{eq:Vdecompositionafter}), (\ref{eq:priorab})
shows this can be achieved by sampling $a_{\parallel}'$ according
to a density proportional to $a_{\parallel+}'$ times the standard
normal density, which is equivalent to sampling $a_{\parallel}'\sim\chi\left(2\right)$.
Finally, sample $v^{*}\sim\psi$ and set $a_{\perp}'\thinspace n'_{\perp}=v^{*}-\left\langle v^{*},n{}_{\parallel}\right\rangle n_{\parallel}$.
\item Forward-event chain \citep{M_S_17}: In \citep{M_S_17}, $\psi$ is
the uniform distribution on $\mathbb{S}^{d-1},$ whereas we consider
the scenario where $\psi$ is the normal distribution. One uses $n'_{\perp}=n{}_{\perp}$,
set $a_{\parallel}'=-a_{\parallel}$ and $a_{\perp}'\sim\chi\left(d-1\right)$.
Alternatively, sample $a_{\parallel}'\sim\chi\left(2\right)$ and
set $a_{\perp}'=a_{\perp}$. For either scheme, we recover the method
of \citep{M_S_17} on $\mathbb{S}^{d-1}$ by normalizing $v',$ i.e.~setting
$\bar{v}'=v'/|v'|$. %
\item Autoregressive bounce: this is a new scheme where one samples $a_{\parallel}'\sim\chi\left(2\right)$
with probability $p_{b}$ and $a_{\parallel}'=-a_{\parallel}$ otherwise,
sample $v^{*}\sim\psi$ and set $a_{\perp}^{*}\thinspace n{}_{\perp}^{*}=v^{*}-\left\langle v^{*},n{}_{\parallel}\right\rangle n_{\parallel}$.
Finally, set $a_{\perp}'\thinspace n'_{\perp}=\rho\thinspace a_{\perp}\thinspace n{}_{\perp}+\sqrt{1-\rho^{2}}\thinspace a_{\perp}^{*}\thinspace n{}_{\perp}^{*}$
for $\rho\in[-1,1]$.
\end{enumerate}
The properties of these randomized bounces are not yet well understood.
In Section \ref{sec:Numerical-results}, we compare them experimentally
on a variety of models.

\subsection{Hamiltonian PD-MCMC\label{subsec:Hamiltonian-PDMP}}

Although all previously proposed methods rely on the linear flow (\ref{eq:LINEARDYNAMICS}),
the framework presented in Section \ref{subsec:Sufficient-conditions-for-invariance}
is much more flexible. We exploit here this generalization to provide
novel continuous-time PD-MCMC algorithms relying on Hamiltonian dynamics.\footnote{The first arXiv version of \citep{BC_D_V_15} proposed a version of
the BPS algorithm using Hamiltonian dynamics but uses a different
approach based on manifolds. The algorithm suggested therein does
not preserve the correct invariant distribution. } As in Section \ref{subsec:Existing-PDMCMC-algorithms}, we consider
targets of the form $\rho(z)=\pi\left(x\right)\psi\left(v\right)$
with $\pi\left(x\right)=\exp\left(-U\left(x\right)\right)$ being
the density of interest on $\mathcal{X}=\mathbb{R}^{d}$ and $\psi$
the standard multivariate normal on $\mathcal{V=\mathbb{R}}^{d}$.
We use here the Hamiltonian flow $\Phi_{t}$ associated with the Hamiltonian
\begin{equation}
\widehat{H}\left(z\right)=V\left(x\right)+K(v),\label{eq:HamiltonianApprox}
\end{equation}
where $K(v)=v^{T}v/2$ and $\mu\left(x\right)\propto\exp(-V\left(x\right))$
is an auxiliary probability density ensuring $\Phi_{t}$ is analytically
tractable, e.g., $V$ is quadratic or linear \citep{P_P_14}. For
example if $\pi\left(x\right)$ is a posterior density arising from
a Gaussian prior, then $\mu\left(x\right)$ could be this Gaussian
prior. Alternatively, $\mu\left(x\right)$ can always be selected
as a Gaussian approximation to $\pi\left(x\right)$. We can then rewrite
the target as $\rho(z)=\exp\left(-H\left(z\right)\right)$ where 
\[
H\left(z\right)=\widetilde{U}\left(x\right)+V\left(x\right)+K(v),
\]
where $\widetilde{U}\left(x\right):=U\left(x\right)-V\left(x\right)$.
This is the same rationale as in elliptical slice sampling-type algorithms
\citep{M_A_2010,B_C_2016}: both schemes use an exact Hamiltonian
dynamics associated with an approximation of $\pi$ to explore the
space. The difference with these algorithms and the method proposed
here is that we correct for the discrepancy between $\mu$ and $\pi$
by using a PDMP mechanism instead of slice sampling techniques. 

The Hamiltonian flow $\Phi_{t}$ is induced by the ODE of drift $\phi=\left(\phi_{x},\phi_{v}\right)$
where $\phi_{x}=\nabla_{v}\widehat{H}\left(z\right)=v$ and $\phi_{v}=-\nabla_{x}\widehat{H}\left(z\right)=-\nabla V\left(x\right)$.
Hence, we have $\nabla\cdot\phi\left(z\right)=0$ and 

\begin{eqnarray*}
\nabla\cdot\phi\left(z\right)-\langle\nabla H\left(z\right),\phi(z)\rangle & = & -\left\langle \nabla_{x}H\left(z\right),\phi_{x}\right\rangle -\left\langle \nabla_{v}H\left(z\right),\phi_{v}\right\rangle \\
 & = & -\left\langle \nabla\widetilde{U}\left(x\right),v\right\rangle -\left\langle \nabla V\left(x\right),v\right\rangle +\left\langle \nabla V\left(x\right),v\right\rangle \\
 & = & -\left\langle \nabla\widetilde{U}(x),v\right\rangle .
\end{eqnarray*}

One can check that Assumption A.\ref{hypA:Gassumptions} is thus verified
for $\mathcal{S}\left(z\right)=(x,-v)$ if we use an event rate and
transition kernel as in the `global' BPS but based on $\tilde{U}$
only\footnote{For $\widetilde{U}=0$, this algorithm corresponds to a continuous-time
HMC algorithm with momentum/velocity refreshment at Poisson times. }
\begin{eqnarray*}
\lambda\left(z\right) & := & \lambda_{\textrm{ref}}+\left\langle \nabla\widetilde{U}\left(x\right),v\right\rangle _{+},\\
Q\left(z,\rd z'\right) & := & \frac{\lambda_{\textrm{ref}}}{\lambda\left(z\right)}\delta_{x}(\rd x')\psi(\rd v')+\frac{\left\langle \nabla\widetilde{U}\left(x\right),v\right\rangle _{+}}{\lambda\left(z\right)}\delta_{x}(\rd x')\delta_{R_{\nabla\widetilde{U}}(x)v}(\rd v').
\end{eqnarray*}

We can alternatively use the randomized bounces described in Section
\ref{subsec:Randomizedbounces} substituting $\widetilde{U}$ for
$U$. Figure \ref{fig:Examples-of-paths} illustrates a sample path
obtained from the resulting Hamiltonian BPS algorithm. Local and doubly
stochastic versions of this algorithm as for BPS \citep{P_dW_12,BC_D_V_15,P_P_14}
can also be directly developed.%
{} 

In the big data examples considered in \citep{BC_D_V_15,B_F_R_16,P_G_C_P_17},
one could for example use for $\mu$ a Gaussian approximation of $\pi$.
A local algorithm can then be obtained using for $\nabla\widetilde{U}_{i}$
the difference of the gradient of the log-likelihood corresponding
to data $i$ and the properly rescaled gradient of the log-approximate
posterior, as in \citep{B_F_R_16}. If the terms $\nabla\widetilde{U}_{i}$
are locally bounded, we can simulate exactly the PDMP using thinning
techniques which boil down to data subsampling \citep{BC_D_V_15,B_F_R_16}.
This provides an alternative to \citep{C_F_G_14} which also exploits
Hamiltonian dynamics and subsampling but does not preserve $\pi$
as invariant distribution.

\begin{figure}
\includegraphics[width=0.9\columnwidth]{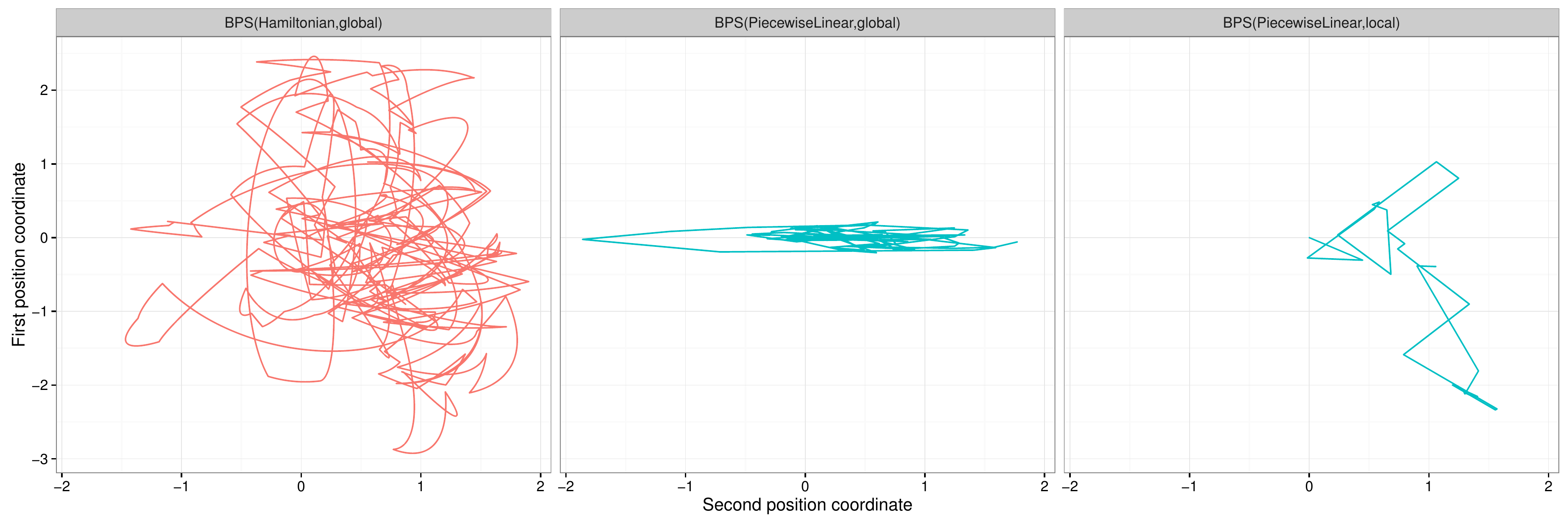}\centering\caption{\label{fig:Examples-of-paths}Examples of paths for the Hamiltonian
BPS (left), global BPS (middle) and local BPS (blue). All algorithms
are run for a wall clock time of 150ms on a 1000-dimensional Gaussian
latent field with sparsely observed Poisson distributed observations
(one observation for every 100 latent variables), see Section \ref{subsec:Hamiltonian-PDMP-results}
for details. The first two position coordinates are shown. }
\end{figure}

Finally, we also note that the methods introduced in this section
can be combined with the HMC algorithm of \citep{P_P_14} proposed
to perform exact simulation of constrained normal distributions. This
extends significantly the applicability of the work in \citep{P_P_14},
which can be viewed as a special case where $\widetilde{U}=0$. An
alternative approach to constrained problems is proposed in \citep{B_BC_D_D_F_R_J_16}
but it is limited to piecewise-linear dynamics.

\subsection{Using generalized Metropolis\textendash Hastings transitions at event
times\label{subsec:GMHalgorithm}}

All the algorithms we have considered so far are such that only a
part of the state $z=\left(x,v\right)$ is updated at event times,
i.e.,\ the transition kernel is of the form $Q\left(z,\rd z'\right)=\delta_{x}(\rd x')Q_{x}\left(v,\rd v'\right).$
We might be interested in designing more general transitions kernels
satisfying Assumption A\ref{hypA:Gassumptions}.3 and similarly Assumption
A\ref{hypA:Lassumptions}.3 or Assumption A\ref{hypA:Sassumptions}.3.

For sake of illustration, consider Assumption A\ref{hypA:Gassumptions}.3.
This can be rewritten as 
\begin{equation}
\int\bar{\rho}\left(\rd z\right)Q\left(z,\rd z'\right)=\bar{\rho}\left(\mathcal{S}^{-1}\left(\rd z'\right)\right)\label{eq:Qpseudoinvariancerewritten}
\end{equation}
for the probability measure $\bar{\rho}\left(\rd z\right)\propto\rho\left(\rd z\right)\lambda\left(z\right)$
assuming that $\int\rho\left(\rd z\right)\lambda\left(z\right)<\infty$,
a weak condition which we assume holds. If the mapping $\mathcal{S}$
is an involution, i.e., $\mathcal{S}^{-1}=\mathcal{S}$, and we can
design a kernel $Q$ satisfying the so-called \emph{skewed detailed
balance }condition 
\begin{equation}
\bar{\rho}\left(\rd z\right)Q\left(z,\rd z'\right)=\bar{\rho}\left(\mathcal{S}\left(\rd z'\right)\right)Q\left(\mathcal{S}\left(z'\right),\mathcal{S}\left(\rd z\right)\right),\label{eq:skeweddetailedbalancecondition}
\end{equation}
then it follows directly by integrating both terms in this equality
with respect to variable $z$ that it will satisfy (\ref{eq:Qpseudoinvariancerewritten}).

We present here a generic mechanism which can be used to achieve this
known as the Generalized Metropolis\textendash Hastings (GMH) algorithm.
The GMH algorithm is a simple extension of MH; see for example \citep[pp. 74--77]{lelievreroussetstoltz2010}.
For a probability measure $\nu\left(\rd z\right)=\nu\left(z\right)\rd z$
on $\mathcal{Z}$, let us consider the following GMH kernel defined
for a Markov proposal kernel $M$ by 
\begin{equation}
T\left(z,\rd z'\right)=\beta\left(z,z'\right)M\left(z,\rd z'\right)+\left\{ 1-\int\beta\left(z,w\right)M\left(z,\rd w\right)\right\} \delta_{\mathcal{S}\left(z\right)}\left(\rd z'\right)\label{eq:GMHkernel}
\end{equation}
where 
\begin{equation}
\beta\left(z,z'\right)=g\left(\frac{\nu\left(\mathcal{S}\left(\rd z'\right)\right)M\left(\mathcal{S}\left(z'\right),\mathcal{S}\left(\rd z\right)\right)}{\nu\left(\rd z\right)M\left(z,\rd z'\right)}\right).\label{eq:GMHkernelacceptanceproba}
\end{equation}
We make the following assumptions:

\begin{condition}\label{hypA:GMHassumptions}Conditions on $\nu,$
$\mathcal{S},$ $M$ and $g$
\begin{enumerate}
\item The mapping $\mathcal{S}$ is an involution, i.e., $\mathcal{S}^{-1}=\mathcal{S}$.
\item The Radon-Nikodym derivative 
\[
\frac{\nu\left(\mathcal{S}\left(\rd z'\right)\right)M\left(\mathcal{S}\left(z'\right),\mathcal{S}\left(\rd z\right)\right)}{\nu\left(\rd z\right)M\left(z,\rd z'\right)}
\]
is defined and positive for almost all $\left(z,z'\right)\in\mathcal{Z\times Z}.$
\item The function $g:\mathbb{R}^{+}\rightarrow\left[0,1\right]$ satisfies
$g\left(r\right)=rg\left(1/r\right)$.
\end{enumerate}
\end{condition}

Assumption A\ref{hypA:GMHassumptions}.1 is satisfied for $g\left(r\right)=\min\left(1,r\right)$.
For a deterministic proposal $M\left(z,\rd z'\right)=\delta_{\Psi\left(z\right)}\left(\rd z'\right),$
Assumption A\ref{hypA:GMHassumptions}.3 is satisfied if $\Psi$ admits
an inverse $\Psi^{-1}$ such that
\begin{equation}
\Psi^{-1}=\mathcal{S}\circ\Psi\circ\mathcal{S}\label{eq:MappingBalance}
\end{equation}
and then the acceptance probability is given by 
\begin{equation}
\beta\left(z,z'\right)=\beta\left(z\right)=g\left(\frac{\nu\left(\mathcal{S}\circ\Psi\left(\rd z\right)\right)}{\nu\left(\rd z\right)}\right).\label{eq:AcceptanceProbaGMHdeterministicmoves}
\end{equation}

\begin{prop}
\label{prop:GMHinvariance}Assume (A\ref{hypA:GMHassumptions}). Then
the GMH kernel $T$ defined by (\ref{eq:GMHkernel}) satisfies the
following skewed detailed balance condition 
\begin{equation}
\nu\left(\rd z\right)T\left(z,\rd z'\right)=\nu\left(\mathcal{S}\left(\rd z'\right)\right)T\left(\mathcal{S}\left(z'\right),\mathcal{S}\left(\rd z\right)\right).\label{eq:skeweddetailedbalanceconditiongeneral}
\end{equation}
If additionally $\mathcal{S}$ is a $\nu$-preserving mapping then
the GMH kernel is $\nu$-invariant.
\end{prop}

The proof of this result follows from direct calculations given in
the Appendix and can also be found in \citep[pp. 74--77]{lelievreroussetstoltz2010}.
Using this result, it is possible to check easily Assumption A\ref{hypA:Gassumptions}.3
for the BPS and Zig-Zag processes. For example, for the BPS, $Q$
is of the form (\ref{eq:GMHkernel}) with $\nu=\bar{\rho}$, $\mathcal{S}^{-1}=\mathcal{S}$,
$g\left(r\right)=\min\left(1,r\right)$ as we use a deterministic
proposal $\Psi\left(z\right)=\left(x,R_{\nabla U}\left(x\right)v\right)$
which verifies $\Psi^{-1}=\mathcal{S}\circ\Psi\circ\mathcal{S}$ so
$\beta\left(z,z'\right)=1$ for all $z,z'$. Hence by Proposition
\ref{prop:GMHinvariance}, $Q$ satisfies the skewed detailed balance
(\ref{eq:skeweddetailedbalancecondition}), hence it satisfies (\ref{eq:Qpseudoinvariancerewritten}).

The benefit of the GMH approach is that it allows us to define much
more general kernels at event times. For example one could use a deterministic
proposal with $\Psi\left(z\right)=\left(x,R_{\nabla\widehat{U}}\left(x\right)v\right)$
where $\nabla\widehat{U}$ is a computationally cheap approximation
of $\nabla U$. It is valid to use such a deterministic proposal at
it satisfies $\Psi^{-1}\left(z\right)=\mathcal{S}\circ\Psi\circ\mathcal{S}(z)$.
In this case, there is a probability of the bounce being rejected
and setting $z'\leftarrow\mathcal{S}\left(z\right)$. We can also
use transition kernels which modify the component $x$ of $z$.%

\section{Discrete-time PDMP and PD-MCMC\label{sec:Discrete-Time-PDMP-PDMCMC}}

We introduce here the class of discrete-time PDMP and present general
conditions for such processes to ensure invariance w.r.t. a strictly
positive density $\rho\left(z\right)=\exp(-H\left(z\right))$. These
conditions parallel the conditions given Section \ref{subsec:Sufficient-conditions-for-invariance}
for continuous-time algorithms. 

\subsection{Discrete-time PDMP}

As in the continuous-time scenario, we assume for simplicity that
$\mathcal{Z}=\mathbb{R}^{n}$. A $\mathcal{Z}$-valued discrete-time
PDMP process $\left\{ z_{t};t\in\mathbb{N}\right\} $ involves a deterministic
dynamics altered by random jumps at random event times. It is defined
through 
\begin{enumerate}
\item a diffeomorphism $\Phi:\mathcal{Z}\rightarrow\mathcal{Z}$ with the
absolute value of the determinant of the Jacobian satisfying $|\nabla\Phi\left(z\right)|>0$
for all $z$,
\item an acceptance probability $\alpha:\mathcal{Z}\rightarrow[0,1]$ with
$1-\alpha\left(z\right)$ being the probability of having an event
at the next time step when the current state is $z$, and
\item a Markov transition kernel $Q$ from $\mathcal{Z}$ to $\mathcal{Z}$
where the state at event time $t$ is given by $z_{t}\sim Q\left(z_{t-1},\cdot\right)$.
\end{enumerate}
Algorithm \ref{alg:continuoustimePDMP-2} describes how to simulate
the path of a discrete-time PDMP. It will be convenient to use the
conventions $\prod_{i=0}^{-1}=1$, $\Phi^{0}\left(z\right)=z$ and
$\Phi^{r+1}\left(z\right)=\Phi^{r}\circ\Phi\left(z\right)$ for $r\in\mathbb{N}$. 

\begin{algorithm}[H]
\protect\caption{Simulation of discrete-time PDMP ~\label{alg:continuoustimePDMP-2}}

\begin{enumerate}
\item Initialize $z_{0}$ arbitrarily on $\mathcal{Z}$ and set $t_{0}\leftarrow0$. 
\item for $k=1,2,\ldots$ do
\begin{enumerate}
\item Sample inter-event time $\tau_{k}$, where $\tau_{k}$ is a non-negative
integer-valued random variable such that 
\begin{equation}
\mathbb{P}\left(\tau_{k}=j\right)=\left\{ 1-\alpha\left(\Phi^{j}\left(z_{t_{k-1}}\right)\right)\right\} \prod_{i=0}^{j-1}\alpha\left(\Phi^{i}\left(z_{t_{k-1}}\right)\right).\label{eq:sim_eventtimeDT}
\end{equation}
\item If $\tau_{k}\geq1$ then for $r\in\{1,...,\tau_{k}\}$$,~$set 
\begin{equation}
z_{t_{k-1}+r}\gets\Phi^{r}(z_{t_{k-1}}).\label{eq:discreteflow}
\end{equation}
\item Set $t_{k}\gets t_{k-1}+\tau_{k}+1$ and sample
\begin{equation}
z_{t_{k}}\sim Q(z_{t_{k}-1},\cdot).\label{eq:SimTransition-DT}
\end{equation}
\end{enumerate}
\end{enumerate}
\end{algorithm}

The process $\left\{ z_{t};t\in\mathbb{N}\right\} $ is nothing but
a Markov process of transition kernel%
\begin{equation}
K\left(z,\rd z'\right)=\alpha\left(z\right)\delta_{\Phi\left(z\right)}\left(\rd z'\right)+\left(1-\alpha\left(z\right)\right)Q\left(z,\rd z'\right).\label{eq:kerneldiscretetimePDMP}
\end{equation}

\subsection{From discrete-time PDMP to PD-MCMC\label{subsec:Sufficient-conditions-for-invariance-discretetime}}

Similarly to Section \ref{subsec:Sufficient-conditions-for-invariance},
assume we are interested in sampling a strictly positive density $\rho\left(z\right)$
given by (\ref{eq:Pho(z)=00003Dexp(-H(z))}) using a discrete-time
PDMP process. Invariance of the kernel $K$ with respect to $\rho$
is satisfied if, by definition, one has 
\begin{equation}
\mathit{\int\rho\left(\rd z\right)K\left(z,\rd z'\right)}=\rho\left(\rd z'\right).\label{eq:invariancePDMC-DT}
\end{equation}
From (\ref{eq:kerneldiscretetimePDMP}), (\ref{eq:invariancePDMC-DT})
can be rewritten as
\begin{equation}
\rho\left(\Phi^{-1}\left(z'\right)\right)\alpha\left(\Phi^{-1}\left(z'\right)\right)\left|\nabla\Phi^{-1}\left(z'\right)\right|\rd z'+\int\rho\left(\rd z\right)\left\{ 1-\alpha\left(z\right)\right\} Q\left(z,\rd z'\right)=\rho\left(\rd z'\right).\label{eq:invariancekerneldiscretetime}
\end{equation}
All the following developments could also be adapted to sample from
distributions on discrete spaces but this will not be discussed here.

\subsubsection{Sufficient conditions for global methods\label{subsec:Sufficient-conditions-for-Globa-discrete}}

We provide here useful sufficient conditions on $\Phi,$ $\alpha$,
and $Q$ to ensure $\rho$-invariance of the associated discrete-time
PDMP, without making any structural assumption on these objects.

\begin{condition}\label{hypA:Gassumptions-DT}Conditions on $\Phi,$
$\alpha$, and $Q$ 
\begin{enumerate}
\item There exists a $\rho$-preserving mapping $\mathcal{S}:\mathcal{Z}\to\mathcal{Z}$.
\item The acceptance probability $\alpha$ satisfies 
\begin{equation}
\left\{ -\log\alpha\left(\mathcal{S}\circ\Phi\left(z\right)\right)\right\} -\left\{ -\log\alpha\left(z\right)\right\} =\log\left|\nabla\Phi\left(z\right)\right|-\left\{ H\left(\Phi\left(z\right)\right)-H\left(z\right)\right\} .\label{eq:sufficientconditioninvariance-DT}
\end{equation}
\item The kernel $Q$ satisfies
\begin{equation}
\int\rho\left(\rd z\right)\left(1-\alpha\left(z\right)\right)Q\left(z,\rd z'\right)=\rho(\mathcal{S}^{-1}\left(\rd z'\right))\left(1-\alpha\left(\mathcal{S}\left(z'\right)\right)\right).\label{eq:Qpseudoinvariance-discrete}
\end{equation}
\end{enumerate}
\end{condition}

\begin{rem}
Conditions A\ref{hypA:Gassumptions-DT}.1 to A\ref{hypA:Gassumptions-DT}.3
parallel the conditions A\ref{hypA:Gassumptions}.1 to A\ref{hypA:Gassumptions}.3.
\end{rem}

\begin{prop}
\label{prop:Ginvariance-DT}Assume (A\ref{hypA:Gassumptions-DT}).
Then the discrete-time PDMP admits $\rho$ as invariant distribution.
\end{prop}

\begin{rem}
When $\mathcal{S}$ is an involution so that $\rho(\mathcal{S}^{-1}\left(\rd z'\right))=\rho(\mathcal{S}\left(\rd z'\right))$,
condition A\ref{hypA:Gassumptions-DT}.3 can be interpreted as a ``skewed''
invariance condition on $\nu(\rd z)\propto\rho(\rd z)(1-\alpha(z))$.
The quantity $\rho(\rd z)\left(1-\alpha\left(z\right)\right)$ is
proportional to the invariant distribution of the ``jump chain,''
i.e.\ the distribution of those states where the proposal $\Phi\left(z\right)$
is rejected. It has a clear analogue in the continuous-time scenario
where the jumps occur at states with distribution proportional to
$\rho(\rd z)\lambda(z)$.
\end{rem}

\subsubsection{Sufficient conditions for local methods\label{subsec:Sufficient-conditions-for-Local-DT}}

In scenarios where $H\left(z\right)$ can be decomposed as in (\ref{eq:Hlocalcsum}),
it will prove convenient to consider an acceptance probability of
the form 
\begin{equation}
\alpha\left(z\right)=\prod_{i=1}^{n}\alpha_{i}\left(z\right)\label{eq:acceptanceprobalocalDT}
\end{equation}
where $\alpha_{i}:\mathcal{Z}\rightarrow[0,1]$ are themselves acceptance
probabilities\footnote{The authors in \citep{M_K_K_14} derive a continuous-time local PD-MCMC
by using this `factorized' acceptance probability, using a mapping
$\Phi\left(z\right)=(x+\epsilon v,v)$ and taking the limit as $\epsilon\rightarrow0$.
However for a strictly positive $\epsilon>0$, they do not define
a discrete-time local PD-MCMC as proposed here.}. To sample an event of probability $\alpha\left(z\right)$, we can
sample independent Bernoulli variables $B_{i}$, such that $B_{i}\sim\mathrm{Ber}(1-\alpha_{i}(z))$
for $i\in[n]$ where $\mathrm{Ber}(p)$ is the Bernoulli distribution
of parameter $p$. Hence the probability of the event $B=(0,...,0)$
where $B=(B_{1},...,B_{n})$ is $\alpha\left(z\right)$. Thus if $B:=(0,...,0)$,
we will set $z'\leftarrow\Phi\left(z\right)$. Otherwise, that is
if $B\in\mathcal{\mathcal{B}}$ where $\mathcal{B}=\left\{ 0,1\right\} ^{n}\setminus\left\{ 0\right\} ^{n}$,
then we will sample $z'\sim Q(z,\cdot)$ where 
\begin{equation}
Q\left(z,\rd z'\right)=\sum_{b\in\mathcal{\mathcal{B}}}\mathbb{Q}_{|B|\geq1}\left(b|z\right)Q_{b}\left(z,\rd z'\right).\label{eq:mixturerepresentationkernellocalDT}
\end{equation}
In this expression $Q_{b}$ is a Markov kernel and $\mathbb{Q}_{|B|\geq1}\left(b|z\right)$
is the distribution of $B$ conditioned upon $|B|:=\sum_{i=1}^{n}B_{i}\geq1$
which is given by
\begin{equation}
\mathbb{Q}_{|B|\geq1}\left(b|z\right)=\frac{\prod_{i=1}^{n}\mathrm{Ber}(b_{i};1-\alpha_{i}(z))}{1-\alpha\left(z\right)}.\label{eq:conditionaldistributionatleastaneventlocalDT}
\end{equation}

Based on these structural assumptions on $\alpha$ and $Q$, we can
provide useful sufficient ``local'' conditions on $\Phi,$ $\{\alpha_{i}:i\in\left[n\right]\}$
and $\left\{ Q_{B}:b\in\mathcal{B}\right\} $ to ensure invariance
of the associated discrete-time PDMP w.r.t.~$\rho$ is satisfied. 

\begin{condition}\label{hypA:Lassumptions-DT}Conditions on $\phi,$
$\{\alpha_{i}:i\in\left[n\right]\}$, and $\left\{ Q_{b}:b\in\mathcal{B}\right\} $
\begin{enumerate}
\item There exists a $\rho$-preserving mapping $\mathcal{S}:\mathcal{Z}\to\mathcal{Z}$.
\item The acceptance probabilities $\{\alpha{}_{i}:i\in\left[n\right]\}$
satisfy
\begin{equation}
\sum_{i=1}^{n}\left\{ -\log\alpha_{i}\left(\mathcal{S\circ}\Phi\left(z\right)\right)\right\} -\left\{ -\log\alpha_{i}\left(z\right)\right\} =\log\left|\nabla\Phi\left(z\right)\right|-\left\{ H\left(\Phi\left(z\right)\right)-H\left(z\right)\right\} .\label{eq:LocalRateBalanceDT}
\end{equation}
\item For all $b\in\mathcal{B}$, the transition kernel $Q_{b}$ satisfies
\begin{equation}
\int\rho\left(\rd z\right)\left(1-\alpha\left(z\right)\right)\mathbb{Q}_{|B|\geq1}\left(b|z\right)Q_{b}\left(z,\rd z'\right)=\rho(\mathcal{S}^{-1}\left(\rd z'\right))\left(1-\alpha\left(\mathcal{S}\left(z'\right)\right)\right)\mathbb{Q}_{|B|\geq1}\left(b|\mathcal{S}\left(z'\right)\right).\label{eq:QlocalpseudoinvarianceDT}
\end{equation}
\end{enumerate}
\end{condition}

For a mapping such that $\left|\nabla\Phi\right|=1$, then Assumption
A\ref{hypA:Lassumptions-DT}.2 is satisfied if for all $i\in\left[n\right]$
\begin{equation}
\left\{ -\log\alpha_{i}\left(\mathcal{S\circ}\Phi\left(z\right)\right)\right\} -\left\{ -\log\alpha_{i}\left(z\right)\right\} =-\left\{ H_{i}\left(\Phi\left(z\right)\right)-H_{i}\left(z\right)\right\} .\label{eq:LocalLocalRateBalance-DT}
\end{equation}

\begin{prop}
\label{prop:Linvariance-DT}Assume (A\ref{hypA:Lassumptions-DT}).
Then the discrete-time PDMP admits $\rho$ as invariant distribution.
\end{prop}

\subsubsection{Sufficient conditions for doubly stochastic methods\label{subsec:Sufficient-conditions-for-Stochastic-DT}}

Consider finally the scenario where $H\left(z\right)$ is given by
(\ref{eq:Unbiasedestimate}). In this context, we consider an acceptance
probability of the form 
\begin{equation}
\alpha\left(z\right)=\exp\left\{ \int\log\alpha_{\omega}\left(z\right)\mu\left(\rd\omega\right)\right\} \label{eq:acceptancedoublystochasticDT}
\end{equation}
where $\alpha_{\omega}:\mathcal{Z}\rightarrow[0,1]$ which is a generalization
of (\ref{eq:acceptanceprobalocalDT}) from the measure $\mu\left(\left\{ i\right\} \right)=1$
on a finite space $\Omega=[n]$ to an arbitrary measure on a general
space. Obviously when $\Omega$ is not finite, the strategy previously
adopted to simulate an event of probability $\alpha\left(z\right)$
is not applicable. However, this can be achieved by simulating a Poisson
process $P$ on $\Omega$ of rate $\Lambda\left(\rd\omega\right)=-\log\alpha_{\omega}\left(z\right)\,\mu\left(\rd\omega\right)$,
the law of which we denote with $\mathbb{Q}\left(\rd P|z\right)$,
and noticing that $\alpha\left(z\right)$ is the void probability
of $P$. A similar idea was used in a different context in \citep{B_P_R_F_06}.
Hence if the number of points is null, i.e.~$|P|=0$, then we will
set $z'\leftarrow\Phi\left(z\right).$ If $|P|\geq1$, that is $P\in\mathcal{\mathcal{\mathscr{\mathcal{\mathcal{P}}}}}$
where $\mathcal{P}$ is the set of configurations of the Poisson process
having at least one point, then we will sample $z'\sim Q(z,\cdot)$
where 
\begin{equation}
Q\left(z,\rd z'\right)=\int_{\mathcal{P}}\mathbb{Q}_{|P|\geq1}\left(\rd P|z\right)Q_{P}\left(z,\rd z'\right).\label{eq:kernelmixture-2}
\end{equation}
In this expression $Q_{P}$ is a Markov kernel and $\mathbb{Q}_{|P|\geq1}\left(\rd P|z\right)$
is the law of the Poisson process $P$ conditioned upon the event
$|P|\geq1$ which is given by
\begin{equation}
\mathbb{Q}_{|P|\geq1}\left(\rd P|z\right)=\frac{\mathbb{I}\left(|P|\geq1\right)}{1-\alpha\left(z\right)}\mathbb{Q}\left(\rd P|z\right).\label{eq:conditionaldistributionatleastaneventdoublystochasticlDT}
\end{equation}

\begin{condition}\label{hypA:Sassumptions-DT}Conditions on $\phi,$
$\{\alpha_{\omega}:\omega\in\Omega\}$ and $\left\{ Q_{P}:P\in\mathcal{\mathcal{\mathscr{\mathcal{\mathcal{P}}}}}\right\} $
\begin{enumerate}
\item There exists a $\rho$-preserving mapping $\mathcal{S}:\mathcal{Z}\to\mathcal{Z}$.
\item The acceptance probabilities $\{\alpha_{\omega}:\omega\in\Omega\}$
satisfy
\begin{equation}
\int\left[\left\{ -\log\alpha_{\omega}\left(\mathcal{S}\circ\Phi\left(z\right)\right)\right\} -\left\{ -\log\alpha_{\omega}\left(z\right)\right\} \right]\mu\left(\rd\omega\right)=\log\left|\nabla\Phi\left(z\right)\right|-\left\{ H\left(\Phi\left(z\right)\right)-H\left(z\right)\right\} .\label{eq:LocalRateBalanceDT-1}
\end{equation}
\item For all $P\in\mathcal{\mathcal{\mathscr{\mathcal{\mathcal{P}}}}}$,
the transition kernel $Q_{P}$ satisfies
\begin{equation}
\int\rho\left(\rd z\right)\left(1-\alpha\left(z\right)\right)\mathbb{Q}_{|P|\geq1}\left(\rd P|z\right)Q_{P}\left(z,\rd z'\right)=\rho(\mathcal{S}^{-1}\left(\rd z'\right))\left(1-\alpha\left(\mathcal{S}\left(z'\right)\right)\right)\mathbb{Q}_{|P|\geq1}\left(\rd P|\mathcal{S}\left(z'\right)\right).\label{eq:QlocalpseudoinvarianceDT-1}
\end{equation}
\end{enumerate}
\end{condition}

Assumption A.\ref{hypA:Sassumptions-DT}.3 is an informal expression
meaning that we assume that for $\mathbb{Q}_{|P|\geq1}\left(\rd P|z\right)$-almost
all $P\in\mathcal{\mathcal{\mathscr{\mathcal{\mathcal{P}}}}}$
\[
\int\rho\left(\rd z\right)\left(1-\alpha\left(z\right)\right)\frac{\mathbb{\rd Q}_{|P|\geq1}\left(P|z\right)}{\rd\mathbb{Q}_{|P|\geq1}\left(P|\mathcal{S}\left(z'\right)\right)}Q_{P}\left(z,\rd z'\right)=\rho(\mathcal{S}^{-1}\left(\rd z'\right))\left(1-\alpha\left(\mathcal{S}\left(z'\right)\right)\right),
\]
and the Radon-Nikodym derivative in the expression above is well-defined
and strictly positive for $Q_{P}\left(z,\rd z'\right)$ almost all
$z'$.

For a mapping such that $\left|\nabla\Phi\right|=1$, Assumption A\ref{hypA:Sassumptions-DT}.2
is satisfied if for all $\omega\in\Omega$
\begin{equation}
\left\{ -\log\alpha_{\omega}\left(\mathcal{S}\circ\Phi\left(z\right)\right)\right\} -\left\{ -\log\alpha_{\omega}\left(z\right)\right\} =-\left\{ H_{\omega}\left(\Phi\left(z\right)\right)-H_{\omega}\left(z\right)\right\} .\label{eq:LocalLocalRateBalance-SDT}
\end{equation}

\begin{prop}
\label{prop:Sinvariance-DT}Assume (A\ref{hypA:Sassumptions-DT}).
Then the discrete-time PDMP admits $\rho$ as invariant distribution.
\end{prop}

\subsection{Existing PD-MCMC algorithms\label{subsec:Existing-PDMCMC-algorithms-DT}}

A few algorithms proposed in the literature can be considered as special
instances of discrete-time PD-MCMC algorithms. They all rely on the
same framework discussed in Section \ref{subsec:Existing-PDMCMC-algorithms},
that is they sample an extended target density $\rho\left(z\right)=\exp(-H\left(z\right))=\pi\left(x\right)\psi\left(v\right)$
defined (\ref{eq:extendedtarget}) on $\mathcal{Z}=\mathbb{R}^{d}\times\mathbb{R}^{d}$
where $\pi$ is the target distribution of interest and $\psi$ is
a standard multivariate normal. They use a mapping such that $\left|\nabla\Phi\right|=1$,
$\Phi^{-1}=\mathcal{S}\circ\Phi\circ\mathcal{S}$ with $\mathcal{S}\left(z\right)=\left(x,-v\right)$
and $\alpha\left(z\right)=\min\left\{ 1,\rho\left(\Phi\left(z\right)\right)/\rho\left(z\right)\right\} $.
A fairly generic scheme is detailed in Algorithm \ref{alg:discretetimePD-MCMC}. 

\begin{algorithm}[H]
\caption{Discrete-time PD-MCMC~\label{alg:discretetimePD-MCMC}}

\begin{enumerate}
\item With probability $\min\left\{ 1,\rho\left(\Phi\left(z\right)\right)/\rho\left(z\right)\right\} ,$
set $z'\leftarrow\Phi\left(z\right)$.
\item Otherwise, sample $z^{*}\sim M\left(z,\cdot\right).$
\item With probability 
\[
\beta\left\{ \left(x,v\right),\left(x^{*},v^{*}\right)\right\} =\min\left\{ 1,\frac{\left[\rho\left(x^{*},v^{*}\right)-\rho\left(\Phi\left(x^{*},-v^{*}\right)\right)\right]_{+}M\left(\left(x^{*},-v^{*}\right),\left(x,-v\right)\right)}{\left[\rho\left(x,v\right)-\rho\left(\Phi\left(x,v\right)\right)\right]_{+}M\left(\left(x,v\right),\left(x^{*},v^{*}\right)\right)}\right\} ,
\]
set $z'\leftarrow z^{*}$, otherwise set $z'\leftarrow(x,-v)$.
\end{enumerate}
\end{algorithm}

This scheme satisfies Assumption A\ref{hypA:Gassumptions-DT}.1 to
Assumption A\ref{hypA:Gassumptions-DT}.3 and is thus $\rho$-invariant.
In particular Assumption A\ref{hypA:Gassumptions-DT}.3 is satisfied
as Steps 2 and 3 correspond to using for the event kernel $Q$ a GMH
kernel satisfying the skewed-detailed balance condition (\ref{eq:skeweddetailedbalanceconditiongeneral})
for $\nu\left(\rd z\right)\propto\rho\left(\rd z\right)\left(1-\alpha\left(z\right)\right)$. 
\begin{rem}
Algorithm~\ref{alg:discretetimePD-MCMC} can be alternatively viewed
as a composition of reversible kernels. First, a delayed-rejection
algorithm proposing $\Phi$ and, in case of rejection, then proposing
$M(z,\mathcal{S}^{-1}(\cdot))$. Second, the involution $\mathcal{S}$
is applied unconditionally. In the delayed-rejection framework, we
can view condition A\ref{hypA:Gassumptions-DT}.3 as a condition on
delayed-rejection kernels expressed in a sort of ``remainder'' form.
While our algorithm uses two proposals, extending this remainder condition
to multiple proposals would require that each $Q_{k}$ satisfies $\int\rho(\rd z)\prod_{i=1}^{k-1}(1-\alpha_{i}(z))Q_{k}(z,\rd z')=\rho(\rd z')\prod_{i=1}^{k-1}(1-\alpha_{i}(z'))$.
\end{rem}

\subsubsection{Guided random walk}

This algorithm was proposed in \citep{gustafson1998}. It is a special
case of Algorithm \ref{alg:discretetimePD-MCMC} which uses $\Phi\left(z\right)=(x+v\epsilon,v)$
for some $\epsilon>0$ and a proposal $M\left(z,\rd z'\right)=\delta_{\mathcal{S}\left(z\right)}\left(\rd z'\right)$
which is accepted with probability 1.%

\subsubsection{Hamiltonian Monte Carlo}

The celebrated HMC algorithm proposed in \citep{duane1987} is also
a special case of Algorithm \ref{alg:discretetimePD-MCMC} which uses
a proposal $M\left(z,\rd z'\right)=\delta_{\mathcal{S}\left(z\right)}\left(\rd z'\right)$.
However, contrary to guided random walk, it is using for $\Phi$ a
symplectic integrator targeting the Hamiltonian $H$. This deterministic
proposal satisfies indeed $\left|\nabla\Phi\right|=1$ and $\Phi^{-1}=\mathcal{S}\circ\Phi\circ\mathcal{S}$
(see, e.g., \citep{neal2011,leimkuhler2015}). The resulting PD-MCMC
kernel $K$ is usually combined with a momentum refreshment step $v\sim\psi$. 

\subsubsection{Reflective Slice Sampling: discrete-time BPS schemes}

Several versions of slice sampling, known as reflective slice sampling,
are based on bounces similar to the BPS and are also a special case
of Algorithm \ref{alg:discretetimePD-MCMC}; see \citep[Section 7]{neal2003}.
They rely $\Phi\left(z\right)=\left(x+v\epsilon,v\right)$ for some
$\epsilon>0$ and a deterministic proposal $M\left(z,\rd z'\right)=\delta_{\Psi\left(z\right)}\left(\rd z'\right)$.
Reflective slice sampling with inner reflections is using $\Psi\left(z\right)=\left(x^{*},v^{*}\right)=\left(x,R_{\nabla U}(x)v\right)$
while reflective slice sampling with outer reflections is using $\Psi\left(z\right)=\left(x^{*},v^{*}\right)=\left(x+v\epsilon+R_{\nabla U}(x+v\epsilon)v\epsilon,R_{\nabla U}(x+v\epsilon)v\right)$.
Both proposals satisfy $\Psi^{-1}=\mathcal{S}\circ\Psi\circ\mathcal{S}$.
The outer version of the algorithm has been recently proposed independently
in \citep{S_T_17}; see also \citep{S_12} for a related proposal
in the context of nested sampling. In either case, the acceptance
probability simplifies to 
\begin{align*}
\beta\left(x,v\right)=\min\left\{ 1,\frac{[\pi\left(x^{*}\right)-\pi(x^{*}-v^{*}\epsilon)]_{+}}{[\pi\left(x\right)-\pi(x+v\epsilon)]_{+}}\right\} .
\end{align*}

Intuitively, these algorithms can be interpreted as discrete-time
versions of the BPS process. Elementary calculations show indeed that
in both cases $\alpha\left(z\right)\rightarrow1-\epsilon\left\langle \nabla U\left(x\right),v\right\rangle _{+}$
and $\beta\left(z\right)\rightarrow1$ as $\epsilon\rightarrow0$\textbf{
}under regularity assumptions. We provide here a weak convergence
result for the resulting Markov chain where $\psi$ is the uniform
distribution on $\mathbb{S}^{d-1}$ to limit technicalities.
\begin{prop}
\label{prop:WeakCVBPS}Under regularity conditions, reflective slice
sampling with inner reflections converges weakly to the BPS for $\lambda_{\textrm{ref}}=0$
as $\epsilon\rightarrow0$. 
\end{prop}

A precise mathematical statement, Theorem~\ref{thm:weakconv}, and
its proof are given in Appendix~\ref{sec:Weak-convergence-of}. We
can modify this algorithm to include a refreshment, i.e.\ by sampling
$v'\sim\psi$ with probability $\lambda_{\textrm{ref}}\epsilon$.
This weak convergence result of Proposition \ref{prop:WeakCVBPS}
can be directly extended to this case to show that the resulting discrete-time
process converges weakly to the BPS process with refreshment rate
$\lambda_{\textrm{ref}}$. Note that the kernel $K$ would still be
$\rho$-invariant if $\Phi$ were using a computationally cheap approximation
$\nabla\widehat{U}$ of $\nabla U$ to bounce. However, this discrete-time
algorithm does not converge to the BPS process as the probability
of accepting $z^{\prime}=\mathcal{S}\left(z\right)$ does not vanish
as $\epsilon\rightarrow0$ in this scenario. Under regularity conditions,
it will instead converge towards the algorithm described at the end
of Section \ref{subsec:GMHalgorithm}.

\subsection{Extensions\label{subsec:Existing-PDMCMC-algorithms-DT-1}}

\subsubsection{Discrete-time BPS with randomized bounces\label{par:Discrete-time-PD-MCMC-with-randomizedbounces}}

As discussed in Section \ref{subsec:Randomizedbounces}, a variety
of randomized bounces has been proposed for continuous PD-MCMC. We
show here how to generalize these ideas to discrete-time. Let $\psi$
denote the standard normal distribution on $\mathbb{R}^{d}$, $\Phi\left(z\right)=\left(x+v\epsilon,v\right)$,
$\alpha\left(z\right)=\min\left\{ 1,\rho\left(\Phi\left(z\right)\right)/\rho\left(z\right)\right\} $
and $\mathcal{S}\left(z\right)=\left(x,-v\right)$ satisfying Assumptions
A\ref{hypA:Gassumptions-DT}.1 and A\ref{hypA:Gassumptions-DT}.2
and we select an event kernel of the form $Q\left(z,\rd z'\right)=\delta_{x}(\rd x')Q_{x}\left(v,\rd v'\right)$
based on a proposal $M_{x}\left(v,\rd v'\right)=M_{x}\left(v,v'\right)\rd v'$.
This leads to Algorithm \ref{alg:discretetimeBouncyPDMCMCrandomized}. 

\begin{algorithm}[H]
\caption{Discrete-time BPS with randomized bounces~\label{alg:discretetimeBouncyPDMCMCrandomized}}

\begin{enumerate}
\item With probability $\min\left\{ 1,\pi\left(x+v\epsilon\right)/\pi\left(x\right)\right\} $,
set $z'\leftarrow\left(x+v\epsilon,v\right)$. 
\item Otherwise 
\begin{enumerate}
\item Sample $v^{*}\sim M_{x}\left(v,\cdot\right)$.
\item With probability
\[
\min\left\{ 1,\frac{\psi\left(v^{*}\right)[\pi\left(x\right)-\pi(x-v^{*}\epsilon)]_{+}M_{x}\left(-v^{*},-v\right)}{\psi\left(v\right)[\pi\left(x\right)-\pi(x+v\epsilon)]_{+}M_{x}\left(v,v^{*}\right)}\right\} ,
\]
 set $z^{\prime}\leftarrow\left(x,v^{*}\right)$. 
\item Otherwise set $z^{\prime}\leftarrow(x,-v)$. 
\end{enumerate}
\end{enumerate}
\end{algorithm}

For the kernel $M_{x}\left(v,\cdot\right)$, we can use the randomized
bounces developed in Section \ref{subsec:Randomizedbounces} as well
as $M_{x}\left(v,\cdot\right)=\psi\left(\cdot\right).$ The forward-event
\citep{M_S_17}, generalized BPS \citep{W_R_17}, and autoregressive
bouncing procedures discussed in Section \ref{subsec:Randomizedbounces}
induce a transition kernel $M_{x}$ satisfying $\psi(v)\langle\nabla U(x),v\rangle_{+}M_{x}(v,v')=\psi(-v')\langle\nabla U(x),-v'\rangle_{+}M_{x}(-v',-v)$,
for which we would expect that the acceptance ratio in Step 2.b of
Algorithm \ref{alg:discretetimeBouncyPDMCMCrandomized} will be close
to 1 for small $\epsilon$.

The invariance with respect to $\rho$ of the transition kernel is
easy to check. Assumption A\ref{hypA:Gassumptions-DT}.1 is clearly
satisfied. Assumption A\ref{hypA:Gassumptions-DT}.2 follows from
direct calculations using $\left|\nabla\Phi\right|=1$ and $\Phi^{-1}=\mathcal{S}\circ\Phi\circ\mathcal{S}$.
Finally Assumption A\ref{hypA:Gassumptions-DT}.3 follows from the
fact that the event kernel corresponding to steps 2.a to 2.c of Algorithm
\ref{alg:discretetimeBouncyPDMCMCrandomized} is a GMH kernel with
$\nu\left(z\right)\propto\rho\left(z\right)\left(1-\alpha\left(z\right)\right)$
with a proposal kernel $M_{x}\left(v,\rd v'\right)$.

\subsubsection{Discrete-time Hamiltonian BPS}

We consider here the discrete-time version of the Hamiltonian BPS
proposed in Section \ref{subsec:Hamiltonian-PDMP}. This is achieved
by setting $\psi$ as the standard normal distribution on $\mathbb{R}^{d}$,
$\alpha\left(z\right)=\min\left\{ 1,\rho\left(\Phi\left(z\right)\right)/\rho\left(z\right)\right\} $
and $\mathcal{S}\left(z\right)=\left(x,-v\right)$. We also consider
an approximation $\hat{H}\left(z\right)$ defined in (\ref{eq:HamiltonianApprox})
of the Hamiltonian $H\left(z\right)$ and recall that $\widetilde{U}\left(x\right):=U\left(x\right)-V\left(x\right)$
and denote $\Psi\left(z\right)=(x,R_{\nabla\widetilde{U}}\left(x\right)v)$
. In Section \ref{subsec:Hamiltonian-PDMP}, we were considering for
$\Phi_{t}$ the exact Hamiltonian flow associated with $\hat{H}\left(z\right)$.
In discrete time we can select for $\Phi$ either this exact flow
$\Phi_{\epsilon}$ for some $\epsilon>0$ or a leapfrog integrator
with $L$ steps which we will denote $\Phi_{\mathrm{HD}}$. The crucial
difference is thus that it is not necessary to restrict ourselves
to a Hamiltonian $\hat{H}\left(z\right)$ for which the Hamiltonian
equations can be solved exactly. The resulting algorithm then proceeds
as follows.

\begin{algorithm}[H]
\caption{Discrete-time Hamiltonian BPS~\label{alg:discretetimeBouncyPDMCMC}}

\begin{enumerate}
\item With probability $\min\left\{ 1,\rho\left(\Phi_{\mathrm{HD}}\left(z\right)\right)/\rho\left(z\right)\right\} $,
set $z'\leftarrow\Phi_{\mathrm{HD}}\left(z\right)$. 
\item Otherwise 
\begin{enumerate}
\item With probability 
\begin{align*}
\hspace{-1cm}\min\left\{ 1,\frac{[\rho\left(x,-R_{\nabla\widetilde{U}}\left(x\right)v\right)-\rho\left(\Phi_{\mathrm{HD}}\left(x,-R_{\nabla\widetilde{U}}\left(x\right)v\right)\right)]_{+}}{[\rho\left(x,v\right)-\rho\left(\Phi_{\mathrm{HD}}\left(x,v\right)\right)]_{+}}\right\}  & =\min\left\{ 1,\frac{[\rho\left(x,v\right)-\rho\left(\Phi_{\mathrm{HD}}\left(x,-R_{\nabla\widetilde{U}}\left(x\right)v\right)\right)]_{+}}{[\rho\left(x,v\right)-\rho\left(\Phi_{\mathrm{HD}}\left(x,v\right)\right)]_{+}}\right\} ,
\end{align*}
set $z^{\prime}\leftarrow\left(x,R_{\nabla\widetilde{U}}\left(x\right)v\right)$. 
\item Otherwise set $z^{\prime}\leftarrow\left(x,-v\right)$. 
\end{enumerate}
\end{enumerate}
\end{algorithm}

The invariance with respect to $\rho$ of the transition kernel is
easy to check. Assumption A\ref{hypA:Gassumptions-DT}.1 is obviously
satisfied. Assumption A\ref{hypA:Gassumptions-DT}.2 follows from
direct calculations using $\left|\nabla\Phi\right|=1$ and $\Phi^{-1}=\mathcal{S}\circ\Phi\circ\mathcal{S}$.
Finally Assumption A\ref{hypA:Gassumptions-DT}.3 follows from the
fact that the event kernel corresponding to step (a) and (b) of Algorithm
\ref{alg:discretetimeBouncyPDMCMC} is a GMH kernel with $\nu\left(z\right)\propto\rho\left(z\right)\left(1-\alpha\left(z\right)\right)$
with a deterministic transition kernel satisfying $\Psi^{-1}=\mathcal{S}\circ\Psi\circ\mathcal{S}$.
If $\Phi$ is a leapfrog integrator of stepsize $\epsilon>0$ targeting
the Hamiltonian $H\left(z\right)$, then the strategy described above
is not directly applicable as $\widetilde{U}\left(x\right)=0$ for
all $x$ so $R_{\nabla\widetilde{U}}\left(x\right)$ is not defined.
However as $\Phi$ can be thought of as the exact time discretization
of a shadow Hamiltonian of the form $\hat{H}_{\epsilon}\left(z\right)=H\left(z\right)-\epsilon^{2}\widetilde{H}\left(z\right)+\mathcal{O}\left(\epsilon^{4}\right)$
\citep[p. 107]{leimkuhler2015}, it may be possible to build bounces
based on $\widetilde{H}\left(z\right)$ to correct for the discrepancy
between the true Hamiltonian dynamics and its leapfrog approximation. 

\subsubsection{Discrete-time gradient-free BPS }

The BPS-type algorithms given thus far all require computation of
the gradient of the potential $\nabla U(x)$ in order to update the
velocity $v$ when a bounce event occurs. However, we may wish to
target potential functions where this gradient cannot be computed
or is very expensive to compute. Additionally, the gradient may not
be informative in some models, such as certain embeddings of discrete
spaces where the gradient may be zero almost everywhere.

A scheme to approximate the gradient $\nabla U(x)$ by computing numerical
differences was advanced in \citep{S_T_17}. Here, some number $n_{cpt}$
of orthogonal unit vectors $\zeta_{i},i\in[n_{cpt}]$ are selected,
and the gradient approximated along each of these vectors by, e.g.,
\[
\Delta_{i}=\frac{U(x+h\zeta_{i})-U(x-h\zeta_{i})}{2h}
\]
for some small value $h$. The combination of these $n_{cpt}$ vectors
yields an approximation to the gradient

\[
\hat{g}=\sum_{i=1}^{n_{cpt}}\Delta_{i}\zeta_{i},
\]
which for $n_{cpt}=d$ is a typical numerical approximation to the
gradient. The new velocity is found by a reversible map from the old
velocity to the new velocity which preserves the magnitude of the
velocity and maintains the projection of the velocity on the gradient
vector. %

We may derive an algorithm which operates in the same spirit as that
of \citep{S_T_17}. By taking $n_{cpt}$ orthogonal unit vectors,
here selected randomly and independently of $v$, we can achieve a
reversible algorithm by simply taking the reflection off of the approximate
gradient
\[
v^{*}=v-2\frac{\langle\hat{g,}v\rangle}{|\hat{g}|^{2}}\hat{g,}
\]

and accepting this proposal in the same way we would accept a typical
bounce in the discrete-time BPS algorithm; specifically, by accepting
the bounce with probability

\[
\min\left\{ 1,\frac{[\pi\left(x\right)-\pi(x-v^{*}\epsilon)]_{+}}{[\pi\left(x\right)-\pi(x+v\epsilon)]_{+}}\right\} .
\]

Alternatively, we propose an algorithm which is related to the continuous-time
randomized bounces of Section \ref{subsec:Randomizedbounces}. We
had previously noted that the independent sampling algorithm proposed
in \citep{F_B_P_R_16} consists of sampling from the distribution
proportional to $\psi(v')\lambda(x,-v')$, independently of the current
value of $v$. Based on the discrete-time invariance condition (\ref{eq:Qpseudoinvariance-discrete}),
we may analogously sample from the distribution proportional to $\psi(v')\left[\pi(x)-\pi(x-v'\epsilon)\right]_{+}$.
This can be accomplished by using rejection sampling with instrumental
distribution $\psi$, noting that the ratio between the densities
is bounded above by $\pi(x)$; thus each rejection sampling proposal
$v^{\dagger}$ is accepted with probability $\left[\pi(x)-\pi(x-v^{\dagger}\epsilon)\right]_{+}/\pi(x)$,
and the first accepted proposal is also accepted as the new state
$v'$. See Algorithm \ref{alg:Discrete-time-gradient-free-BPS} for
details of this rejection-sampling scheme.

\begin{algorithm}

\caption{Discrete-time gradient-free BPS\label{alg:Discrete-time-gradient-free-BPS}}

\begin{enumerate}
\item With probability $\min\left\{ 1,\pi\left(x+v\epsilon\right)/\pi\left(x\right)\right\} $,
set $z'\leftarrow\left(x+v\epsilon,v\right)$. 
\item Otherwise 
\begin{enumerate}
\item Sample $v^{*}\sim\psi$.
\item With probability
\[
\frac{[\pi\left(x\right)-\pi(x-v^{*}\epsilon)]_{+}}{\pi(x)}
\]
 set $z^{\prime}\leftarrow\left(x,v^{*}\right)$. 
\item Otherwise go to Step 2.a.
\end{enumerate}
\end{enumerate}
\end{algorithm}

\subsubsection{Efficient Implementation of Discrete-time PD-MCMC\label{subsec:Efficient-Implementation-of-DT}}

All the implementations of discrete-time PD-MCMC schemes we are aware
of consist of simulating the algorithm using the kernel (\ref{eq:kerneldiscretetimePDMP}),
that is, at each%
{} time step it is checked whether an event occurs with probability
$1-\alpha\left(z\right)$ when in state $z$. However, it is possible
to improve over this implementation in some interesting scenarios.
Assume there exists $\bar{\alpha}:\mathcal{Z}\rightarrow\left[0,1\right]$
such that for $k\in\mathbb{N}$ we have $\alpha\left(\Phi^{k}\left(z\right)\right)\geq\bar{\alpha}\left(z,k\right)>0$
where $\bar{\alpha}\left(z,k\right)$ is computationally cheaper to
evaluate than $\alpha\left(\Phi^{k}\left(z\right)\right)$. It is
then possible to simulate an inter-event time of distribution (\ref{eq:sim_eventtimeDT})
by simulating a time from the instrumental distribution $\mathbb{\bar{P}}\left(\tau=j\right)=\left\{ 1-\bar{\alpha}\left(z,j\right)\right\} \prod_{i=0}^{j-1}\bar{\alpha}\left(z,i\right)$
which is then accepted with probability $\left\{ 1-\alpha\left(\Phi^{\tau}\left(z\right)\right)\right\} /\left\{ 1-\bar{\alpha}\left(z,\tau\right)\right\} $.
For a linear dynamics $\Phi\left(z\right)=\left(x+v\epsilon,v\right)$,
we can obtain such bounds by upper bounding the derivative of $t\mapsto U\left(x+vt\right)$.

If $\alpha\left(z\right)=\min\left\{ 1,\rho\left(\Phi\left(z\right)\right)/\rho\left(z\right)\right\} $,
we can also always use for example the lower bound $\bar{\alpha}\left(z,k\right)=\prod_{i=1}^{n}\bar{\alpha}_{i}\left(z,k\right)$
where $\bar{\alpha}_{i}\left(z,k\right)=\min\left\{ 1,\rho_{i}\left(\Phi^{k+1}\left(z\right)\right)/\rho_{i}\left(\Phi^{k}\left(z\right)\right)\right\} $
for $\rho\left(z\right)=\prod_{i=1}^{n}\rho_{i}\left(z\right)$. It
has the potential advantage that simulating an event of probability
$\bar{\alpha}\left(z,k\right)$ can be performed in parallel by simulating
independent Bernoulli random variables $B_{i}\sim\mathrm{Ber}(1-\bar{\alpha}_{i}(z,k))$
for $i\in[n]$.

Finally there are scenarios where it is possible to directly simulate
an event time from (\ref{eq:sim_eventtimeDT}). For example, assume
that $\pi\left(x\right)=\exp(-U\left(x\right))$ where $U$ is strictly
convex, $\Phi\left(z\right)=(x+v\epsilon,v)$ and $\alpha\left(z\right)=\min\left\{ 1,\rho\left(\Phi\left(z\right)\right)/\rho\left(z\right)\right\} =\min\left\{ 1,\exp\left(-\left(U\left(x+v\epsilon\right)-U\left(x\right)\right)\right)\right\} $
then it is easy to show that Algorithm \ref{alg:exactsimulationeventtimeDT}
returns a sample from (\ref{eq:sim_eventtimeDT}). This adapts the
approaches developed in \citep[Section 2.3.1]{BC_D_V_15} for the
continuous-time BPS algorithm to the discrete-time case. 

\begin{algorithm}[H]
\caption{Simulation inter-event time for discrete-time BPS~for strictly log-concave
targets\label{alg:exactsimulationeventtimeDT}}

\begin{enumerate}
\item Minimize the potential along the continuous trajectory 
\[
t^{*}=\arg\min\left\{ U\left(x+vt\right):t\in\mathbb{R^{\mathrm{+}}}\right\} .
\]
\item Set 
\[
k^{*}=\arg\min\left\{ U\left(x+vk\epsilon\right):k\in\left\{ \left\lfloor t^{*}/\epsilon\right\rfloor ,\left\lceil t^{*}/\epsilon\right\rceil \right\} \right\} .
\]
\item Solve for $t\geq t^{*}$
\[
U\left(x+vt\right)-U\left(x+vk^{*}\epsilon\right)=E,\text{\hspace{1cm}}E\sim\mathcal{E\mathrm{xp}}\left[0,1\right].
\]
\item Return $\tau=\left\lfloor t/\epsilon\right\rfloor $. %
\end{enumerate}
\end{algorithm}

All these strategies can be easily combined. For example, we can use
an upper bound $\bar{\alpha}\left(z,k\right)=\prod_{i=1}^{n}\bar{\alpha}_{i}\left(z,k\right)$
where $\rho_{i}\left(z\right)$ is strictly log-concave for some $i\in[n]$.

\section{Discrete-time local PD-MCMC\label{sec:Discrete-time-LocalPDMCMC}}

\subsection{Algorithm description}

Given the framework provided in Section \ref{subsec:Sufficient-conditions-for-Local-DT},
it is not difficult to obtain discrete-time local PD-MCMC schemes
for $\rho\left(z\right)=\exp(-\sum_{i=1}^{n}H_{i}\left(z\right))=\pi\left(x\right)\psi\left(v\right)=\exp(-U\left(x\right))\psi\left(v\right)$
on $\mathcal{Z}=\mathbb{R}^{d}\times\mathbb{R}^{d}$ where $\pi$
is the target distribution of interest with $\psi$ is a multivariate
normal. We can for example select a dynamics, involution and acceptance
probability satisfying $\left|\nabla\Phi\right|=1$, $\alpha_{i}\left(z\right)=\min\left\{ 1,\rho_{i}\left(\Phi\left(z\right)\right)/\rho_{i}\left(z\right)\right\} $
with $\rho_{i}\left(z\right)=\exp(-H_{i}\left(z\right))$, $\mathcal{S}\left(z\right)=\left(x,-v\right)$,
$\Phi^{-1}=\mathcal{S}\circ\Phi\circ\mathcal{S}$ and $\rho\circ\mathcal{S}=\rho$.
A rather generic local PD-MCMC scheme is presented in Algorithm \ref{alg:discretetimelocalPDMCMC}.

\begin{algorithm}[H]
\caption{Discrete-time local PD-MCMC~\label{alg:discretetimelocalPDMCMC}}

\begin{enumerate}
\item For $i\in[n]$, sample $B_{i}\sim\mathrm{Ber}\left\{ \left[\rho_{i}\left(z\right)-\rho_{i}\left(\Phi\left(z\right)\right)\right]_{+}/\rho_{i}\left(z\right)\right\} $.
\item If $B_{i}=0$ for all $i\in[n]$, set $z'\leftarrow\Phi\left(z\right)$. 
\item Otherwise, sample $z^{*}\sim M_{B}\left(z,\cdot\right).$
\item With probability 
\begin{align}
\min\left\{ 1,\frac{M_{B}\left(\mathcal{S}\left(z^{*}\right),\mathcal{S}\left(z\right)\right)}{M_{B}\left(z,z^{*}\right)}\prod_{i=1}^{n}\frac{\rho_{i}\left(\mathcal{S}\left(z^{*}\right)\right)\mathrm{Ber}\left(B_{i};1-\alpha_{i}\left(\mathcal{S}\left(z^{*}\right)\right)\right)}{\rho_{i}\left(z\right)\mathrm{Ber}\left(B_{i};1-\alpha_{i}\left(z\right)\right)}\right\} .\label{eq:acceptanceprobatotal}
\end{align}
set $z'\leftarrow z^{*}$. Otherwise, set $z'\leftarrow(x,-v)$.
\end{enumerate}
\end{algorithm}

Here Steps 3 and 4 of Algorithm \ref{alg:discretetimelocalPDMCMC}
corresponds to a GMH kernel satisfying the skewed-detailed balance
condition (\ref{eq:skeweddetailedbalanceconditiongeneral}) for $\nu_{b}\left(\rd z\right)\propto\rho\left(\rd z\right)\left(1-\alpha\left(z\right)\right)\mathbb{Q}_{|B|\geq1}\left(b|z\right)$
and a proposal $M_{B}\left(z,\rd z'\right)$ for any $b\in\mathcal{B}$. 

Consider a special case of Algorithm \ref{alg:discretetimelocalPDMCMC}
given in Algorithm \ref{alg:discretetimelocaBPS_DT} which corresponds
to a discrete-time version of local BPS. It is using $\Phi\left(z\right)=\left(x+v\epsilon,v\right)$,
$\mathcal{S}(z)=\left(x,-v\right)$ and a deterministic proposal $M_{b}\left(z,\rd z'\right)=\delta_{\Psi_{b}\left(z\right)}\left(\rd z'\right)$
satisfying $\Psi_{b}^{-1}=\mathcal{S}\circ\Psi_{b}\circ\mathcal{S}$.
We also use $\rho_{i}\left(z\right)=\exp(-U_{i}\left(x\right)):=\pi_{i}\left(x\right)$
so that $U\left(x\right)=\sum_{i=1}^{m}U_{i}\left(x\right)$ and $\rho_{n}\left(z\right)=\psi\left(v\right)$
with $n=m+1$. We could have selected $\alpha_{n}\left(z\right)=\alpha_{\mathrm{ref}}$
to refresh the velocity periodically but we omit it for ease of presentation.
The only difference with Algorithm \ref{alg:discretetimelocalPDMCMC}
is that we actually use here an alternative acceptance probability
which is lower than (\ref{eq:acceptanceprobatotal}) but has the advantages
that it factorizes across $i$. It will prove useful as it is then
possible to simulate an event with the required acceptance probability
by simulating independent events in parallel.

\begin{algorithm}[H]
\caption{Discrete-time local BPS \label{alg:discretetimelocaBPS_DT}}

\begin{enumerate}
\item For $i\in[m]$, sample $B_{i}\sim\mathrm{Ber}\left\{ \left[\pi_{i}\left(x\right)-\pi_{i}\left(x+v\epsilon\right)\right]_{+}/\pi_{i}\left(x\right)\right\} $.
\item If $B_{i}=0$ for all $i\in[n]$, set $z'\leftarrow\left(x+v\epsilon,v\right)$. 
\item Otherwise, 
\begin{enumerate}
\item Set $z^{*}\gets\Psi_{B}\left(z\right):=\left(x,v^{*}\right),$ where
$v^{*}\gets R_{\nabla\overline{U}}(x)v$ with $\nabla\overline{U}\left(x\right):=\sum_{i:B_{i}=1}\nabla U_{i}\left(x\right)$.
\item With probability 
\begin{align}
 & \prod_{i=1}^{m}\min\left\{ 1,\frac{\rho_{i}\left(\mathcal{S}\circ\Psi_{B}\left(z\right)\right)\mathrm{Ber}\left(B_{i};1-\alpha_{i}\left(\mathcal{S}\circ\Psi_{B}\left(z\right)\right)\right)}{\rho_{i}\left(z\right)\mathrm{Ber}\left(B_{i};1-\alpha_{i}\left(z\right)\right)}\right\} \nonumber \\
= & \prod_{i:B_{i}=0}\min\left\{ 1,\frac{\min\left(\pi_{i}(x),\pi_{i}(x-v^{*}\epsilon)\right)}{\min\left(\pi_{i}(x),\pi_{i}(x+v\epsilon)\right)}\right\} \prod_{i:B_{i}=1}\min\left\{ 1,\frac{[\pi_{i}(x)-\pi_{i}(x-v^{*}\epsilon)]_{+}}{[\pi_{i}(x)-\pi_{i}(x+v\epsilon)]_{+}}\right\} ,\label{eq:acceptanceprobalocalBP}
\end{align}
set $z'\gets\Psi_{B}\left(z\right)$.
\item Otherwise, set $z'\gets\left(x,-v\right)$.
\end{enumerate}
\end{enumerate}
\end{algorithm}

Note that $\nabla\overline{U}\left(x\right)$ depends on both $u$,
$v$ and $\epsilon$, we stress this dependence as it is omitted notationally.

Algorithms \ref{alg:discretetimelocalPDMCMC} and \ref{alg:discretetimelocaBPS_DT}
might appear of limited interest as they require to sample $n$ Bernoulli
random variables at each iteration. In the next sections, we show
how we can propose implementations that parallel the priority queue
implementation of the local BPS proposed in \citep{P_dW_12}, see
\citep[Section 3.3.1]{BC_D_V_15} for a detailed description, as well
as the subsampling algorithms proposed in \citep[Section 3.3.2]{BC_D_V_15,B_F_R_16,K_K_16}.

\subsection{Prefetching implementation\label{subsec:prefetchingDT}}

We first describe a priority queue type implementation of Algorithm
\ref{alg:discretetimelocaBPS_DT} based on parallel prefetching ideas
\citep{Br_2006,A_K_W_S_A_14} in scenarios where 
\[
U\left(x\right)=\sum_{i=1}^{m}U_{i}\left(x_{S_{i}}\right),
\]
$x_{S_{i}}$ being a subset of the components of $x$ and $\pi_{i}(x)=\exp\left(-U_{i}\left(x_{S_{i}}\right)\right)$.
There are many possible variations of this implementation.

\begin{algorithm}[H]
\caption{Discrete-time local BPS implementation via parallel prefetching ~\label{alg:discretetimelocalBPSprefetching}}

\begin{enumerate}
\item Initialization
\begin{enumerate}
\item For $i\in[m]$, sample non-negative event times $\tau_{i}$ with distribution
\[
\max\left(0,1-\frac{\pi_{i}(x+v(\tau_{i}+1)\epsilon)}{\pi_{i}(x+v\tau_{i}\epsilon)}\right)\prod_{k=0}^{\tau_{i}-1}\min\left(1,\frac{\pi_{i}(x+v(k+1)\epsilon)}{\pi_{i}(x+vk\epsilon)}\right).
\]
\end{enumerate}
\item Iteration $t,$ $t\geq1$
\begin{enumerate}
\item If $\min\tau_{i}>0$, then set $z'\leftarrow\left(x+\epsilon v,v\right)$.
Update $\tau_{i}\gets\tau_{i}-1$.
\item Otherwise,
\begin{enumerate}
\item Compute 
\begin{equation}
\nabla\overline{U}\left(x\right):=\sum_{i:\tau_{i}=0}\nabla U_{i}\left(x_{S_{i}}\right),\label{eq:graduestimate1-1-1-1}
\end{equation}
and let $v^{*}\gets R_{\nabla\overline{U}}\left(x\right)v$. 
\item With probability 
\begin{equation}
\prod_{i:\tau_{i}>0}\min\left\{ 1,\frac{\min\left(\pi_{i}(x),\pi_{i}(x-v^{*}\epsilon)\right)}{\min\left(\pi_{i}(x),\pi_{i}(x+v\epsilon)\right)}\right\} \prod_{i:\tau_{i}=0}\min\left\{ 1,\frac{[\pi_{i}(x)-\pi_{i}(x-v^{*}\epsilon)]_{+}}{[\pi_{i}(x)-\pi_{i}(x+v\epsilon)]_{+}}\right\} ,\label{eq:prefretching_acceptance_probability}
\end{equation}
set $z'\gets(x,v^{*})$. Sample again $\tau_{i}$ for all $i$ where
$v_{j}^{*}\neq v_{j}$ for some $j\in S_{i}$.
\item Otherwise set $z'\leftarrow\left(x,-v\right)$. Sample $\tau_{i}$
for all $i$.
\end{enumerate}
\end{enumerate}
\end{enumerate}
\end{algorithm}

The efficiency of Algorithm \ref{alg:discretetimelocalBPSprefetching}
relies on the capability of computing the $\tau_{i}$ efficiently.
This may be possible when, for example, this is done in parallel or
when we some property of $\pi_{i}$ allows it, such as in the case
of log-concave targets detailed as in Algorithm \ref{alg:exactsimulationeventtimeDT}
given above.

\subsection{Subsampling implementations\label{subsec:subsamplingDT}}

For sufficiently small $\epsilon$, we might expect that in Step 1
of Algorithm \ref{alg:discretetimelocaBPS_DT} would yield very few
indices for which $B_{i}=1$. This motivates an approach which can
sample these variables more efficiently by finding an upper bound
on the probability that $B_{i}=1$, essentially allowing us to bound
the number of indices for which $B_{i}=1$. We present Algorithm \ref{alg:discretetimelocalBPSthinningGeneric};
here, the acceptance of the bounce move (\ref{eq:acceptanceprobalocalBP})
is computed in two stages: in Step 4.b we simulate events of probability
$1-\min\left\{ 1,\frac{[\pi_{i}(x)-\pi_{i}(x-v^{*}\epsilon)]_{+}}{[\pi_{i}(x)-\pi_{i}(x+v\epsilon)]_{+}}\right\} $
for each $i$ where $B_{i}=1$, if these succeed then in Step 4.c
we simulate events of probability $1-\min\left\{ 1,\frac{\min\left(\pi_{i}(x),\pi_{i}(x-v^{*}\epsilon)\right)}{\min\left(\pi_{i}(x),\pi_{i}(x+v\epsilon)\right)}\right\} $
for each $i$ where $B_{i}=0$. We suggest that one can make use of
efficient procedures described in Algorithm \ref{alg:discretetimelocalBPSthinningBernoulli}
and Algorithm \ref{alg:discretetimelocalBPSthinningPoisson} to sample
multiple Bernoulli random variables in both Steps 1 and 4.c; in both
cases we expect few cases where the respective Bernoulli variables
are 1. While Step 4.b also samples a set of Bernoulli variables, our
assumption that $\epsilon$ is small suggests that the number of variables
sampled here will be small; as such this step may be inexpensive and
there is likely little to be gained by a more sophisticated simulation
scheme.

\begin{algorithm}[H]
\caption{Discrete-time local BPS implementation via Binomial sampling ~\label{alg:discretetimelocalBPSthinningGeneric}}

\begin{enumerate}
\item For $i\in[m]$, sample $B_{i}\sim\mathrm{Ber}\left\{ \left[\pi_{i}\left(x\right)-\pi_{i}\left(x+v\epsilon\right)\right]_{+}/\pi_{i}\left(x\right)\right\} $. 
\item Set $V\leftarrow\{i\in[m]:B_{i}=1\}$.
\item If $V=\emptyset$, then set $z'\leftarrow\left(x+v\epsilon,v\right)$.
\item If $V\neq\emptyset$, then
\begin{enumerate}
\item Compute
\begin{equation}
\nabla\overline{U}\left(x\right):=\sum_{i\in V}\nabla U_{i}\left(x_{S_{i}}\right)\label{eq:dbps_V_gradient-1}
\end{equation}
and let $v^{*}\gets R_{\nabla\overline{U}}(x)v$.
\item For $i\in V$, sample $B'_{i}\sim\mathrm{Ber}\left(1-\min\left\{ 1,\frac{[\pi_{i}(x)-\pi_{i}(x-v^{*}\epsilon)]_{+}}{[\pi_{i}(x)-\pi_{i}(x+v\epsilon)]_{+}}\right\} \right)$.%
\item For $i\in[m]\setminus V$, sample $B'_{i}\sim\mathrm{Ber}\left(1-\min\left\{ 1,\frac{\min\left(\pi_{i}(x),\pi_{i}(x-v^{*}\epsilon)\right)}{\min\left(\pi_{i}(x),\pi_{i}(x+v\epsilon)\right)}\right\} \right)$.
\item If $B'_{i}=1$ for any $i\in[m]$ then set $z'\leftarrow\left(x,-v\right)$
and otherwise set $z'\leftarrow\left(x,v^{*}\right)$.
\end{enumerate}
\end{enumerate}
\end{algorithm}

We suggest two possible alternatives for efficiently sampling a set
of Bernoulli variables. Here, use the notation $X_{i}\sim\mathrm{Ber}(p_{i})$
for all $i\in I$ to emphasize that these are general schemes not
necessarily associated with sampling either $B_{i}$ or $B'_{i}$.
First, consider the scenario where one has some uniform control over
the probability that $X_{i}=1$, that is we assume that there exists
$0\leq\bar{p}\leq1$ such that for all $i$
\[
\mathbb{P}\left(X_{i}=1\right):=p_{i}\leq\bar{p}.
\]
In this case, we can determine the set $\left\{ i:X_{i}=1\right\} $
using Algorithm \ref{alg:discretetimelocalBPSthinningBernoulli}.
This incurs a computational complexity $O(1+|I|\bar{p})$ compared
to $O(|I|)$ for the direct implementation \citep{Hormann2003}.%
{} This implementation can be thought of as the discrete-time version
of the thinning ideas leading to the ``naive'' subsampling techniques
presented in \citep{BC_D_V_15,B_F_R_16,B_BC_D_D_F_R_J_16}. 

\begin{algorithm}[H]
\caption{Efficient sampling of Bernoulli variables via Binomial sampling ~\label{alg:discretetimelocalBPSthinningBernoulli}}

Given a set of indices $I$, associated Bernoulli probabilities $\left\{ p_{i};i\in I\right\} $,
and bound $p_{i}\leq\bar{p}$,
\begin{enumerate}
\item Sample $S\sim\mathrm{Bin}\left(|I|,\bar{p}\right)$.
\item Sample $S$ indices $i_{1},...,i_{S}$ in $I$ uniformly at random
without replacement and denote $\mathcal{S}=\left(i_{1},...,i_{S}\right)$.
\item For $i\in\mathcal{S}$, sample $X_{i}\sim\mathrm{Ber}\left(p_{i}/\bar{p}\right).$
\item For $i\in I\setminus\mathcal{S}$, set $X_{i}\gets0$.
\end{enumerate}
\end{algorithm}

Second, if we instead have access to local bounds $0\leq\bar{p}_{i}\leq1$
such that 

\[
\mathbb{P}\left(X_{i}=1\right):=p_{i}\leq\bar{p}_{i},
\]

we could obviously use the previous strategy by setting $\bar{p}:=\max_{i\in I}\bar{p}_{i}$
but this strategy can be highly inefficient if, e.g., most bounds
$\bar{p}_{i}$ are very close to zero and a few are close to 1. In
this scenario, it is possible to use instead Algorithm \ref{alg:discretetimelocalBPSthinningPoisson}
which relies on the simulation of Poisson random variables. This algorithm
can be thought of as the discrete-time version of the thinning ideas
leading to the ``informed'' subsampling techniques presented in
\citep{BC_D_V_15,K_K_16}.

\begin{algorithm}[H]
\caption{Efficient sampling of Bernoulli variables via Poisson sampling ~\label{alg:discretetimelocalBPSthinningPoisson}}

Given a set of indices $I$, associated Bernoulli probabilities $\left\{ p_{i};i\in I\right\} $,
and local bounds $p_{i}\leq\bar{p}_{i}$,
\begin{enumerate}
\item Sample $S\sim\mathrm{Poi}\left(\kappa\right)$ where $\kappa=\sum_{i\in I}\kappa_{i}$
with $\kappa_{i}=-\log\left(1-\overline{p}_{i}\right)$.
\item Sample $\left\{ N_{i};i\in I\right\} $ from the multinomial distribution
of parameters $\left(S,\left\{ \frac{\kappa{}_{i}}{\kappa};i\in I\right\} \right)$
and denote $\mathcal{S}=\left\{ i:N_{i}\geq1\right\} $.
\item For $i\in\mathcal{S}$, sample $X_{i}\sim\mathrm{Ber}\left(p_{i}/\overline{p}_{i}\right)$.
\item For $i\in I\setminus\mathcal{S}$, set $X_{i}\gets0$.
\end{enumerate}
\end{algorithm}

For this algorithm to be of practical interest, the bounds $\overline{p}_{i}$
and the associated Poisson rates $\kappa_{i}$ should not have to
be recomputed at each time step as for the examples considered in
\citep{BC_D_V_15,K_K_16}. In this scenario, it is then possible to
use the alias method or ordered marginally uniform random variables
on $\left[0,1\right]$ to sample efficiently from the multinomial
distributions in complexity $O(S)$ \citep{Hormann2003}.%

The availability of an upper bound for Step 1, denoted here $\bar{p}_{i}(x,v)$,
can be seen as equivalent to a lower bound on $\alpha_{i}(z)$ as
discussed in Section \ref{subsec:Sufficient-conditions-for-Local-DT},
since

\[
\bar{p}_{i}(x,v)\geq\left[\pi_{i}\left(x\right)-\pi_{i}\left(x+v\epsilon\right)\right]_{+}/\pi_{i}\left(x\right)=1-\alpha_{i}(x,v).
\]

For Step 4.c, we would seek an upper bound 
\begin{align*}
\bar{p}'_{i}(x,v,v^{*}) & \geq1-\min\left\{ 1,\frac{\min\left(\pi_{i}(x),\pi_{i}(x-v^{*}\epsilon)\right)}{\min\left(\pi_{i}(x),\pi_{i}(x+v\epsilon)\right)}\right\} .
\end{align*}

This bound may be achieved, for example, when $|\nabla U_{i}(x')|<\delta$
for all $\left\{ x':|x'-x|<\epsilon|v^{*}|\right\} $. In this case,
an upper bound can be derived using 
\[
\frac{\min\left(\pi_{i}(x),\pi_{i}(x-v^{*}\epsilon)\right)}{\min\left(\pi_{i}(x),\pi_{i}(x+v\epsilon)\right)}\geq\min\left(1,\pi_{i}(x-v^{*}\epsilon)/\pi_{i}(x)\right)>1-\delta|v^{\ast}|\epsilon.
\]

\section{Discrete-time doubly stochastic PD-MCMC\label{sec:Discrete-time-doublystochasticPDMCMC}}

\subsection{Algorithm description}

By using the framework provided in Section \ref{subsec:Sufficient-conditions-for-Stochastic-DT},
we can obtain discrete-time stochastic PD-MCMC schemes for $\rho\left(z\right)=\exp(-\int H_{\omega}\left(z\right)\mu(\rd\omega))=\pi\left(x\right)\psi\left(v\right)$
on $\mathcal{Z}=\mathbb{R}^{d}\times\mathbb{R}^{d}$ where $\pi$
is the target distribution of interest with $\psi$ is a multivariate
normal. We will write $H_{\omega}\left(z\right)=U_{\omega}\left(x\right)+\frac{1}{2}v^{T}v$.
We can for example select a dynamics and an involution satisfying
$\left|\nabla\Phi\right|=1$, $\Phi^{-1}=\mathcal{S}\circ\Phi\circ\mathcal{S}$
and $\rho\circ\mathcal{S}=\rho$. A rather generic doubly-stochastic
PD-MCMC scheme for such dynamics is presented in Algorithm \ref{alg:discrete-time-doubly-stochastic-MCMC}.

The kernel $Q_{P}$, which must satisfy (\ref{eq:conditionaldistributionatleastaneventdoublystochasticlDT}),
may be implemented using a scheme similar to the GMH. Using standard
results on Poisson point processes and Assumption A\ref{hypA:Sassumptions-DT}.3,
the condition (\ref{eq:QlocalpseudoinvarianceDT-1}) can be simplified
as
\begin{align*}
\int\rho(\rd z)\exp\left\{ \int\left\{ \log\alpha_{\omega}(z)-\log\alpha_{\omega}(\mathcal{S}(z'))\right\} \mu(\rd\omega)\right\} \frac{\prod_{\omega\in P}\log\alpha_{\omega}(z)}{\prod_{\omega\in P}\log\alpha_{\omega}(\mathcal{S}(z'))}Q_{P}(z,\rd z') & =\rho(\mathcal{S}(\rd z')),
\end{align*}

suggesting a GMH kernel with deterministic proposal $\Psi_{P}(z)$
satisfying $\Psi_{P}^{-1}=\mathcal{S}\circ\Psi_{P}\circ\mathcal{S}$
and acceptance probability
\begin{align}
\beta(z,P)= & \exp\left(-\int\left[\log\alpha_{\omega}(z)-H_{\omega}(z)-\log\alpha_{\omega}(\mathcal{S}\circ\Psi_{P}(z))+H_{\omega}(\mathcal{S}\circ\Psi_{P}(z))\right]_{+}\mu(\rd\omega)\right)\label{eq:bounce-acceptance-doubly-stochastic-bps}\\
 & \times\min\left\{ 1,\prod_{\omega\in P}\frac{\log\alpha_{\omega}(\mathcal{S}\circ\Psi_{P}(z))}{\log\alpha_{\omega}(z)}\right\} \nonumber 
\end{align}

which arises by treating the integral terms and the product terms
as two factors, each with its own acceptance probability. Based on
this, we present Algorithm \ref{alg:discrete-time-doubly-stochastic-MCMC},
wherein we sample an event of probability (\ref{eq:bounce-acceptance-doubly-stochastic-bps})
using a two-stage acceptance procedure.

\begin{algorithm}
\caption{Discrete-time doubly stochastic PD-MCMC\label{alg:discrete-time-doubly-stochastic-MCMC}}

\begin{enumerate}
\item Sample a Poisson process $P$ with rate $-\log\alpha_{\omega}(z)\mu(\rd\omega)$.
\item If $P=\emptyset$, then set $z'\gets\Phi(z)$.
\item If $P\neq\emptyset$,
\begin{enumerate}
\item Sample a Poisson process $P'$ with rate $\left[\log\alpha_{\omega}(z)-H_{\omega}(z)-\log\alpha_{\omega}(\mathcal{S}\circ\Psi_{P}(z))+H_{\omega}(\mathcal{S}\circ\Psi_{P}(z))\right]_{+}\mu(\rd\omega)$.
\item If $P'=\emptyset$, then set $z'\gets\Psi_{P}(z)$ with probability
$\min\left\{ 1,\prod_{\omega\in P}\frac{\log\alpha_{\omega}(\mathcal{S}\circ\Psi_{P}(z))}{\log\alpha_{\omega}(z)}\right\} $.
\item Otherwise set $z'\gets\mathcal{S}(z)$.
\end{enumerate}
\end{enumerate}
\end{algorithm}

By selecting $\alpha_{\omega}(z)=\min\left(1,\frac{\rho_{\omega}(\Phi(z))}{\rho_{\omega}(z)}\right)=\exp\left(-\left[H_{\omega}(\Phi(z))-H_{\omega}(z)\right]_{+}\right)$
for $\rho_{\omega}\left(z\right)=\exp(-H_{\omega}\left(z\right))$,
the acceptance probability (\ref{eq:bounce-acceptance-doubly-stochastic-bps})
takes the form
\begin{align*}
\beta(z,P)= & \exp\left(-\int\left\{ \max\left[H_{\omega}(\mathcal{S}\circ\Psi_{P}(z)),H(\Phi\circ\mathcal{S}\circ\Psi_{P}(z))\right]-\max\left[H_{\omega}(z),H_{\omega}(\Phi(z))\right]\right\} _{+}\mu(\rd\omega)\right)\\
 & \times\min\left\{ 1,\prod_{\omega\in P}\frac{\left[H_{\omega}(\Phi\circ\mathcal{S}\circ\Psi_{P}(z))-H_{\omega}(\mathcal{S}\circ\Psi_{P}(z))\right]_{+}}{\left[H_{\omega}(\Phi(z))-H_{\omega}(z)\right]_{+}}\right\} .
\end{align*}
Further allowing $\mathcal{S}(z)=(x,-v)$, $\Phi(z)=(x+v\epsilon,v)$
and $\Psi_{P}(z)=(x,R_{\nabla\overline{U}}(x)v)$ with $\nabla\overline{U}(x)=\sum_{\omega\in P}\nabla U_{\omega}(x)$
yields

\begin{align*}
\beta(z,P) & =\exp\left(-\int\left\{ \left[U_{\omega}(x-R_{\nabla\overline{U}}(x)v\epsilon)-U_{\omega}(x)\right]_{+}-\left[U_{\omega}(x+v\epsilon)-U_{\omega}(x)\right]_{+}\right\} _{+}\mu(\rd\omega)\right)\\
 & \times\min\left\{ 1,\prod_{\omega\in P}\frac{\left[U_{\omega}(x-R_{\nabla\overline{U}}(x)v\epsilon)-U_{\omega}(x)\right]_{+}}{\left[U_{\omega}(x+v\epsilon)-U_{\omega}(x)\right]_{+}}\right\} .
\end{align*}

The first term of this acceptance ratio, viewed as a void probability
of a Poisson process, can be interpreted as the ``excess'' rate
of $\alpha(x,-R_{\nabla\overline{U}}(x)v)$ over $\alpha(x,v)$; in
other words, the probability that no extra points would be simulated
for $P$ when in state $(x,-R_{\nabla\overline{U}}(x)v)$.

In either case, the simulation of Poisson processes $P$ and $P'$
is possible when those rates can be bounded. If we have some lower
bound $\underline{\alpha_{\omega}}(z)\leq\alpha_{\omega}(z)$ for
which we can simulate a Poisson process of intensity $-\log\underline{\alpha_{\omega}}(z)\mu(\rd\omega)$,
then we can recover $P$ by thinning this process. This condition
is sufficient for simulation of $P'$ as the corresponding intensity
is bounded by $-\log\underline{\alpha_{\omega}}(\mathcal{S}\circ\Psi(z))$;
however, it may be possible to bound the intensity of $P'$ more tightly
in some situations.

The idea of introducing a Poisson process so as to deal with the intractability
of target distribution can also be exploited within a standard MCMC
setting. For simplicity, assume a symmetric proposal density $q\left(\left.z'\right|z\right)$
then it is easy to check that Algorithm \ref{alg:discretetimeStochaticMH}
corresponds to a transition kernel which is reversible with respect
to $\rho\left(z\right)=\exp(-\int H_{\omega}\left(z\right)\mu\left(\rd\omega\right))$.

\begin{algorithm}[H]
\caption{Noisy Metropolis\textendash Hastings using unbiased estimator of the
log-target ~\label{alg:discretetimeStochaticMH}}

\begin{enumerate}
\item Sample $z^{*}\sim q\left(\left.\cdot\right|z\right).$
\item Sample a Poisson process $P$ on $\Omega$ with rate $\left[H_{\omega}\left(z^{*}\right)-H_{\omega}\left(z\right)\right]_{+}\mu\left(\rd\omega\right)$.
\item If $P=\emptyset$, then set $z'\leftarrow z^{*}.$
\item Otherwise set $z'\leftarrow z.$
\end{enumerate}
\end{algorithm}

\subsection{For measures containing atoms}

In the previous section, we assumed that the measure $\mu$ was non-atomic.
Here we consider the case where $\mu$ may contain atoms; this extension
allows us to view the local algorithms as a special case of the doubly-stochastic
algorithm where $\Omega=[n]$. To avoid any issues that may arise
due to indistinguishable points, we simulate here a Poisson process
$P^{*}$ on $\Omega\times\mathbb{R}$ with rate $\mathbb{I}(0<y<-\log\alpha_{\omega}(z))\mu(\rd\omega)\mathrm{Leb}(\rd y)$,
which projected onto $\Omega$ is equivalent to the rate we used in
the non-atomic case. Whereas in the non-atomic case we would take
$\nabla\overline{U}(x)=\sum_{\omega\in P^{*}}\nabla U_{\omega}(x)$,
we propose to here instead take $\nabla\overline{U}^{*}(x)=\sum_{\omega\in P_{\omega}}\nabla U_{\omega}(x),$
where $P_{\omega}$ denotes the set of unique values of $\omega$
among the points in $P^{*}$. We define the projection $\upsilon(P^{*})=P_{\omega}$.
Denote the corresponding bounce proposal $\Psi_{P_{\omega}}(x,v)=(x,R_{\nabla\overline{U}^{*}}(x)v)$.

While it remains sufficient to use the acceptance probability (\ref{eq:bounce-acceptance-doubly-stochastic-bps}),
we note that a partition of $P^{*}$ into sets of equivalent $P_{\omega}$
(and therefore equivalent bounce proposals $\Psi_{P_{\omega}}(z)$)
will yield a sufficient condition which is ``integrated out'' in
the sense that the total density of the forward and reverse transitions
are captured.

Allow $\Omega^{*}$ to represent the set of atoms in $\Omega$. The
probability of an atom $\omega^{*}\in\Omega^{*}$ being absent in
the projected Poisson process $P_{\omega}$ is $\exp\left(-\int_{0}^{-\log\alpha_{\omega^{*}}(z)\mu(\{\omega^{*}\})}\mathrm{Leb}(\rd y)\right)=\alpha_{\omega^{*}}(z)\mu(\{\omega^{*}\})$.
From this, we can see that the void probability of $P^{*}$ (and equivalently
the void probability of $P_{\omega}$) can be written
\[
\alpha(z)=\exp\left(\int_{\Omega\setminus\Omega^{*}}\log\alpha_{\omega}(z)\mu(\rd\omega)\right)\times\prod_{\omega\in\Omega^{*}}\alpha_{\omega}(z)\mu(\{\omega\}),
\]

which is in some sense a hybrid of the local and doubly-stochastic
acceptance ratios. Define the measure on $P_{\omega}$ as the pushforward
of the measure $\mathbb{Q}_{|P^{*}|\geq1}$ for the mapping $\upsilon$;
the distribution of $P_{\omega}$, conditional on rejecting the forward
move $\Phi(z)$, is

\[
\mathbb{Q}_{|P_{\omega}|\geq1}^{*}(\rd P_{\omega}|z)=\mathbb{Q}_{|P|\geq1}\left(\upsilon^{-1}(\rd P_{\omega})|z\right).
\]

Similarly to Assumption A\ref{hypA:Sassumptions-DT}.3, it is sufficient
that the bounce transition kernel $Q_{P_{\omega}}^{*}$ satisfy for
$\mathbb{Q}_{|P_{\omega}|\geq1}^{*}(\rd P_{\omega}|z)$-almost all
$P_{\omega}\in\mathcal{P}$

\[
\int\rho(\rd z)(1-\alpha(z))\frac{\rd\mathbb{Q}_{|P_{\omega}|\geq1}^{*}(\rd P_{\omega}|z)}{\rd\mathbb{Q}_{|P_{\omega}|\geq1}^{*}(\rd P_{\omega}|\mathcal{S}(z'))}Q_{P_{\omega}}^{*}(z,\rd z')=\rho(\mathcal{S}(\rd z'))(1-\alpha(\mathcal{S}(z'))),
\]

and that the Radon-Nikodym derivative above is well-defined and strictly
positive for $Q_{P_{\omega}}^{*}(z,\rd z')$-almost all $z'$. The
above implies an algorithm similar to Algorithm \ref{alg:discrete-time-doubly-stochastic-MCMC}
but where $\Psi_{P_{\omega}}(z)$ would be accepted with a probability
of
\begin{align*}
\alpha(z)= & \exp\left(-\int_{\Omega\setminus\Omega^{*}}\left[\log\alpha_{\omega}(z)-H_{\omega}(z)-\log\alpha_{\omega}(\mathcal{S}\circ\Psi_{P_{\omega}}(z))+H_{\omega}(\mathcal{S}\circ\Psi_{P_{\omega}}(z))\right]_{+}\mu(\rd\omega)\right)\\
\times & \min\left(1,\prod_{\omega\in P_{\omega}\setminus\Omega^{*}}\frac{\log\alpha_{\omega}(\mathcal{S}\circ\Psi_{P_{\omega}}(z))}{\log\alpha_{\omega}(z)}\prod_{\omega^{*}\in\Omega^{*}}\frac{\exp\left(-H_{\omega}(\mathcal{S}\circ\Psi_{P_{\omega}}(z))\mu(\{\omega\})\right)\mathrm{Ber}\left(\mathbb{I}(\omega^{*}\in P_{\omega});1-\alpha_{\omega^{*}}(\mathcal{S}\circ\Psi(z))\right)}{\exp\left(-H_{\omega}(z)\mu(\{\omega\})\right)\mathrm{Ber}\left(\mathbb{I}(\omega^{*}\in P_{\omega});1-\alpha_{\omega^{*}}(z)\right)}\right).
\end{align*}

\section{Numerical results\label{sec:Numerical-results}}

\subsection{Hamiltonian BPS\label{subsec:Hamiltonian-PDMP-results}}

In \citep{BC_D_V_15}, the local BPS algorithm was shown to outperform
various state-of-the-art HMC algorithms in sparse precision Gaussian
random field models with Poisson observations. In this section, we
investigate the relative performance of local BPS and Hamiltonian
BPS in the same setting. We find that Hamiltonian BPS has a modest
advantage over local BPS when the number of observations is small
but the dimensionality of the latent variables is high. On the other
hand, when the number of observations is equal to the number of latent
variables, the situation is reversed. However in both regimes Hamiltonian
BPS outperforms global BPS, and it is worth keeping in mind that there
are situations where Hamiltonian BPS is applicable while the local
BPS is not computationally attractive, for example if a single variable
is connected to all factors. 

\subsubsection{Hamiltonian flow}

In the notation of Section \ref{subsec:Hamiltonian-PDMP}, we consider
an example where $V$ corresponds to the isotropic prior normal distribution
of a Bayesian model and so $\tilde{U}$ corresponds to the negative
log-likelihood. Under this assumption, the corresponding Hamiltonian
flow is given for $i\in[d]$ by
\[
\flow_{t}(z)=\exp\text{\ensuremath{\left(t\left[\begin{array}{cc}
0 & I\\
-I & 0
\end{array}\right]\right)z=\sin(t)\begin{bmatrix}0 & I\\
-I & 0
\end{bmatrix}z+\cos(t)z}}
\]

More generally, if $V$ is an arbitrary normal distribution, the situation
considered here can be used after a change of variables. The computational
trade-off results we present in this section are hence representative
of situations where we have a high-dimensional Gaussian prior with
a precision matrix admitting a Cholesky decomposition that can be
computed in time $O(d)$, which arises for example in certain time
series models and corresponds to a best case scenario for Hamiltonian
BPS. 

\subsubsection{Exact simulation of bounce times\label{subsec:Computation-of-bounce}}

Let $j\in\left\{ 1,2,\dots,k\right\} $ index the observations. Assume
that the negative log-likelihood $\widetilde{U}\left(x\right)$ can
be decomposed as $\widetilde{U}(x)=\sum_{j=1}^{k}\widetilde{U}_{j}(x_{i(j)})$
for some function $i(\cdot)$ mapping observation indices to the latent
variable indices. As a pre-processing step, we compute (numerically
or analytically) a bound $B_{j}(b)\ge\sup\left\{ |\nabla\widetilde{U}_{j}(x)|:|x|<b\right\} $. 

Let $x=x_{i(j)}$ and $v=v_{i(j)}$ denote the initial position and
velocity at the beginning of the current piecewise Hamiltonian segment
for the latent variable $i(j)$ associated with observation $j$.
From Section 2.3.3 of \citep{BC_D_V_15}, it is enough to simulate
the bounce time of a single factor $\widetilde{U}_{j}(x_{i(j)})$.
Using the methodology developed in \citep[Section 2.3.2]{BC_D_V_15},
we simulate the bounce time of each factor using thinning and the
following bound on the intensity $\chi(t)$:

\begin{eqnarray*}
\chi(t) & = & \max\left\{ 0,(-x\sin(t)+v\cos(t))\nabla\widetilde{U}_{j}(v\sin(t)+x\cos(t)\right\} \\
 & = & \max\left\{ 0,\sqrt{x^{2}+v^{2}}\cos(t-\alpha)\nabla\widetilde{U}_{j}\left(\sqrt{x^{2}+v^{2}}\cos(t-\beta)\right)\right\} \\
 & \le & \sqrt{x^{2}+v^{2}}B_{j}\left(\sqrt{x^{2}+v^{2}}\right),
\end{eqnarray*}
 where $\alpha=\arctan(-v/x)$, $\beta=\arctan(v/d)$.

\subsubsection{Results}

\begin{figure}
\includegraphics[width=1\columnwidth]{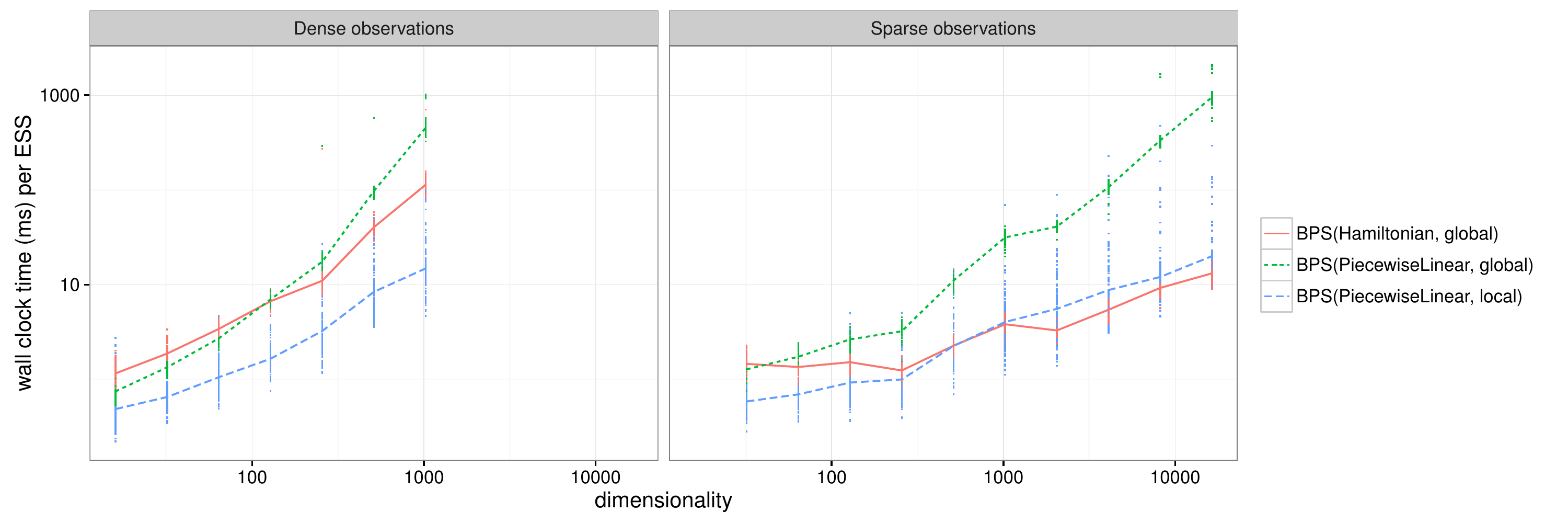}

\caption{\label{fig:Results-on-the}Results on the model described in Section
\ref{subsec:Hamiltonian-PDMP-results}. Dense observations (left)
correspond to the case where the number of observations grows linearly
with the dimensionality of the latent field, $k=d$, while sparse
observation (right) correspond to the case where the number of observations
is held fix ($k=16$).}
\end{figure}

We consider a likelihood given by conditionally independent Poisson
observations with observations $y_{i}$ having a natural exponential
family parameter given by the latent random variable $x_{i}$:

\[
\P(y_{i}=n|x)\propto\exp(-x_{i}n+\exp x_{i}).
\]

We compare three algorithms: local and global BPS with piecewise linear
trajectories, and Hamiltonian BPS. Computation of the bounce times
for the piecewise linear trajectories is done as in \citep{BC_D_V_15}.
For the bounce times of Hamiltonian BPS, we use the result from Section
\ref{subsec:Computation-of-bounce} with $B_{j}(b)=\exp(b)+y_{j}$. 

We show in Figure \ref{fig:Results-on-the} the scaling of the CPU
wall clock time required to obtain one effective sample size (ESS)
as a function of the dimensionality $d$ (log-log scale). %
The wall clock time is measured in milliseconds on a 2.8 GHz Intel
Core i7, and the ESS is computed using a batch mean estimator with
a test function given by $f(x)=x_{1}^{2}$. Expectations from piecewise-deterministic
trajectories are computed analytically as shown in \citep{BC_D_V_15}
and from piecewise Hamiltonian trajectories, using numerical integration.
For each dimension and algorithm, we run $100$ independent chains
and average the running times per ESS. 

\subsection{Empirical comparisons of local and global BPS to HMC and Standard
and Elliptical Slice Sampling}

\subsubsection{Setup}

We consider four models, built from two prior distributions: first,
a Brownian bridge prior, and second, a diagonal precision prior. For
each prior, we consider either a Poisson likelihood with synthetic
observations (with the same structure as described in the previous
section), or no likelihood function. We consider the following sampling
methods: the Elliptical Slice Sampler \citep{M_A_2010}, the ``Standard''
Slice Sampler (with exponential slice growing and slice shrinking)
\citep{neal2003}, HMC, or more precisely the NUTS algorithm implemented
in Stan, the local and global BPS algorithm with linear trajectories,
and the Hamiltonian BPS algorithm. For each combination, we run the
algorithms on latent fields of dimensionality $\left\{ 2^{0},2^{1},2^{2},\dots,2^{7}\right\} $,
and replicate the experiment 50 times with different random seeds.
We measure ESS and wall clock time. ESS is computed using a batch
mean estimator with a test function given by $f(x)=x_{1}^{2}$. 

\subsubsection{Results\label{subsec:slice-Results}}

\begin{figure}
\includegraphics[width=1\columnwidth]{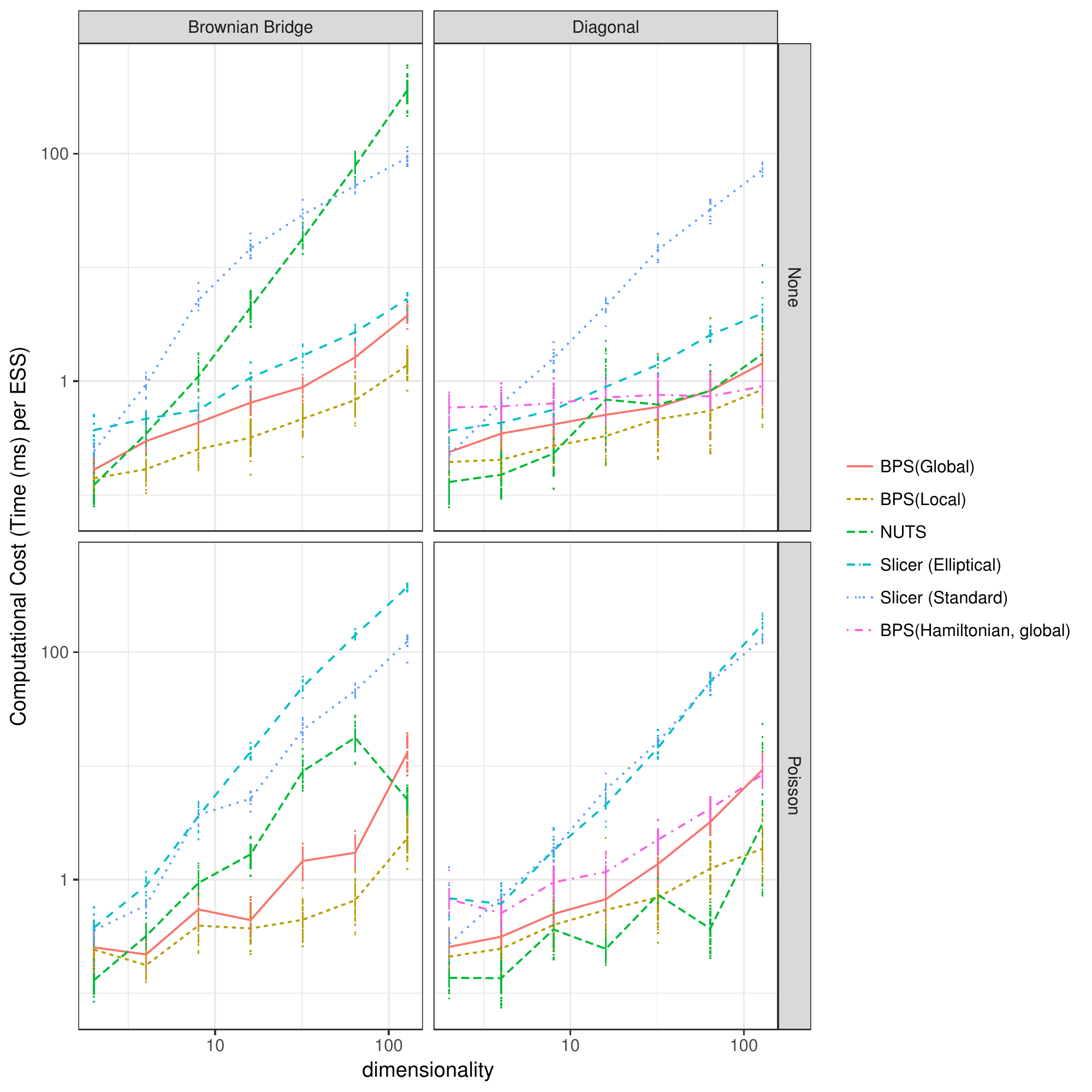}

\caption{\label{fig:dim-scaling-combined-1}Main results of Section \ref{subsec:slice-Results}.
The ordinate shows the empirical computational complexity (ms per
ESS) and the abscissa (lower is better), the dimensionality. Both
axes are in log-scale. The dots show the variability from 50 independent
runs with different random seeds, and the line, the averages. }
\end{figure}

We summarize the main results of this section in Figure \ref{fig:dim-scaling-combined-1},
where the empirical computational complexity (wall clock time (ms)
per ESS) is plotted in log-log scale against the dimensionality of
the field for the four models. For sufficiently high-dimensional scenarios
(> 10 dimensions), local BPS outperforms all other methods in 3 out
of the 4 settings. In the fourth setting, (Diagonal Precision + Poisson
Likelihood), NUTS (HMC) and Local BPS outperform the other methods,
but neither strictly dominate the other. Elliptic Slice Sampling is
competitive when there is no likelihood, but it is still not better
than Local BPS, presumably because the latter can use the full trajectory
when computing averages whereas Elliptical is discrete-time. However,
once the Poisson Likelihood is added, Elliptical Sampling seems to
have worse asymptotics, empirically roughly $O(n^{3/2})$ versus roughly
$O(n^{1+\epsilon})$ for the best performing methods. 

\subsection{Randomized bounces}

In this section, we compare the performance of several collision operators
on two collections of problems of increasing dimensionality. 

\subsubsection{Setup}

The first collection of target distributions we consider consists
in funnel distributions from \citep{neal2011}, namely multivariate
normals of varying dimension $d$ with diagonal covariance matrix
and standard deviations for each components given by $1,(d-1)/d,(d-2)/d,\dots,1/d$.
Since the algorithms considered are rotationally invariant, this is
representative of problems with averse conditioning. The second collection
consists in isotropic multivariate normal of increasing dimensionality
$d$. The isotropic examples are useful to identify cases where symmetries
create a clear imperative for refreshment as discussed in \citep{BC_D_V_15}.
For each class of target distributions, we look at problems of dimensionality
$2^{1},2^{2},\dots,2^{7}$. 

We compare 8 algorithms, corresponding to $4$ different bounce operators
and $2$ refreshment strategies (either independent refreshment at
times determined by a unit rate homogeneous Poisson process, or no
refreshment). The bounce operator labeled Flip corresponds to $Q_{x}(v,\rd v')=\delta_{-v}(\rd v')$,
Det-Rand corresponds to the forward-event chain algorithm of \citep{M_S_17},
and Rand-Rand corresponds to the independent sampling algorithm of
\citep{F_B_P_R_16}. We recorded the Monte Carlo averages $\hat{f_{i}}$
of the test function $f(x)=x_{1}^{2}$ for the trajectory up to event
time index $i=2^{0},2^{1},\dots,2^{14}$ and computed the errors $e_{i}=|\hat{f}_{i}-1|$
. We then averaged the errors over $20$ independent executions of
the algorithms using different random seeds. All experiments in this
section are performed on a global (continuous-time) BPS algorithm.
Both simulation of collision times and computation of Monte Carlo
averaged are performed using closed form expressions that can be found
in \citep{BC_D_V_15}. 

\subsubsection{Results}

We show in Figure \ref{fig:random-bounce-results} the average error
as a function of the event index (log-log scale). %

\begin{figure}
\includegraphics[width=1\columnwidth]{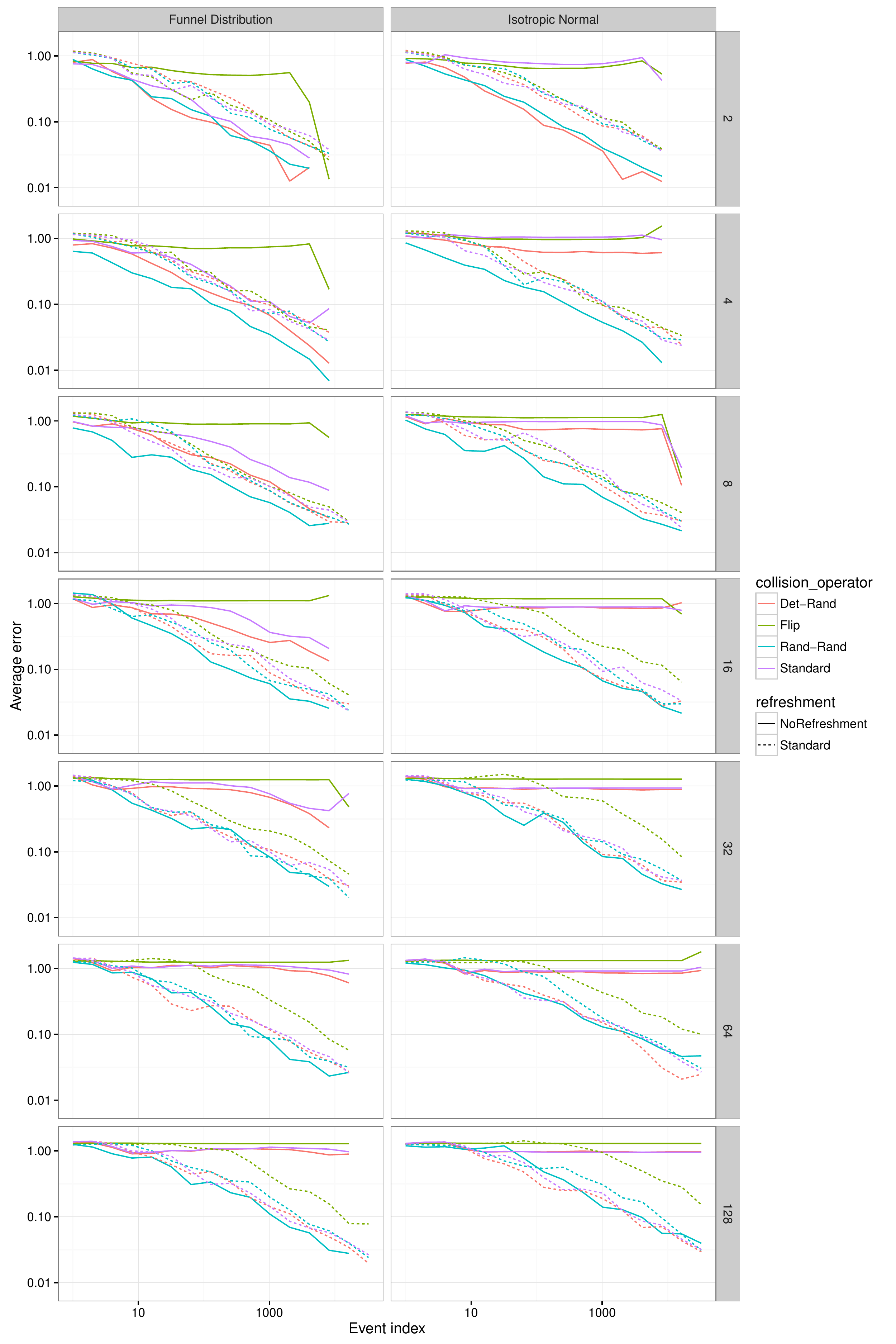}

\caption{\label{fig:random-bounce-results}Errors of Monte Carlo partial sums
averaged over 20 independent runs for different bounce operators,
models (columns) and dimensionalities (rows). }
\end{figure}

Our results show that in the low dimensional regime, at least two
randomized bounce operators (Det-Rand and Rand-Rand) combined with
no refreshment outperform the standard bounce with refreshment. However,
this advantage asymptotically vanishes as the dimensionality of the
problem increases. In fact, when refreshment is turned off, for all
the operators but Rand-Rand, performance dramatically collapses with
dimensionality. The performance drop-off is so pronounced that it
may not be detected by conventional estimators of effective sample
size. We can measure it here since the true value of the expectations
are known. 

We conjecture that this sharp drop in performance is due to a concentration
of measure phenomenon making the variance of the randomized operators
in the direction parallel to the gradient decrease with $d$, hence,
informally speaking, making certain randomized operators such as Det-Rand
more and more deterministic as $d$ increases. The lack of irreducibility
of deterministic bounce operators without refreshment is shown formally
in \citep{BC_D_V_15}. This conjecture is also supported by the fact
that reintroducing refreshment makes all methods behave similarly
in high-dimensional settings (except for the cruder Flip operator).

This is noteworthy as one of the motivations for previous work on
alternative bounce operators is that such operators may alleviate
the need for refreshment in certain scenarios%
. Our results provide a cautionary example that in certain high-dimensional
scenarios, it is still preferable to perform refreshment even when
randomized bounces are used. Interestingly, this happens not only
in the isotropic case but also in the non-isotropic, funnel distribution
case, where one might expect refreshment to play a more minor role
due to lack of symmetry.

\section{Discussion}

We have introduced a general framework which allows us to develop
novel continuous-time and discrete-time PD-MCMC algorithms addressing
some of the limitations of existing techniques. They allow to exploit
dynamics dependent on the target distribution. Moreover, contrary
to continuous-time algorithms, it is always possible to simulate exactly
the event times.

There are many possible methodological extensions of these algorithms.
To simplify presentation, we have presented our results for auxiliary
distributions of the form $\psi\left(v\right)=g(|v|)$ but, as in
the HMC context \citep{G_C_11}, it is possible to adapt these techniques
to the scenario where $\rho\left(z\right)=\pi\left(x\right)\psi_{x}\left(v\right)$
with $\psi_{x}\left(v\right)=g(|v^{T}M\left(x\right)v|^{1/2})$ for
$M\left(x\right)$ a positive definite matrix capturing the local
curvature of $U$ around $x$. From preliminary experiments, we observe
that using a position-dependent mass matrix $M(x)$ can provide significant
gains in complex scenarios. Even selecting simply a suitable constant
matrix $M$ can already improved substantially performance as already
demonstrated for the BPS \citep{P_dW_12,G_16,P_G_C_P_17}. Moreover,
the proposed framework is very flexible but all the algorithms proposed
so far in continuous-time are based on a divergence-free vector field
and in discrete-time on a deterministic mapping with unit Jacobian
determinant. There is conceptually no need to restrict ourselves to
such scenarios and it would be interesting to come up with useful
algorithms exploiting this degree of freedom. 

From a theoretical point of view, PD-MCMC techniques appear to provide
state-of-the-art performance on some interesting sampling problems
but there are only few theoretical results available \citep{B_R_16,D_BC_D_17,M_16}
and there is much work to be done to better understand their properties.

\appendix

\section{Proofs of invariance\label{sec:Proofofinvariance}}
\begin{proof}[Proof of Proposition~\ref{prop:Ginvariance}]
Using Assumption A\ref{hypA:Gassumptions}.3 then Assumption A\ref{hypA:Gassumptions}.1,
we obtain 
\begin{align*}
\iint\rho\left(\rd z\right)\lambda\left(z\right)Q\left(z,\rd z'\right)f\left(z'\right) & =\int\rho\left(\mathcal{S}^{-1}\left(\rd z'\right)\right)\lambda\left(\mathcal{S}\left(z'\right)\right)f\left(z'\right)\\
 & =\int\rho\left(\rd z'\right)\lambda\left(\mathcal{S}\left(z'\right)\right)f\left(z'\right).
\end{align*}
Hence, (\ref{eq:PDMPinvariant}) is equal to
\begin{align*}
 & \int\rho\left(\rd z\right)\left[\lambda\left(z\right)\int Q\left(z,\rd z'\right)\left[f\left(z'\right)-f\left(z\right)\right]-\left\{ \nabla\cdot\phi\left(z\right)-\left\langle \nabla H\left(z\right),\phi\left(z\right)\right\rangle \right\} f\left(z\right)\right]\\
= & \int\rho\left(\rd z\right)\left[\left\{ \lambda\left(\mathcal{S}\left(z\right)\right)-\lambda\left(z\right)\right\} -\left\{ \nabla\cdot\phi\left(z\right)-\left\langle \nabla H\left(z\right),\phi\left(z\right)\right\rangle \right\} \right]f\left(z\right)=0
\end{align*}
under Assumption A\ref{hypA:Gassumptions}.2. This establishes the
result.
\end{proof}
\begin{proof}[Proof of Proposition~\ref{prop:Linvariance}]
As $Q$ is given by (\ref{eq:Kernellocal}), we obtain 
\begin{align*}
\iint\rho\left(\rd z\right)\lambda\left(z\right)Q\left(z,\rd z'\right)f\left(z'\right) & =\iint\rho\left(\rd z\right)\lambda\left(z\right)\left\{ \sum_{i=1}^{n}\frac{\lambda_{i}\left(z\right)}{\lambda\left(z\right)}Q_{i}\left(z,\rd z'\right)\right\} f\left(z'\right)\\
 & =\sum_{i=1}^{n}\iint\rho\left(\rd z\right)\lambda_{i}\left(z\right)Q_{i}\left(z,\rd z'\right)f\left(z'\right)\\
 & =\sum_{i=1}^{n}\int\rho\left(\mathcal{S}^{-1}\left(\rd z'\right)\right)\lambda_{i}\left(\mathcal{S}\left(z'\right)\right)f\left(z'\right)\\
 & =\sum_{i=1}^{n}\int\rho\left(\rd z'\right)\lambda_{i}\left(\mathcal{S}\left(z'\right)\right)f\left(z'\right)
\end{align*}
where we have used Assumptions A\ref{hypA:Lassumptions}.3 and A\ref{hypA:Lassumptions}.1.
Hence, (\ref{eq:PDMPinvariant}) is equal to
\begin{align*}
 & \int\rho\left(\rd z\right)\left[\lambda\left(z\right)\int Q\left(z,\rd z'\right)\left[f\left(z'\right)-f\left(z\right)\right]-\left\{ \nabla\cdot\phi\left(z\right)-\left\langle \nabla H\left(z\right),\phi\left(z\right)\right\rangle \right\} f\left(z\right)\right]\\
= & \int\rho\left(\rd z\right)\left[\sum_{i=1}^{n}\left\{ \lambda_{i}\left(\mathcal{S}\left(z\right)\right)-\lambda_{i}\left(z\right)\right\} -\left\{ \nabla\cdot\phi\left(z\right)-\left\langle \nabla H\left(z\right),\phi\left(z\right)\right\rangle \right\} \right]f\left(z\right)=0
\end{align*}
under Assumption A\ref{hypA:Lassumptions}.2. The result follows.
\end{proof}
\begin{proof}[Proof of Proposition~\ref{prop:Sinvariance}]
The proof is similar to the proof of Proposition \ref{prop:Linvariance}
and is therefore omitted.
\end{proof}

\begin{proof}[Proof of Proposition~\ref{prop:GMHinvariance}]
We have 
\begin{equation}
\nu\left(\rd z\right)T\left(z,\rd z'\right)=\nu\left(\rd z\right)M\left(z,\rd z'\right)\beta\left(z,z'\right)+\nu\left(\rd z\right)\delta_{\mathcal{S}\left(z\right)}\left(\rd z'\right)\gamma\left(z\right)\label{eq:GMHexpression}
\end{equation}
where 
\[
\gamma\left(z\right)=1-\int\beta\left(z,z'\right)M\left(z,\rd z'\right).
\]
First notice that, if using Assumption A\ref{hypA:GMHassumptions}.2,
we define 
\[
r(z,z'):=\frac{\nu\left(\mathcal{S}\left(\rd z'\right)\right)M\left(\mathcal{S}\left(z'\right),\mathcal{S}\left(\rd z\right)\right)}{\nu\left(\rd z\right)M\left(z,\rd z'\right)},
\]
then using the properties of the push-forward measure and Assumption
A\ref{hypA:GMHassumptions}.1, we have for any measurable function
$h$

\begin{align*}
 & \iint h(z,z')r\left(\mathcal{S}(z'),\mathcal{S}(z)\right)\nu\left(\mathcal{S}\left(\rd z'\right)\right)M\left(\mathcal{S}\left(z'\right),\mathcal{S}\left(\rd z\right)\right)\\
 & =\int\nu\left(\mathcal{S}\left(\rd z'\right)\right)\int h(z,z')r\left(\mathcal{S}(z'),\mathcal{S}(z)\right)M\left(\mathcal{S}\left(z'\right),\mathcal{S}\left(\rd z\right)\right)\\
 & =\int\nu\left(\mathcal{S}\left(\rd z'\right)\right)\int h\left(\mathcal{S}\left(z\right),z'\right)r\left(\mathcal{S}(z'),z\right)M\left(\mathcal{S}\left(z'\right),\rd z\right)\\
 & =\iint\nu\left(\mathcal{S}\left(\rd z'\right)\right)M\left(\mathcal{S}\left(z'\right),\rd z\right)r\left(\mathcal{S}(z'),z\right)h\left(\mathcal{S}\left(z\right),z'\right)\\
 & =\iint\nu\left(\rd z'\right)M\left(z',\rd z\right)r\left(z',z\right)h\left(\mathcal{S}\left(z\right),\mathcal{S}\left(z'\right)\right)\\
 & =\iint\nu\left(\mathcal{S}\left(\rd z'\right)\right)M\left(\mathcal{S}\left(z'\right),\mathcal{S}\left(\rd z\right)\right)h\left(\mathcal{S}\left(z\right),\mathcal{S}\left(z'\right)\right)\\
 & =\iint\nu\left(\rd z'\right)M\left(z',\rd z\right)h\left(z,z'\right).
\end{align*}
This establishes that the measure $\nu\left(\rd z\right)M\left(z,\rd z'\right)$
is absolutely continuous w.r.t.\ $\nu\left(\mathcal{S}\left(\rd z'\right)\right)M\left(\mathcal{S}\left(z'\right),\mathcal{S}\left(\rd z\right)\right)$
with a Radon-Nikodym derivative given by 
\[
r\left(\mathcal{S}(z'),\mathcal{S}(z)\right)=\frac{\nu\left(\rd z\right)M\left(z,\rd z'\right)}{\nu\left(\mathcal{S}\left(\rd z'\right)\right)M\left(\mathcal{S}\left(z'\right),\mathcal{S}\left(\rd z\right)\right)}.
\]

For the first term on the r.h.s.\ of (\ref{eq:GMHexpression}), we
have 
\begin{align}
\nu\left(\rd z\right)M\left(z,\rd z'\right)\beta\left(z,z'\right)= & \nu\left(\rd z\right)M\left(z,\rd z'\right)g\left(\frac{\nu\left(\mathcal{S}\left(\rd z'\right)\right)M\left(\mathcal{S}\left(z'\right),\mathcal{S}\left(\rd z\right)\right)}{\nu\left(\rd z\right)M\left(z,\rd z'\right)}\right)\nonumber \\
= & \nu\left(\mathcal{S}\left(\rd z'\right)\right)M\left(\mathcal{S}\left(z'\right),\mathcal{S}\left(\rd z\right)\right)g\left(\frac{\nu\left(\rd z\right)M\left(z,\rd z'\right)}{\nu\left(\mathcal{S}\left(\rd z'\right)\right)M\left(\mathcal{S}\left(z'\right),\mathcal{S}\left(\rd z\right)\right)}\right)\nonumber \\
= & \nu\left(\mathcal{S}\left(\rd z'\right)\right)M\left(\mathcal{S}\left(z'\right),\mathcal{S}\left(\rd z\right)\right)\beta\left(\mathcal{S}\left(z'\right),\mathcal{S}\left(z\right)\right)\label{eq:term1GMH}
\end{align}
where we have used Assumption A\ref{hypA:GMHassumptions}.3 then Assumption
A\ref{hypA:GMHassumptions}.1. 

The second term on the r.h.s.\ of (\ref{eq:GMHexpression}) satisfies
\begin{align}
\nu\left(\rd z\right)\delta_{\mathcal{S}\left(z\right)}\left(\rd z'\right)\gamma\left(z\right)= & \nu\left(\mathcal{S}^{-1}\left(\rd z'\right)\right)\delta_{\mathcal{S}^{-1}\left(z'\right)}\left(\rd z\right)\gamma\left(\mathcal{S}^{-1}\left(z'\right)\right)\nonumber \\
= & \nu\left(\mathcal{S}\left(\rd z'\right)\right)\delta_{\mathcal{S}\left(z'\right)}\left(\rd z\right)\gamma\left(\mathcal{S}\left(z'\right)\right)\label{eq:term2GMH}
\end{align}
using Assumption A\ref{hypA:GMHassumptions}.1. The sum of the terms
(\ref{eq:term1GMH}) and (\ref{eq:term2GMH}) is equal to $\nu\left(\mathcal{S}\left(\rd z'\right)\right)T\left(\mathcal{S}\left(z'\right),\mathcal{S}\left(\rd z\right)\right)$.
Hence the GMH kernel satisfies the skewed detailed balance condition
(\ref{eq:skeweddetailedbalancecondition}).
\end{proof}
\begin{proof}[Proof of Proposition~\ref{prop:Ginvariance-DT}]
The proof follows from simple manipulations. We have from Assumption
A\ref{hypA:Gassumptions-DT}.3 then Assumption A\ref{hypA:Gassumptions-DT}.1
that the l.h.s.\ of (\ref{eq:invariancekerneldiscretetime}) satisfies
\begin{align*}
 & \rho\left(\Phi^{-1}\left(z'\right)\right)\alpha\left(\Phi^{-1}\left(z'\right)\right)\left|\nabla\Phi^{-1}\left(z'\right)\right|\rd z'+\int\rho\left(\rd z\right)\left(1-\alpha\left(z\right)\right)Q\left(z,\rd z'\right)\\
= & \rho\left(\Phi^{-1}\left(z'\right)\right)\alpha\left(\Phi^{-1}\left(z'\right)\right)\left|\nabla\Phi^{-1}\left(z'\right)\right|\rd z'+\rho\left(\mathcal{S}^{-1}\left(\rd z'\right)\right)\left(1-\alpha\left(\mathcal{S}\left(z'\right)\right)\right)\\
= & \rho\left(\Phi^{-1}\left(z'\right)\right)\alpha\left(\Phi^{-1}\left(z'\right)\right)\left|\nabla\Phi^{-1}\left(z'\right)\right|\rd z'+\rho\left(\rd z'\right)\left(1-\alpha\left(\mathcal{S}\left(z'\right)\right)\right).
\end{align*}
Hence the condition (\ref{eq:invariancekerneldiscretetime}) is satisfied
if for all $z'$
\begin{align*}
 & \rho\left(\Phi^{-1}\left(z'\right)\right)\alpha\left(\Phi^{-1}\left(z'\right)\right)\left|\nabla\Phi^{-1}\left(z'\right)\right|-\rho\left(z'\right)\alpha\left(\mathcal{S}\left(z'\right)\right)=0.
\end{align*}
By rewriting this expression for $z'=\Phi\left(z\right)$, and using
the fact that $\left|\nabla\Phi^{-1}\left(z'\right)\right|=\left|\nabla\Phi\left(\Phi^{-1}\left(z'\right)\right)\right|^{-1}$
so $\left|\nabla\Phi^{-1}\left(z'\right)\right|=\left|\nabla\Phi\left(z\right)\right|^{-1}$,
we obtain Assumption A\ref{hypA:Gassumptions-DT}.2.
\end{proof}
\begin{proof}[Proof of Proposition~\ref{prop:Linvariance-DT}]
The proof follows from simple manipulations. We consider first the
second term on the r.h.s.\ of (\ref{eq:invariancekerneldiscretetime}).
This satisfies
\begin{align*}
\int\rho\left(\rd z\right)\left\{ 1-\alpha\left(z\right)\right\} Q\left(z,\rd z'\right) & =\int\rho\left(\rd z\right)\left\{ 1-\alpha\left(z\right)\right\} \sum_{B\in\mathcal{B}}\frac{\prod_{i:B_{i}=0}\alpha_{i}\left(z\right)\prod_{i:B_{i}=1}\left(1-\alpha_{i}\left(z\right)\right)}{1-\alpha\left(z\right)}Q_{B}\left(z,\rd z'\right)\\
 & =\sum_{B\in\mathcal{B}}\int\rho\left(\rd z\right)\prod_{i:B_{i}=0}\alpha_{i}\left(z\right)\prod_{i:B_{i}=1}\left(1-\alpha_{i}\left(z\right)\right)Q_{B}\left(z,\rd z'\right)\\
 & =\sum_{B\in\mathcal{B}}\rho\left(\mathcal{S}^{-1}\left(\rd z'\right)\right)\prod_{i:B_{i}=0}\alpha_{i}\left(\mathcal{S}\left(z'\right)\right)\prod_{i:B_{i}=1}\left(1-\alpha_{i}\left(\mathcal{S}\left(z'\right)\right)\right)\\
 & =\rho\left(\mathcal{S}^{-1}\left(\rd z'\right)\right)\sum_{B\in\mathcal{B}}\prod_{i:B_{i}=0}\alpha_{i}\left(\mathcal{S}\left(z'\right)\right)\prod_{i:B_{i}=1}\left(1-\alpha_{i}\left(\mathcal{S}\left(z'\right)\right)\right)\\
 & =\rho\left(\mathcal{S}^{-1}\left(\rd z'\right)\right)\left(1-\prod_{i=1}^{n}\alpha_{i}\left(\mathcal{S}\left(z'\right)\right)\right)\\
 & =\rho\left(\rd z'\right)\left(1-\alpha\left(\mathcal{S}\left(z'\right)\right)\right),
\end{align*}
where we have used Assumption A\ref{hypA:Lassumptions-DT}.3 then
Assumption A\ref{hypA:Lassumptions-DT}.1. The first term on the l.h.s.\ of
(\ref{eq:invariancekerneldiscretetime}) is given by 
\[
\rho\left(\Phi^{-1}\left(z'\right)\right)\alpha\left(\Phi^{-1}\left(z'\right)\right)\left|\nabla\Phi^{-1}\left(z'\right)\right|=\rho\left(\Phi^{-1}\left(z'\right)\right)\left|\nabla\Phi^{-1}\left(z'\right)\right|\prod_{i=1}^{n}\alpha_{i}\left(\Phi^{-1}\left(z'\right)\right).
\]
Hence the condition (\ref{eq:invariancekerneldiscretetime}) is satisfied
if for all $z'$
\begin{align*}
\rho\left(\Phi^{-1}\left(z'\right)\right)\left|\nabla\Phi^{-1}\left(z'\right)\right|\prod_{i=1}^{n}\alpha_{i}\left(\Phi^{-1}\left(z'\right)\right)-\rho\left(z'\right)\prod_{i=1}^{n}\alpha_{i}\left(\mathcal{S}\left(z'\right)\right)=0.
\end{align*}
By rewriting this expression for $z'=\Phi\left(z\right)$, we obtain
Assumption A\ref{hypA:Lassumptions-DT}.2.\ which is also implied
by (\ref{eq:LocalLocalRateBalance-DT}) if $\left|\nabla\Phi\left(z\right)\right|=1$
for all $z$.
\end{proof}
\begin{proof}[Proof of Proposition~\ref{prop:Sinvariance-DT}]
The proof is very similar to the proof of Proposition \ref{prop:Linvariance-DT}.
We similarly consider the second term on the r.h.s. of (\ref{eq:invariancekerneldiscretetime})
which satisfies
\begin{align*}
\int\rho\left(\rd z\right)\left(1-\alpha\left(z\right)\right)Q\left(z,\rd z'\right) & =\int\rho\left(\rd z\right)\left(1-\alpha\left(z\right)\right)\int_{\mathcal{P}}\mathbb{Q}_{|P|\geq1}\left(\rd P|z\right)Q_{P}\left(z,\rd z'\right)\\
 & =\int_{\mathcal{P}}\int\rho\left(\rd z\right)\left(1-\alpha\left(z\right)\right)\mathbb{Q}_{|P|\geq1}\left(\rd P|z\right)Q_{P}\left(z,\rd z'\right)\\
 & =\int_{\mathcal{P}}\rho\left(\mathcal{S}^{-1}\left(\rd z'\right)\right)\left(1-\alpha\left(\mathcal{S}\left(z'\right)\right)\right)\mathbb{Q}_{|P|\geq1}\left(\rd P|\mathcal{S}\left(z'\right)\right)\\
 & =\rho\left(\mathcal{S}^{-1}\left(\rd z'\right)\right)\left(1-\alpha\left(\mathcal{S}\left(z'\right)\right)\right)\\
 & =\rho\left(\rd z'\right)\left(1-\alpha\left(\mathcal{S}\left(z'\right)\right)\right)
\end{align*}
where we have used Assumption A\ref{hypA:Sassumptions-DT}.3 then
Assumption A\ref{hypA:Sassumptions-DT}.1. The first term on the l.h.s.\ of
(\ref{eq:invariancekerneldiscretetime}) is given by 
\[
\rho\left(\Phi^{-1}\left(z'\right)\right)\alpha\left(\Phi^{-1}\left(z'\right)\right)\left|\nabla\Phi^{-1}\left(z'\right)\right|=\rho\left(\Phi^{-1}\left(z'\right)\right)\left|\nabla\Phi^{-1}\left(z'\right)\right|\exp\left\{ \int\log\alpha_{\omega}\left(\Phi^{-1}\left(z'\right)\right)\mu\left(\rd\omega\right)\right\} 
\]
Hence the condition (\ref{eq:invariancekerneldiscretetime}) is satisfied
if for all $z'$
\begin{align*}
\rho\left(\Phi^{-1}\left(z'\right)\right)\left|\nabla\Phi^{-1}\left(z'\right)\right|\exp\left\{ \int\log\alpha_{\omega}\left(\Phi^{-1}\left(z'\right)\right)\mu\left(\rd\omega\right)\right\} -\rho\left(z'\right)\exp\left\{ \int\log\alpha_{\omega}\left(\mathcal{S}\left(z'\right)\right)\mu\left(\rd\omega\right)\right\} =0
\end{align*}
By rewriting this expression for $z'=\Phi\left(z\right)$, we obtain
Assumption A\ref{hypA:Sassumptions-DT}.2.which is also implied by
(\ref{eq:LocalLocalRateBalance-SDT}) if $\left|\nabla\Phi\left(z\right)\right|=1$
for all $z$.
\end{proof}

\section{Weak convergence of discrete-time BPS\label{sec:Weak-convergence-of}}

\subsection{Main result}

Let $\rho(x,v)$ denote the target density on $\mathcal{Z}:=\mathbb{R}^{d}\times\mathbb{R}^{d}$,
where 
\[
\rho(x,v)=\pi(x)\psi(v)\propto\mathrm{e}^{-U(x)}\psi(v).
\]
We will establish our weak convergence results for $\psi$ an isotropic
distribution on $B_{R}(\mathbb{R}^{d})$ for some $R>0$, $B_{R}(\mathbb{R}^{d})$
being the Euclidean ball of radius $R$. This includes the uniform
distribution on $\mathbb{S}^{d-1}$. Let us define
\[
D_{\mathcal{Z}}[0,\infty):=\left\{ \left.f:[0,\infty)\mapsto\mathcal{Z}\right|\text{\ensuremath{f} is right continuous with left limits}\right\} .
\]

Let also $C_{0}(\mathcal{Z})$ be the space of continuous functions
$f:\mathcal{Z}\mapsto\mathbb{R}$ such that $f(x)\to0$ as $|x|\to\infty$,
in the sense that for all $\epsilon>0$ the set $\left\{ z\in\mathcal{Z}:\left|f(z)\right|\geq\epsilon\right\} $
is compact. Finally let $\left\{ P^{t}:t\geq0\right\} $ denote the
semigroup of transition kernels of BPS\textsf{ }and write 
\[
\mathcal{L}f(x,v)=\left.\frac{\mathrm{d}}{\mathrm{d}t}f(x+tv,v)\right|_{t=0}+\lambda(x,v)\left[f(x,R(x)v)-f(x,v)\right],
\]
for the infinitesimal generator of BPS where the domain will be discussed
later on.

For any $\epsilon>0$, we write $K^{(\epsilon)}$ for the transition
kernel of the discrete-time BPS, DPBS, with step size $\epsilon>0$.
This kernel satisfies 
\begin{align}
K^{(\epsilon)}\left(\left(x,v\right),\mathrm{d}\left(x^{\prime},v^{\prime}\right)\right) & =\min\left(1,\frac{\pi\left(x+v\epsilon\right)}{\pi\left(x\right)}\right)\delta_{\left(x+v\epsilon,v\right)}\left(\mathrm{d}\left(x^{\prime},v^{\prime}\right)\right)\nonumber \\
 & +\max\left(0,1-\max\left(\frac{\pi\left(x+v\epsilon\right)}{\pi\left(x\right)},\frac{\pi\left(x-R\left(x\right)v\epsilon\right)}{\pi\left(x\right)}\right)\right)\delta_{\left(x,R\left(x\right)v\right)}\left(\mathrm{d}(x^{\prime},v^{\prime})\right)\label{eq:kerneldef}\\
 & +\max\left(0,\min\left(1,\frac{\pi\left(x-R\left(x\right)v\epsilon\right)}{\pi\left(x\right)}\right)-\frac{\pi\left(x+v\epsilon\right)}{\pi\left(x\right)}\right)\delta_{\left(x,-v\right)}\left(\mathrm{d}\left(x^{\prime},v^{\prime}\right)\right).\nonumber 
\end{align}

To keep notation reasonably compact we will often write 
\begin{align*}
K^{(\epsilon)}f(x,v) & =:p_{1}^{(\epsilon)}f(x+v\epsilon,v)+p_{2}^{(\epsilon)}f(x,R(x)v)+[1-p_{1}^{(\epsilon)}(x,v)-p_{2}^{(\epsilon)}(x,v)]f(x,-v),
\end{align*}
with $p^{(i)}(x,v)$ for $i=1,2,3$ the probabilities appearing in
\eqref{eq:kerneldef}.

We will write $\{Z^{(\epsilon)}(k);k\geq0\}$ for the Markov chain
generated by the transition kernel $K^{(\epsilon)}$, with $Z^{(\epsilon)}(0)\sim\rho$.
We also define the càdlàg process $\{\zeta_{t}^{(\epsilon)}:t\in[0,\infty)\}$,
through 
\[
\zeta^{(\epsilon)}(s)=Z_{k}^{(\epsilon)},\qquad s\in\big[k\epsilon,(k+1)\epsilon\big).
\]

We will make the following assumptions. 
\begin{assumption}
\label{assumption:c2plus}The potential function $U:\mathbb{R}^{d}\to[0,\infty)$
is twice continuously differentiable, with absolutely continuous second
derivatives and 
\begin{equation}
\int\mathrm{\exp}\left(-U(x)\right)\left|\nabla U(x)\right|^{2}\mathrm{d}x<\infty.\label{assumption:gradsecondmoment}
\end{equation}
\end{assumption}

\begin{assumption}
\label{assumption:continuousrate}For any $z=(x,v)\in\mathcal{Z}$
the function $t\mapsto\lambda(x+tv,v)$ is continuous.
\end{assumption}

\begin{assumption}
\label{assumption:piboundedder}The probability density function $\pi:\mathbb{R}^{d}\to[0,\infty)$
has bounded, integrable derivatives up to order two and in addition
for some $\varepsilon>0$ we have
\begin{equation}
\int_{\mathbb{R}^{d}}\sup_{\left|y-x\right|<\varepsilon}\left\Vert \Delta\pi(y)\right\Vert \mathrm{d}y<\infty,\label{eq:integrablesuphessian}
\end{equation}
where $\|\Delta f\|$ denotes the operator norm of the Hessian matrix
of $\pi$. 
\end{assumption}

\begin{assumption}
\label{assumption:holder}There exists some $M:\mathbb{R}^{d}\to[0,\infty)$,
such that 
\[
\int\pi(\mathrm{d}x)M(x)<\infty,
\]
and for some $\varepsilon,\delta>0$ we have for all $|v|=1$ and
$r<\varepsilon$ 
\begin{equation}
\left|\nabla U(x+rv)-\nabla U(x)\right|\leq M(x)|r|^{\delta}.\label{ass:holder}
\end{equation}
\end{assumption}

We have the following result. 
\begin{thm}
\label{thm:weakconv} Let $(\epsilon_{n};n\geq1)$ be a positive sequence
such that \textup{$\epsilon_{n}\to0$ as $n\rightarrow\infty$.} Under
Assumptions \ref{assumption:c2plus}, \ref{assumption:continuousrate},
\ref{assumption:piboundedder} and \ref{assumption:holder} the law
of $\{\zeta^{\left(\epsilon_{n}\right)}(\cdot)\}$ converges weakly
to that of BPS as probability measures on $D_{\mathcal{Z}}[0,\infty)$\textup{\emph{
as}}\textup{ $n\rightarrow\infty$.} 
\end{thm}

Before we embark on the proof of Theorem~\ref{thm:weakconv} we prove
some useful properties for the semigroup and the generator. 

\subsection{The Feller property}

Recall that a Markov process taking values in $\mathcal{Z}\subseteq\mathbb{R}^{n}$,
with transition semigroup $\{P^{t}:t\geq0\}$, is called a Feller
process if 
\begin{description}
\item [{(F1)}] for all $t\geq0$ and $f\in C_{0}(\mathcal{Z})$ we have
$P^{t}f\in C_{0}(\mathcal{Z})$, and 
\item [{(F2)}] $P^{t}f(z)\to f(z)$ as $t\to0$ for $f\in C_{0}(\mathcal{Z})$
and $z\in\mathcal{Z}$. 
\end{description}
\begin{lem}
\label{lem:feller} Let Assumptions \ref{assumption:c2plus} and \ref{assumption:continuousrate}
hold. Then $\left\{ P^{t};t\geq0\right\} $ is a Feller semigroup
and the martingale problem for $(\mathcal{L},\rho)$ admits a unique
solution. 
\end{lem}

\begin{proof}
First we prove the uniqueness for the martingale problem assuming
the Feller property. Then we will prove the Feller property. 

Since the semigroup $\{P^{t}:t\geq0\}$ is Feller it follows from
\citep[Theorem~19.6]{K_02} that the semigroup is also strongly continuous,
whence by the Hille-Yosida Theorem (see for example \citep[Theorem~1.2.6]{EK_05})
if follows that $\mathcal{L}$ is \textsl{dissipative}, that is for
any $f\in C_{0}(\mathcal{Z})$ we have 
\[
\|\alpha f-\mathcal{L}f\|_{\infty}\geq\alpha\|f\|_{\infty},\qquad\alpha>0,
\]
that $\mathrm{Dom}(\mathcal{L})$ is dense in $C_{0}(\mathcal{Z})$
and that for some $\alpha>0$, $\mathrm{Range}(\alpha I-\mathcal{L})=C_{0}(\mathcal{Z})$.
Therefore, since $\mathrm{Range}(\alpha I-\mathcal{L})=\overline{\mathrm{Dom}(\mathcal{L})}=C_{0}(\mathcal{Z})$
is separating and $\mathcal{L}$ is dissipative, \citep[Corollary~4.4.4]{EK_05}
implies that uniqueness holds for the martingale problem for $(\mathcal{L},\rho)$,
after one notices that obviously $C_{0}(\mathcal{Z})\subseteq C_{b}(\mathcal{Z})$
the space of continuous bounded functions. 

To complete the proof we now show that BPS is Feller. For $t\geq0$
and $z=(x,v)\in\mathcal{Z}$, write $\Phi_{t}(z)=(x+tv,v)$. To prove
\textbf{(F2)} notice that for any $f\in C_{0}(\mathcal{Z})$ and $z=(x,v)\in\mathcal{Z}$
we have 
\begin{align*}
\left|P^{t}f(x,v)-f(x,v)\right| & \leq\left|f(x+tv,v)-f(x,v)\right|+\mathcal{E}_{t}(x,v)
\end{align*}
where it is clear that for any $z=(x,v)$ we have 
\[
\mathcal{E}_{t}(x,v)\leq2\|f\|_{\infty}\left|1-\exp\left\{ -\int_{0}^{t}\lambda(x+sv,v)\mathrm{d}s\right\} \right|\to0,
\]
as $t\downarrow0$ and \textbf{(F2)} follows easily by continuity
of $f$. 

To prove \textbf{(F1),} following the proof of \citep[Theorem 9.6]{D_93},
for $g\in C_{b}\left(\mathbb{R}_{+}\times\mathcal{Z}\right)$ and
$z=(x,v)$ we define the kernel
\[
Gg(t,z)=\mathbb{E}^{z}\left[f(Z_{t})\mathbb{I}\left\{ t<T_{1}\right\} +g(t-T_{1},Z_{T_{1}})\mathbb{I}\left\{ t\geq T_{1}\right\} \right],
\]
where $T_{1},T_{2},\dots$ are the event times of BPS. We can also
write 
\begin{align}
Gg(t,z) & :=f(\Phi_{t}(z))\exp\left\{ -\int_{0}^{t}\lambda(\Phi_{s}(z))\mathrm{d}s\right\} \nonumber \\
 & \quad+\int_{0}^{t}\int_{\mathcal{Z}}g(t-s,z')Q(\Phi_{s}(z);\mathrm{d}z')\lambda(\Phi_{s}(z))\exp\left\{ -\int_{0}^{s}\lambda(\Phi_{r}(z))\mathrm{d}r\right\} \mathrm{d}s,\label{eq:Gphsi}
\end{align}
where $Q\left((x,v),(\mathrm{d}x',\mathrm{d}v')\right)=\delta_{x}(\mathrm{d}x')\delta_{R(x)v}(\mathrm{d}v')$.
From \citep[Lemma 27.3]{D_93} we have that 
\[
G^{n}g(t,z)=\mathbb{E}^{z}\left[f(Z_{t})\mathbb{I}\left\{ t<T_{n}\right\} +g(t-T_{n},Z_{T_{n}})\mathbb{I}\left\{ t\geq T_{n}\right\} \right],
\]
and for each $t\geq0$, $z\in\mathcal{Z}$ and $g\in C_{b}\left(\mathbb{R}_{+}\times\mathcal{Z}\right)$
we have 
\[
\lim_{n\to\infty}G^{n}g(t,z)=P^{t}f(z).
\]
We will first prove that $Gg\in C_{b}\left(\mathbb{R}_{+}\times\mathcal{Z}\right)$
for any $g\in C_{b}\left(\mathbb{R}_{+}\times\mathcal{Z}\right)$.
Let $(t_{n},z_{n})\to(t,z)$ as $n\to\infty$. Then we have 
\begin{align}
\left|\int_{0}^{t}\lambda(\Phi_{s}(z))\mathrm{d}s-\int_{0}^{t_{n}}\lambda(\Phi_{s}(z_{n}))\mathrm{d}s\right| & \leq\int_{t\wedge t_{n}}^{t\vee t_{n}}\lambda(\Phi_{s}(z))\mathrm{d}s+\int_{0}^{t_{n}}\left|\lambda(\Phi_{s}(z_{n}))-\lambda(\Phi_{s}(z))\right|\mathrm{d}s.\label{eq:continuityoflambdaintegral}
\end{align}
Both integrals vanish by bounded convergence, since by continuity
of $\lambda$ and $\phi$ the second integrand vanishes pointwise,
while both integrands are bounded by boundedness of the flow $\Phi_{s}(z):s\in[0,t]\}$
and continuity of $\lambda$. On the other hand letting 
\[
\Psi(t,z):=\int_{0}^{t}\int_{\mathcal{Z}}g(t-s,z')Q(\Phi_{s}(z);\mathrm{d}z')\lambda(\Phi_{s}(z))\exp\left\{ -\int_{0}^{s}\lambda(\Phi_{r}(z))\mathrm{d}r\right\} \mathrm{d}s
\]
we have 
\begin{align*}
 & \left|\Psi(t_{n},z_{n})-\Psi(t,z)\right|\\
 & \leq\int_{t\wedge t_{n}}^{t\vee t_{n}}\int_{\mathcal{Z}}g(t-s,z')Q(\Phi_{s}(z);\mathrm{d}z')\lambda(\Phi_{s}(z))\exp\left\{ -\int_{0}^{s}\lambda(\Phi_{r}(z))\mathrm{d}r\right\} \mathrm{d}s\\
 & \quad+\left\Vert g\right\Vert _{\infty}\int_{0}^{t_{n}}\left|\lambda(\Phi_{s}(z))\exp\left\{ -\int_{0}^{s}\lambda(\Phi_{r}(z))\mathrm{d}r\right\} -\lambda(\Phi_{s}(z_{n}))\exp\left\{ -\int_{0}^{s}\lambda(\Phi_{r}(z_{n}))\mathrm{d}r\right\} \right|\mathrm{d}s\\
 & \quad+\int_{0}^{t_{n}}\lambda(\Phi_{s}(z))\exp\left\{ -\int_{0}^{s}\lambda(\Phi_{r}(z))\mathrm{d}r\right\} \left|\int_{\mathcal{Z}}g(t-s,z')Q(\Phi_{s}(z);\mathrm{d}z')-\int_{\mathcal{Z}}g(t-s,z')Q(\Phi_{s}(z_{n});\mathrm{d}z')\right|\mathrm{d}s.
\end{align*}
Again all three integrals vanish by bounded convergence, by the boundedness
of the orbits of the flow, where for the second we also use the continuity
of $z\mapsto\lambda(\Phi_{s}(z))$ and \eqref{eq:continuityoflambdaintegral},
whereas for the third one we use the continuity of the transition
kernel $Q(z,\mathrm{d}\cdot)$. 

We have thus shown that $Gg$ defined in (\ref{eq:Gphsi}) is continuous.
In addition since $g$ is bounded, it follows that 
\begin{align*}
\left|\Psi(t,z)\right| & \leq C\int_{0}^{t}\lambda(\Phi_{s}(z))\exp\left\{ -\int_{0}^{s}\lambda(\Phi_{r}(z))\mathrm{d}r\right\} \mathrm{d}s\leq C
\end{align*}
since the integrand defines a probability density function. Since
$f$ is also bounded, it thus follows that $Gg(t,z)\in C_{b}\left(\mathbb{R}_{+}\times\mathcal{Z}\right)$.
Therefore $G^{n}g(t,z)$ will also be continuous and bounded.

Finally, recall from the proof of \citep[Lemma 9.3]{D_93} that 
\[
\left|G^{n}g(t,z)-P^{t}f(z)\right|\leq2C\mathbb{P}_{z}\left(t\geq T_{n}\right)\to0,
\]
as $n\to\infty$, where $T_{n}$ is the time of $n$-th event, when
BPS starts from $z$. Suppose now that $z=(x,v)\in B_{R'}(0)\subset\mathcal{Z}$,
the ball of radius $R'$ around the origin. Then by construction of
BPS there will be a compact set $K_{R'}\subset\mathcal{Z}$, such
that $\left\{ Z_{s}=\left(X_{s},V_{s}\right);0\leq s\leq t\right\} \subset K_{R'}$.
Therefore, since $\lambda$ is locally bounded, we have that 
\[
\sup_{s\leq t}\lambda(Z_{s})\leq\sup_{w\in K_{R'}}\lambda(w)=:\bar{\lambda}<\infty.
\]
Thus $\mathbb{P}_{z}\left(T_{n}\leq t\right)\leq\mathbb{P}\left(T_{n}'\leq t\right)$
where $T_{n}'$ are the event times of a Poisson process with rate
$\bar{\lambda}$. Therefore 
\[
\sup_{z\in B_{R'}(0)}\left|G^{n}g(t,z)-P^{t}f(z)\right|\leq2C\sup_{z\in B_{R'}(0)}\mathbb{P}_{z}\left(t\geq T_{n}\right)\leq2C\mathbb{P}\left(T_{n}'\leq t\right)\to0,
\]
as $n\to\infty$. It follows that $G^{n}g(t,z)\to P^{t}f(z)$ as $n\to\infty$
uniformly on compact sets. Since the functions $z\mapsto G^{n}g(t,z)$
are continuous, it follows that $P^{t}f(z)$ is continuous on every
compact set and thus is continuous.

Finally let $f\in C_{0}(\mathcal{Z})$. Thus for any $\epsilon>0$,
there exists a compact set $K(\epsilon)\subset\mathcal{Z}$ such that
$|f(z)|<\epsilon$ for all $z\notin K(\epsilon)$. By assumption the
velocity component lives in $B_{R}(\mathbb{R}^{d})$ . The reason
for this restriction is that otherwise there is always the chance
of coming back from infinity at finite time which implies then that
$P^{t}$ does not leave $C_{0}(\mathcal{Z})$ invariant. Thus if $z^{(n)}\to\infty$
we must have $x^{(n)}\to\infty$. Then, given $\epsilon>0$ choose
$K$ large enough so that 
\[
\sup_{|x|>K,v}\left|f(x,v)\right|<\epsilon.
\]
Then choose $K'>K+Rt$ and $N$, such that for all $n\geq N$ we have
$\left|x^{(n)}\right|\geq K'$. Then since $X_{t}\in B\left(x^{(n)},Rt\right)$
it follows that $|X_{t}|>K$ and thus $|f(X_{t},V_{t})|\leq\epsilon$.
Since $\epsilon>0$ is arbitrary the result follows. 
\end{proof}

\subsection{Preliminary calculations}

We first need precise estimates for $p^{(i)}(x,v)$, $i=2,3$ and
small $\epsilon$. We will often use the formula 
\begin{align}
1-\frac{\pi(x+sv)}{\pi(x)} & =1-\exp\{-U(x+sv)-U(x)\}=\int_{0}^{s}\frac{\pi(x+rv)}{\pi(x)}\langle\nabla U(x+rv),v\rangle\mathrm{d}r.\label{eq:FTCformula}
\end{align}

\subsubsection*{The probability $p_{2}^{(\epsilon)}(x,v)$.}

Recall that 
\[
\langle\nabla U(x),-R(x)v\rangle=-\langle\nabla U(x),R(x)v\rangle=-\langle\nabla U(x),-v\rangle=\langle\nabla U(x),v\rangle.
\]
Therefore if $\langle U(x),v\rangle\leq0$, we will have $p_{\epsilon}^{(2)}(x,v)=0$,
for all $\epsilon$ small enough. Thus we can assume that $\langle U(x),v\rangle>0$
in which case we also have $\langle U(x),-R(x)v\rangle>0$ and therefore
for all $\epsilon>0$ small enough we have that $U(x+\epsilon v),U(x-R(x)v\epsilon)>U(x)$
and thus 
\[
\max\left\{ \frac{\pi(x+\epsilon v)}{\pi(x)},\frac{\pi(x-R(x)\epsilon v)}{\pi(x)}\right\} <1.
\]
In this case, using \eqref{eq:FTCformula}, we estimate 
\begin{align*}
p_{2}^{(\epsilon)}(x,v) & =1-\max\left\{ \frac{\pi(x+\epsilon v)}{\pi(x)},\frac{\pi(x-R(x)\epsilon v)}{\pi(x)}\right\} \\
 & =\min\left\{ 1-\exp\left[-\int_{0}^{\epsilon}\langle\nabla U(x+sv),v\rangle\mathrm{d}s\right],1-\exp\left[-\int_{0}^{\epsilon}\langle\nabla U(x-sR(x)v),v\rangle\mathrm{d}s\right]\right\} \\
 & =\min\left\{ \int_{0}^{\epsilon}\frac{\pi(x+sv)}{\pi(x)}\langle\nabla U(x+sv),v\rangle\mathrm{d}s,\int_{0}^{\epsilon}\frac{\pi(x-sR(x)v)}{\pi(x)}\langle\nabla U(x-sR(x)v),v\rangle\mathrm{d}s\right\} \\
 & =\min\bigg\{\epsilon\langle\nabla U(x),v\rangle+\int_{0}^{\epsilon}\left[\frac{\pi(x+sv)}{\pi(x)}\langle\nabla U(x+sv),v\rangle-\langle\nabla U(x),v\rangle\right]\mathrm{d}s,\\
 & \qquad\qquad\epsilon\langle\nabla U(x),-R(x)v\rangle+\int_{0}^{\epsilon}\left[\frac{\pi(x-sR(x)v)}{\pi(x)}\langle\nabla U(x-sR(x)v),v\rangle-\langle\nabla U(x),v\rangle\right]\mathrm{d}s\bigg\}\\
 & =\epsilon\langle\nabla U(x),v\rangle+\min\bigg\{\int_{0}^{\epsilon}\left[\frac{\pi(x+sv)}{\pi(x)}\langle\nabla U(x+sv),v\rangle-\langle\nabla U(x),v\rangle\right]\mathrm{d}s,\\
 & \qquad\qquad\qquad\qquad\int_{0}^{\epsilon}\left[\frac{\pi(x-sR(x)v)}{\pi(x)}\langle\nabla U(x-sR(x)v),v\rangle-\langle\nabla U(x),v\rangle\right]\mathrm{d}s\bigg\}\\
 & =\epsilon\langle\nabla U(x),v\rangle+\mathcal{E}_{1}^{(\epsilon)}(x,v).
\end{align*}
Since we have assumed that for $\epsilon$ small enough we have $\pi(x-sR(x)v)<\pi(x)$,
then we have for $\epsilon$ small enough
\begin{align}
\left|\mathcal{E}_{1}^{(\epsilon)}(x,v)\right| & \leq\int_{0}^{\epsilon}\frac{\pi(x-sR(x)v)}{\pi(x)}\left|\langle\nabla U(x-sR(x)v),v\rangle-\langle\nabla U(x),v\rangle\right|\mathrm{d}s\nonumber \\
 & \leq\int_{0}^{\epsilon}\left|\nabla U(x-sR(x)v)-\nabla U(x)\right|\mathrm{d}s\nonumber \\
 & \leq\int_{0}^{\epsilon}M(x)s^{\delta}\left|v\right|^{\delta}\mathrm{d}s=CM(x)\left|v\right|^{\delta+1}\epsilon^{\delta+1},\label{eq:boundonerror1}
\end{align}
where we used Assumption \ref{assumption:holder}. Overall we have
that 
\begin{equation}
p_{2}^{(\epsilon)}(x,v)=\epsilon\max\{\langle\nabla U(x),v\rangle,0\}+CM(x)\left|v\right|^{\delta+1}\epsilon^{\delta+1}.\label{eq:boundonp2}
\end{equation}

\subsection{Proof of Theorem~\ref{thm:weakconv}}

Let $\epsilon_{n}\to0$. To ease notation we will write $\zeta^{(n)}$
rather than $\zeta^{(\epsilon_{n})}$. Define (see \citep[Remark 8.3(b)]{EK_05})
\begin{align}
\xi_{n}(t) & :=\epsilon_{n}^{-1}\int_{0}^{\epsilon_{n}}\mathbb{E}\left[\left.f\left(\zeta^{(n)}(t+s)\right)\right|\mathcal{G}_{t}^{n}\right]\mathrm{d}s,\label{eq:xidef}\\
\phi_{n}(t) & :=\epsilon_{n}^{-1}\mathbb{E}\left[\left.f\left(\zeta^{(n)}(t+\epsilon_{n})\right)-f\left(\zeta^{(n)}(t)\right)\right|\mathcal{G}_{t}^{n}\right],\label{eq:phidef}
\end{align}
where $\mathcal{G}_{t}^{n}:=\sigma\left(\zeta^{(n)}(s):s\leq t\right)$,
the natural filtration of $\left\{ \zeta^{(n)}(t):t\geq0\right\} $.
Recall that $\zeta^{(n)}(0)\sim\rho$ for all $n$. 

Since $\left\{ P^{t}:t\geq0\right\} $ is strongly continuous, we
have that $\mathcal{L}:\mathrm{Dom}\left(\mathcal{L}\right)\subset C_{0}(\mathcal{Z})\mapsto C_{0}(\mathcal{Z})$
is densely defined. Thus we can think of $\mathcal{L}$ as a subset
of $C_{0}(\mathcal{Z})\times C_{0}(\mathcal{Z})$ and therefore as
a subset of $C_{b}(\mathcal{Z})\times C_{b}(\mathcal{Z})$. For our
purposes we will define $\mathcal{L}$ on the space 
\[
D:=C_{c}^{\infty}(\mathcal{Z}):=\left\{ f:\mathcal{Z}\mapsto\mathcal{R}\,\text{infinitely differentiable with compact support}\right\} ,
\]
which is clearly a subset of $\mathrm{Dom}(\mathcal{L})$. Therefore
we will be working with the restricted generator $\left.\mathcal{L}\right|_{D}$.
Therefore \citep[Corollary~8.15 of Chapter~4]{EK_05} applies to our
scenario. Notice that \citep[Corollary~8.15 of Chapter~4]{EK_05}
does not require $D$ to be a core of the generator. 

To apply \citep[Corollary~8.15 of Chapter~4]{EK_05} we need to check
the following: 
\begin{itemize}
\item \textbf{Compact Containment:} For every $\eta>0$ and $T>0$ there
is a compact set $\rho_{\eta,T}\subset\mathcal{Z}$ such that 
\begin{equation}
\inf_{n}\mathbb{P}\left\{ \zeta^{(\epsilon_{n})}(t)\in\rho_{\eta,T},\,\,\text{for all \ensuremath{0\leq t\leq T}}\}\right\} \geq1-\eta.\label{cond:relcompact}
\end{equation}
\item \textbf{Separating algebra:} the closure of the linear span of $D$
contains an algebra that separates points;
\item \textbf{Martingale problem:} the martingale problem in $D_{E}([0,\infty))$
for $(\mathcal{L},\pi)$ admits at most one solution; this has already
been established in Lemma~\ref{lem:feller}.
\item \textbf{Generator~~convergence:} for each $f\in\mathcal{D}(\mathcal{L})$
and $T>0$, for $\xi_{n},\phi_{n}$ as defined in \eqref{eq:xidef},\eqref{eq:phidef}
\begin{align}
\sup_{n}\sup_{s\leq T}\mathbb{E}[|\xi^{(n)}(s)|] & <\infty\label{eq:firstcondition}\\
\sup_{n}\sup_{s\leq T}\mathbb{E}[|\phi^{(n)}(s)|] & <\infty\label{eq:secondcondition}\\
\lim_{n\to\infty}\mathbb{E}\left[\left|\xi^{(n)}(t)-f(X^{(n)}(t)\right|\right] & =0,\label{eq:thirdcondition}\\
\lim_{n\to\infty}\mathbb{E}\left[\left|\phi^{(n)}(t)-\mathcal{L}f(X^{(n)}(t)\right|\right] & =0,\label{eq:fourthcondition}
\end{align}
and in addition 
\begin{equation}
\lim_{n\to\infty}\mathbb{E}\left\{ \sup_{t\in\mathbb{Q}\cap[0,T]}|\xi_{n}(t)-f(X_{n}(t))|\right\} =0,\label{eq:833}
\end{equation}
and for some $p>1$ 
\begin{equation}
\sup_{n\to\infty}\mathbb{E}\left[\left(\int_{0}^{T}|\phi_{n}(s)|^{p}\mathrm{d}s\right)^{1/p}\right]<\infty.\label{eq:834}
\end{equation}
\end{itemize}
We will apply the theorem to the sequence of processes $X^{(n)}(\cdot)=\zeta^{(n)}(\cdot)$
with $\xi_{n},\phi_{n}$ as defined in \eqref{eq:xidef},\eqref{eq:phidef}.

\subsubsection{Compact Containment}

Let $\eta>0$, $T>0$ be arbitrary. We need to provide a compact set
$\rho_{\eta,T}\subset\mathcal{Z}$ such that \eqref{cond:relcompact}
holds. Let $\zeta^{(\epsilon_{n})}(0)=(X_{0},V_{0})\sim\rho$. Then
notice that for all $t\leq T$, the first component component will
of $\zeta^{(\epsilon_{n})}(t)$ will take on the values $X^{(\epsilon_{n})}(k)$
for $k$ ranging from $0$ up to $\lceil T/\epsilon_{n}\rceil$, while
the second component $V_{k}$ will only change in direction through
the reflection and negation steps, while the modulus will remain fixed
at $|V_{0}|$. From the definition of $K^{(\epsilon_{n})}$ we thus
know that for any $n$, for each $k$ we have that 
\[
\left|X_{k}^{(\epsilon_{n})}\right|\leq|X_{0}|+k\epsilon_{n}|V_{0}|.
\]
Let $R>0$ be large enough so that 
\[
\rho\left\{ B_{R}(0)\times B_{R}(0)\right\} \geq1-\eta
\]
and define 
\[
\rho_{\eta,T}:=B_{R(1+T)}(0)\times B_{R}(0)\subset\mathcal{Z}.
\]
It is then clear that 
\begin{align*}
\mathbb{P}\left\{ \zeta^{(\epsilon_{n})}(t)\in\rho_{\eta,T}\,\,\,\,\text{for all \ensuremath{0\leq t\leq T}}\right\}  & \geq\mathbb{P}\left\{ \zeta^{(\epsilon_{n})}(0)\in B_{R}(0)\times B_{R}(0)\right\} \geq1-\eta.
\end{align*}

\subsubsection{Separating Algebra. }

This holds since $C_{c}^{\infty}(\mathcal{Z})$ is dense in $C_{c}(\mathcal{Z})$,
continuous functions of compact support, which is in turn dense in
$C_{0}(\mathcal{Z})$ which is an algebra that separates points.

\subsubsection{Convergence of generators.}

Recall that for $f\in\mathcal{D}(\mathcal{L})$ 
\begin{align*}
\xi_{n}(t) & :=\epsilon_{n}^{-1}\int_{0}^{\epsilon_{n}}\mathbb{E}\left[\left.f\left(\zeta^{(\epsilon_{n})}(t+s)\right)\right|\mathcal{G}_{t}^{n}\right]\mathrm{d}s,\\
\phi_{n}(t) & :=\epsilon_{n}^{-1}\mathbb{E}\left[\left.f\left(\zeta^{(\epsilon_{n})}(t+\epsilon_{n})\right)-f\left(\zeta^{(\epsilon_{n})}(t)\right)\right|\mathcal{G}_{t}^{n}\right].
\end{align*}
Conditions~\eqref{eq:firstcondition},\eqref{eq:secondcondition}
are automatically satisfied by stationarity.

\paragraph{Conditions~\eqref{eq:thirdcondition},\eqref{eq:833}.}

Since \eqref{eq:833} implies \eqref{eq:thirdcondition} we only need
to check \eqref{eq:833}.

Let $t\in[0,T]$ and $k:=\lfloor t/\epsilon\rfloor$. Since for $t+s\leq(k+1)\epsilon$
we have $\zeta^{(\epsilon_{n})}(t+s)=Z^{(\epsilon_{n})}(k)$ it follows
that 
\begin{align*}
\left|\xi_{n}(t)-f(\zeta^{(\epsilon_{n})}(t)\right| & =\left|\epsilon_{n}^{-1}\int_{0}^{\epsilon_{n}}\left\{ \mathbb{E}\left[\left.f\left(\zeta^{(\epsilon_{n})}(t+s)\right)\right|\mathcal{G}_{t}^{n}\right]-f\left(\zeta^{(\epsilon_{n})}(t)\right)\right\} \mathrm{d}s\right|\\
 & =\left|\epsilon_{n}^{-1}\left[(k+1)\epsilon_{n}-t\right]\left\{ \mathbb{E}\left[\left.f\left(Z^{(\epsilon_{n})}(k+1)\right)\right|\mathcal{G}_{t}^{n}\right]-f\left(Z^{(\epsilon_{n})}(k)\right)\right\} \right|\\
 & \leq\left|K^{(\epsilon_{n})}f\left(Z^{(\epsilon_{n})}(k)\right)-f\left(Z^{(\epsilon_{n})}(k)\right)\right|.
\end{align*}
Therefore, we can estimate 
\begin{align}
 & \mathbb{E}\left\{ \sup_{t\in\mathbb{Q}\cap[0,T]}|\xi_{n}(t)-f(\zeta^{(\epsilon_{n})}(t))|\right\} \nonumber \\
 & \leq\mathbb{E}\left[\sup_{k\leq T/\epsilon}\left|K^{(\epsilon)}f\left(Z^{(\epsilon_{n})}(k)\right)-f\left(Z^{(\epsilon_{n})}(k)\right)\right|\right]\nonumber \\
 & \leq\mathbb{E}\left[\sup_{k\leq T/\epsilon}p_{1}^{(\epsilon)}\left(X^{(\epsilon_{n})}(k),V^{(\epsilon_{n})}(k)\right)\left|f\left(X^{(\epsilon_{n})}(k)+\epsilon_{n}V^{(\epsilon_{n})}(k),V^{(\epsilon_{n})}(k)\right)-f\left(X^{(\epsilon_{n})}(k),V^{(\epsilon_{n})}(k)\right)\right|\right]\nonumber \\
 & +\mathbb{E}\left[\sup_{k\leq T/\epsilon}p_{2}^{(\epsilon)}\left(X^{(\epsilon_{n})}(k),V^{(\epsilon_{n})}(k)\right)\left|f\left(X^{(\epsilon_{n})}(k),R(X^{(\epsilon_{n})}(k))V^{(\epsilon_{n})}(k)\right)-f\left(X^{(\epsilon_{n})}(k),V^{(\epsilon_{n})}(k)\right)\right|\right]\nonumber \\
 & +\mathbb{E}\left[\sup_{k\leq T/\epsilon}p_{3}^{(\epsilon)}\left(X^{(\epsilon_{n})}(k),V^{(\epsilon_{n})}(k)\right)\left|f\left(X^{(\epsilon_{n})}(k),-V^{(\epsilon_{n})}(k)\right)-f\left(X^{(\epsilon_{n})}(k),V^{(\epsilon_{n})}(k)\right)\right|\right]\nonumber \\
 & \leq\epsilon_{n}\mathbb{E}\left[\sup_{k\leq T/\epsilon}\left|V^{(\epsilon_{n})}(k)\right|\right]\sup|\nabla f|\nonumber \\
 & +\sup|f|\left(\mathbb{E}\left[\sup_{k\leq T/\epsilon}p_{2}^{(\epsilon)}\left(X^{(\epsilon_{n})}(k),V^{(\epsilon_{n})}(k)\right)\right]+\mathbb{E}\left[\sup_{k\leq T/\epsilon}p_{3}^{(\epsilon)}\left(X^{(\epsilon_{n})}(k),V^{(\epsilon_{n})}(k)\right)\right]\right),\label{eq:termstwoandthree}
\end{align}
by a simple Taylor expansion, since $f$ and $|\nabla f|$ are bounded.

Since $\left|V^{(\epsilon_{n})}(k)\right|\leq R$ for all $k$ the
first term clearly vanishes. In addition by \eqref{eq:boundonp2},
\eqref{ass:holder} and stationarity it follows that 
\begin{align*}
\mathbb{E}\left[\sup_{k\leq T/\epsilon}p_{2}^{(\epsilon_{n})}\left(X^{(\epsilon_{n})}(k),V^{(\epsilon_{n})}(k)\right)\right] & \leq\mathbb{E}\left[\sup_{K\leq T/\epsilon_{n}}\left|\epsilon_{n}\nabla U\left(X^{(\epsilon_{n})}(k)\right)\right|\right]+C\epsilon_{n}^{\delta+1}\mathbb{E}\left[\sup_{K\leq T/\epsilon_{n}}M\left(X^{(\epsilon_{n})}(k)\right)\right]\\
 & \leq\left\{ \mathbb{E}\left[\sup_{K\leq T/\epsilon_{n}}\left|\epsilon_{n}\nabla U\left(X^{(\epsilon_{n})}(k)\right)\right|^{2}\right]\right\} ^{1/2}+C\epsilon_{n}^{\delta+1}\sum_{k=1}^{T/\epsilon_{n}}\mathbb{E}\left[M\left(X^{(\epsilon_{n})}(k)\right)\right]\\
 & \leq\left\{ \sum_{k=1}^{T/\epsilon_{n}}\epsilon_{n}^{2}\mathbb{E}\left[\left|\nabla U\left(X^{(\epsilon_{n})}(0)\right)\right|^{2}\right]\right\} ^{1/2}+C\epsilon_{n}^{\delta}T\mathbb{E}\left[M\left(X^{(\epsilon_{n})}(0)\right)\right]\\
 & \leq\left\{ T\epsilon_{n}\pi\left[\left|\nabla U\right|^{2}\right]\right\} ^{1/2}+C\epsilon_{n}^{\delta}\mathbb{E}\left[M\left(X^{(\epsilon_{n})}(0)\right)\right]=o(1),
\end{align*}
by Assumption \ref{assumption:gradsecondmoment}. 

To control the last term of \eqref{eq:termstwoandthree}, again by
stationarity we have 
\begin{align*}
\mathbb{E}\left[\sup_{k\leq T/\epsilon}p_{3}^{(\epsilon)}\left(X^{(\epsilon_{n})}(k),V^{(\epsilon_{n})}(k)\right)\right] & \leq\sum_{k=1}^{T/\epsilon_{n}}\mathbb{E}\left[p_{3}^{(\epsilon)}\left(X^{(\epsilon_{n})}(k),V^{(\epsilon_{n})}(k)\right)\right]=\frac{T}{\epsilon_{n}}\mathbb{E}\left[p_{3}^{(\epsilon)}\left(X^{(\epsilon_{n})}(0),V^{(\epsilon_{n})}(0)\right)\right].
\end{align*}
Letting $(X,V)\sim\rho$ we thus have 
\begin{align}
\frac{T}{\epsilon_{n}}\mathbb{E}\left[p_{3}^{(\epsilon_{n})}\left(X^{(\epsilon_{n})}(0),V^{(\epsilon_{n})}(0)\right)\right] & =\frac{T}{\epsilon_{n}}\mathbb{E}\left[\max\left\{ 0,\min\left(1-\frac{\pi\left(X+\epsilon_{n}V\right)}{\pi\left(X\right)},\frac{\pi\left(X-\epsilon_{n}R\left(X\right)V\right)}{\pi\left(X\right)}-\frac{\pi\left(X+\epsilon_{n}V\right)}{\pi\left(X\right)}\right)\right)\right]\nonumber \\
 & \leq\frac{T}{\epsilon_{n}}\mathbb{E}\left[\max\left\{ 0,\min\left(\left|1-\frac{\pi\left(X+\epsilon_{n}V\right)}{\pi\left(X\right)}\right|,\left|\frac{\pi\left(X-\epsilon_{n}R\left(X\right)V\right)}{\pi\left(X\right)}-\frac{\pi\left(X+\epsilon_{n}V\right)}{\pi\left(X\right)}\right|\right)\right)\right]\nonumber \\
 & \leq\frac{T}{\epsilon_{n}}\mathbb{E}\left[\left|\frac{\pi\left(X-\epsilon_{n}R\left(X\right)V\right)}{\pi\left(X\right)}-\frac{\pi\left(X+\epsilon_{n}V\right)}{\pi\left(X\right)}\right|\right]\nonumber \\
 & =\frac{T}{\epsilon_{n}}\int\psi(\mathrm{\mathrm{d}}v)\int\mathrm{d}x\left|\pi\left(x-\epsilon_{n}R\left(x\right)v\right)-\pi\left(x+\epsilon_{n}v\right)\right|\nonumber \\
 & =T\int\psi(\mathrm{\mathrm{d}}v)\int\mathrm{d}x\frac{\left|\pi\left(x-\epsilon_{n}R\left(x\right)v\right)-\pi\left(x+\epsilon_{n}v\right)\right|}{\epsilon_{n}}.\label{eq:boundonp3-1}
\end{align}
Notice that 
\begin{align*}
\frac{\pi\left(x-\epsilon_{n}R\left(x\right)v\right)-\pi\left(x+\epsilon_{n}v\right)}{\epsilon_{n}} & =\epsilon_{n}^{-1}\left[\pi(x)-\epsilon_{n}\pi(x)\langle-U(x),R(x)v)-\pi(x)-\epsilon_{n}\pi(x)\langle-U(x),R(x)v)+o(\epsilon_{n})\right]\\
 & =\epsilon_{n}^{-1}\left[-\epsilon_{n}\pi(x)\langle U(x),v)+\epsilon_{n}\pi(x)\langle U(x),R(x)v)+o(\epsilon_{n})\right],
\end{align*}
and thus the integrand vanishes pointwise. Let $F:\mathbb{R}\to\mathbb{R}$
have bounded continuous second derivatives. Then letting $G(s):=F(s)-F(0)-sF'(0)$
we have by the Mean Value Theorem for some $\xi\in[0,s]$ with $s>0$
\[
\frac{G(s)-G(0)}{s^{2}}=\frac{G'(\xi)}{s}=\frac{G'(0)}{s}+\frac{G'(\xi)-G'(0)}{s}=\frac{G'(0)}{s}+\frac{G''(\xi')}{s}\xi.
\]
Thus it follows that
\begin{align}
\pi\left(x-\epsilon_{n}R\left(x\right)v\right)-\pi\left(x+\epsilon_{n}v\right) & =\pi\left(x-\epsilon_{n}R\left(x\right)v\right)-\pi(x)+\pi(x)-\pi\left(x+\epsilon_{n}v\right)\nonumber \\
 & =\left\langle \nabla\pi(x),-R(x)v\right\rangle \epsilon_{n}+\left\langle R(x)v,\Delta\pi(x+\xi_{1}'v)R(x)v\right\rangle \epsilon_{n}\xi_{2}''\nonumber \\
 & \qquad\qquad-\left.\left\langle \nabla\pi(x),v\right\rangle \epsilon_{n}+\left\langle R(x)v,\Delta\pi(x+\xi_{1}'v)R(x)v\right\rangle \epsilon_{n}\xi_{2}''\right],\label{eq:p3expansion}
\end{align}
for $\xi_{i}',\xi_{i}''\in[0,\epsilon_{n}]$ for $i=1,2$. Since 
\[
\left\langle \nabla\pi(x),-R(x)v\right\rangle =\pi(x)\left\langle -\nabla U(x),-R(x)v\right\rangle =-\pi(x)\left\langle \nabla U(x),v\right\rangle =-\left\langle \nabla\pi(x),v\right\rangle ,
\]
it follows that for $n$ large enough so that $\epsilon_{n}R<\varepsilon$
\begin{align*}
\epsilon_{n}^{-1}\left|\pi\left(x-\epsilon_{n}R\left(x\right)v\right)-\pi\left(x+\epsilon_{n}v\right)\right| & \leq2C\epsilon_{n}\sup_{y\in B(x,\varepsilon)}\left\Vert \Delta\pi(y)\right\Vert ,
\end{align*}
which is integrable by Assumption \ref{assumption:piboundedder}.
Thus by dominated convergence it follows that the last term of \eqref{eq:termstwoandthree}
vanishes and thus \eqref{eq:833} holds.

\paragraph{Condition~\eqref{eq:fourthcondition}.}

Letting $k:=\lfloor t/\epsilon\rfloor$ and $(X,V)\sim\rho$, we have
by stationarity
\begin{align*}
\lefteqn{} & \mathbb{E}\left[\left|\phi^{(n)}(t)-\mathcal{L}f\left(X^{(n)}(t)\right)\right|\right]\\
 & =\mathbb{E}\left[\left|\epsilon_{n}^{-1}\mathbb{E}\left[\left.f\left(\zeta^{(\epsilon_{n})}(t+\epsilon_{n})\right)-f\left(\zeta^{(\epsilon_{n})}(t)\right)\right|\mathcal{G}_{t}^{n}\right]-\mathcal{L}f(\zeta^{(\epsilon_{n})}(t))\right|\right]\\
 & =\mathbb{E}\left\{ \left|\epsilon_{n}^{-1}\left[K^{(\epsilon_{n})}f\big(Z^{(\epsilon_{n})}(k)\big)-f\big(Z^{(\epsilon_{n})}(k)\big)\right]-\mathcal{L}f\big(Z^{(\epsilon_{n})}(k)\big)\right|\right\} \\
 & =\mathbb{E}\left\{ \left|\epsilon_{n}^{-1}\left[K^{(\epsilon_{n})}f\big(Z^{(\epsilon_{n})}(k)\big)-f\big(Z^{(\epsilon_{n})}(k)\big)\right]-\mathcal{L}f\big(Z^{(\epsilon_{n})}(k)\big)\right|\right\} \\
 & =\mathbb{E}\left\{ \left|\epsilon_{n}^{-1}p_{1}^{(\epsilon_{n})}(X,V)\left[f\big(X+\epsilon_{n}V,V\big)-f\big(X,V\big)\right]-\langle\nabla f(X,V),V\rangle\right|\right\} \\
 & \qquad+\mathbb{E}\left\{ \left|\left[\epsilon_{n}^{-1}p_{2}^{(\epsilon_{n})}(X,V)-\lambda(X,V)\right]\times\left[f\big(X,R(X)V\big)-f\big(X,V\big)\right]\right|\right\} \\
 & \qquad\qquad+\mathbb{E}\left\{ \epsilon_{n}^{-1}p_{3}^{(\epsilon_{n})}(X,V)\left|f\big(X,-V\big)-f\big(X,V\big)\right|\right\} \\
 & =I_{1}+I_{2}+I_{3}.
\end{align*}
From \eqref{eq:boundonp2} and the fact that $f$ is assumed bounded
it easily follows that $I_{2}\to0$. Also we proved that $I_{3}\to0$
while checking Condition~\eqref{eq:833}. Therefore we just have
to handle $I_{1}$. We start with the triangle inequality 
\begin{align*}
I_{1} & \leq\mathbb{E}\left\{ \left|\epsilon_{n}^{-1}\left[f\big(X+\epsilon_{n}V,V\big)-f\big(X,V\big)\right]-\langle\nabla f(X,V),V\rangle\right|\right\} +\mathbb{E}\left\{ \left|p_{1}^{(\epsilon_{n})}(X,V)-1\right|\right\} .
\end{align*}
The first term vanishes by continuity of $\nabla f$ and bounded convergence,
while for the second term we have 

\begin{align*}
\mathbb{E}\left\{ \left|p_{1}^{(\epsilon_{n})}(X,V)-1\right|\right\}  & =\mathbb{E}\left\{ \left|\min\left\{ 1,\frac{\pi\left(X+\epsilon_{n}V\right)}{\pi(X)}\right\} -1\right|\right\} \\
 & \leq\mathbb{E}\left\{ \left|\frac{\pi\left(X+\epsilon_{n}V\right)}{\pi(X)}-1\right|\right\} \\
 & =\int\rho(\mathrm{d}x,\mathrm{d}v)\left|\frac{\pi\left(x+\epsilon_{n}v\right)}{\pi(x)}-1\right|=\int\psi(\mathrm{d}v)\int\left|\pi\left(x+\epsilon_{n}v\right)-\pi\left(x\right)\right|\mathrm{d}x\to0,
\end{align*}
by dominated convergence.

\paragraph{Condition~\eqref{eq:834}.}

Notice that for $s\in[k\epsilon_{n},(k+1)\epsilon_{n})$ we have 
\[
\phi_{n}(s)=\epsilon_{n}^{-1}\left[K^{(\epsilon_{n})}f\left(Z^{(\epsilon_{n})}(k)\right)-f\left(Z^{(\epsilon_{n})}(k)\right)\right].
\]
Thus for $p>1$ by the $C_{r}$-inequality
\begin{align*}
\int_{0}^{T}|\phi_{n}(s)|^{p}\mathrm{d}s & =\sum_{k=0}^{\lfloor T/\epsilon_{n}\rfloor}\epsilon_{n}\left|\epsilon_{n}^{-1}\left[K^{(\epsilon_{n})}f\left(Z^{(\epsilon_{n})}(k)\right)-f\left(Z^{(\epsilon_{n})}(k)\right)\right]\right|^{p}\\
 & =\sum_{k=0}^{\lfloor T/\epsilon_{n}\rfloor}\epsilon_{n}^{1-p}\left|\left[K^{(\epsilon_{n})}f\left(Z^{(\epsilon_{n})}(k)\right)-f\left(Z^{(\epsilon_{n})}(k)\right)\right]\right|^{p}\\
 & \leq C\sum_{k=0}^{\lfloor T/\epsilon_{n}\rfloor}\epsilon_{n}^{1-p}\left[p_{1}^{(\epsilon_{n})}\left(Z^{(\epsilon_{n})}(k)\right)\right]^{p}\left|f\left(X^{(\epsilon_{n})}(k)+\epsilon_{n}V^{(\epsilon_{n})}(k),V^{(\epsilon_{n})}(k)\right)-f\left(Z^{(\epsilon_{n})}(k)\right)\right|^{p}\\
 & \qquad+C\sum_{k=0}^{\lfloor T/\epsilon_{n}\rfloor}\epsilon_{n}^{1-p}\left[p_{2}^{(\epsilon_{n})}\left(Z^{(\epsilon_{n})}(k)\right)\right]^{p}\left|f\left(X^{(\epsilon_{n})}(k),R\left(X^{(\epsilon_{n})}(k)\right)V^{(\epsilon_{n})}(k)\right)-f\left(Z^{(\epsilon_{n})}(k)\right)\right|^{p}\\
 & \qquad+C\sum_{k=0}^{\lfloor T/\epsilon_{n}\rfloor}\epsilon_{n}^{1-p}\left[p_{3}^{(\epsilon_{n})}\left(Z^{(\epsilon_{n})}(k)\right)\right]^{p}\left|f\left(X^{(\epsilon_{n})}(k),-V^{(\epsilon_{n})}(k)\right)-f\left(Z^{(\epsilon_{n})}(k)\right)\right|^{p}\\
 & =:J_{1}+J_{2}+J_{3}.
\end{align*}
We first estimate 
\begin{align*}
|J_{1}| & \leq C\|\nabla f\|\sum_{k=0}^{\lfloor T/\epsilon_{n}\rfloor}\epsilon_{n}^{1-p}\left|\epsilon_{n}V^{(\epsilon_{n})}(k)\right|^{p}\leq C\|\nabla f\|\sum_{k=0}^{\lfloor T/\epsilon_{n}\rfloor}\epsilon_{n}R^{p}=CT\|\nabla f\|R^{p}.
\end{align*}
Next we treat the second term, where from \eqref{eq:boundonp2} and
\eqref{eq:boundonerror1} we have 
\begin{align*}
J_{2} & \leq2C\|f\|^{p}\sum_{k=0}^{\lfloor T/\epsilon_{n}\rfloor}\epsilon_{n}^{1-p}\left[p_{2}^{(\epsilon_{n})}\left(Z^{(\epsilon_{n})}(k)\right)\right]^{p}\\
 & \leq2C\|f\|^{p}\sum_{k=0}^{\lfloor T/\epsilon_{n}\rfloor}\epsilon_{n}^{1-p}\left[\epsilon_{n}\max\{\langle\nabla U\left(X^{(\epsilon_{n})}(k)\right),V^{(\epsilon_{n})}(k))\rangle,0\}+\mathcal{E}_{1}^{(\epsilon)}\left(X^{(\epsilon_{n})}(k),V^{(\epsilon_{n})}(k)\right)\right]^{p}\\
 & \leq2C\|f\|^{p}R^{p}\sum_{k=0}^{\lfloor T/\epsilon_{n}\rfloor}\epsilon_{n}^{1-p}\left[\epsilon_{n}\left|\nabla U\left(X^{(\epsilon_{n})}(k)\right)\right|+CK\left(X^{(\epsilon_{n})}(k)\right)\epsilon_{n}^{\delta+1}\right]^{p}\\
 & \leq2C\|f\|^{p}\sum_{k=0}^{\lfloor T/\epsilon_{n}\rfloor}\epsilon_{n}^{1-p}\left\{ \epsilon_{n}^{p}\left|\nabla U\left(X^{(\epsilon_{n})}(k)\right)\right|^{p}+\epsilon_{n}^{p(\delta+1)}K\left(X^{(\epsilon_{n})}(k)\right)\right\} .
\end{align*}
Thus, using Holder's inequality we have 
\begin{align*}
\mathbb{E}[J_{2}^{1/p}]^{p}\leq\mathbb{E}[J_{2}] & \leq2C\|f\|^{p}\mathbb{E}\left\{ \sum_{k=0}^{\lfloor T/\epsilon_{n}\rfloor}\epsilon_{n}^{1-p}\left\{ \epsilon^{p}\left|\nabla U\left(X^{(\epsilon_{n})}(k)\right)\right|^{p}+\epsilon^{p(\delta+1)}K\left(X^{(\epsilon_{n})}(k)\right)\right\} \right\} \\
 & \leq C\epsilon_{n}\sum_{k=0}^{\lfloor T/\epsilon_{n}\rfloor}\mathbb{E}\left[\left|\nabla U\left(X^{(\epsilon_{n})}(k)\right)\right|^{p}\right]+C\epsilon_{n}^{1+p\delta}\sum_{k=0}^{\lfloor T/\epsilon_{n}\rfloor}\mathbb{E}\left[K\left(X^{(\epsilon_{n})}(k)\right)\right].
\end{align*}
By stationarity it easily follows that 
\begin{align*}
\mathbb{E}[J_{2}^{1/p}]^{p} & \leq C\epsilon_{n}\frac{T}{\epsilon_{n}}\mathbb{E}\left[\left|\nabla U\left(X^{(\epsilon_{n})}(0)\right)\right|^{p}\right]+C\epsilon_{n}^{1+p\delta}\frac{T}{\epsilon_{n}}\mathbb{E}\left[K\left(X^{(\epsilon_{n})}(0)\right)\right]\leq C,
\end{align*}
uniformly in $n$. To control $J_{3}$, since $p>1$ by subbaditivity
we have
\begin{align*}
\mathbb{E}\left[\left|J_{3}\right|^{1/p}\right] & \leq C\mathbb{E}\left[\left(\sum_{k=0}^{\lfloor T/\epsilon_{n}\rfloor}\epsilon_{n}^{1-p}\left[p_{3}^{(\epsilon_{n})}\left(Z^{(\epsilon_{n})}(k)\right)\right]^{p}\right)^{1/p}\right]\\
 & \leq C\mathbb{E}\left[\sum_{k=0}^{\lfloor T/\epsilon_{n}\rfloor}\left(\epsilon_{n}^{1-p}\left[p_{3}^{(\epsilon_{n})}\left(Z^{(\epsilon_{n})}(k)\right)\right]^{p}\right)^{1/p}\right]\\
 & =C\mathbb{E}\left[\sum_{k=0}^{\lfloor T/\epsilon_{n}\rfloor}\epsilon_{n}^{1/p-1}p_{3}^{(\epsilon_{n})}\left(Z^{(\epsilon_{n})}(k)\right)\right]\\
 & \leq CT\epsilon_{n}^{1/p-2}\mathbb{E}\left[p_{3}^{(\epsilon_{n})}\left(Z^{(\epsilon_{n})}(k)\right)\right].
\end{align*}
Now recall from \eqref{eq:boundonp3-1} and \eqref{eq:p3expansion},
for $n$ large enough so that $\epsilon_{n}R<\varepsilon$, it follows
that
\begin{align*}
\mathbb{E}\left[p_{3}^{(\epsilon_{n})}\left(X^{(\epsilon_{n})}(0),V^{(\epsilon_{n})}(0)\right)\right] & =\int\psi(\mathrm{\mathrm{d}}v)\int\mathrm{d}x\left|\pi\left(x-\epsilon_{n}R\left(x\right)v\right)-\pi\left(x+\epsilon_{n}v\right)\right|\\
 & \leq\int\psi(\mathrm{\mathrm{d}}v)\int\mathrm{d}x\sup_{y\in B(x,\epsilon_{n})}\left\Vert \Delta\pi(y)\right\Vert \epsilon_{n}^{2}.
\end{align*}
By Assumption \ref{assumption:piboundedder} it thus follows that
\begin{align*}
\mathbb{E}\left[\left|J_{3}\right|^{1/p}\right] & \leq CT\epsilon_{n}^{1/p-2}\mathbb{E}\left[p_{3}^{(\epsilon_{n})}\left(Z^{(\epsilon_{n})}(k)\right)\right].\\
 & \leq CT\epsilon_{n}^{1/p}\int\psi(\mathrm{\mathrm{d}}v)\int\mathrm{d}x\sup_{y\in B(x,\varepsilon)}\left\Vert \Delta\pi(y)\right\Vert \to0,
\end{align*}
since $1/p>0$.
\end{document}